\def\ecc{{\cal E}}

\def\Bbb{\mathbb}
\def\CH{{\rm{CH}}}
\def\PH{{\rm{PH}}}

\def\defit{\bf}
\def\interior{{\rm int}}

\def\diam{{\rm{diam}}}

\def\Bbb{\mathbb}
\def\reals{\Bbb R}
\def\disk{\Bbb D}
\def\circle{\Bbb T}
\def\complex{\Bbb C}

\def\integers{\Bbb Z}

\def\cal{\mathcal}

\documentclass[12pt]{amsart}  
\usepackage{amssymb}    

\usepackage{graphicx,amssymb,amsthm,verbatim,amsfonts}

\theoremstyle{plain}                    
\newtheorem{thm}{Theorem}[section]
\newtheorem{cor}[thm]{Corollary}
\newtheorem{lemma}[thm]{Lemma}

\newcounter{ques}
   
\numberwithin{equation}{section}

\begin{document}
\baselineskip=18pt


%

\title [Quadrilateral meshes for  PSLGs]
          {Quadrilateral meshes for PSLGs}

\author {Christopher J. Bishop}
\address{C.J. Bishop\\
         Mathematics Department\\
         SUNY at Stony Brook \\
         Stony Brook, NY 11794-3651}
\email {bishop@math.sunysb.edu}
\thanks{The  author is partially supported by NSF Grant DMS 13-05233.
       }

\date{January 2011; revised December 2014; revised November 2015;}

\maketitle


\begin{abstract}
We prove that every planar straight line graph with $n$ 
vertices  has a conforming 
quadrilateral mesh with $O(n^2)$ elements,   all angles $\leq 120^\circ$ 
and all new angles $\geq 60^\circ$.
Both the complexity and the angle bounds are sharp.
\end{abstract}

\clearpage


\setcounter{page}{1}
\renewcommand{\thepage}{\arabic{page}}
\section{Introduction} \label{Intro} 

The purpose of this paper is to prove:

\begin{thm} \label{Quad Mesh} 
Suppose  $\Gamma$ is a   planar straight line 
graph  with $n$ vertices.  Then $\Gamma$ 
 has a  conforming  quadrilateral mesh   with 
$O(n^2) $ elements, all angles $\leq 120^\circ$ and 
all new  angles $ \geq 60^\circ$.
\end{thm}

The precise definitions of all the terminology will be given 
in the next few sections, but briefly, this means that
each face of $\Gamma$ (each bounded complementary component)
can be meshed with quadrilaterals 
 so that the meshes are consistent across the
edges of $\Gamma$ and all angles are in the interval $[60^\circ, 120^\circ]$
except when forced to be smaller by two edges of $\Gamma$
that meet at an angle $< 60^\circ$.

Bern and Eppstein showed in \cite{BE-2000} that 
any simple polygon $P$ has a linear 
sized  quadrilateral mesh 
with all angles $\leq 120^\circ$.  They also used 
Euler's formula to  prove 
that any quadrilateral mesh of a regular hexagon must 
contain an angle of measure   $\geq 120^\circ$. Thus 
the upper angle bound is sharp.
Moreover, if a polygon contains 
an angle of measure $  120^\circ + \epsilon$, $\epsilon >0$,
 then this angle  must be 
subdivided in the mesh (in order to achieve the upper bound),
 giving at least one new angle $  \leq  60^\circ
+ \epsilon/2$. Thus the $60^\circ$  lower bound is also sharp.
 In \cite{Bishop-optimal} I showed that every simple polygon has a
linear sized  quadrilateral 
mesh with all  angles $\leq 120^\circ$
 and all new angles  $ \geq 60^\circ$.
Theorem \ref{Quad Mesh}  extends this result to planar straight
line graphs (PSLGs).
The complexity bound increases from $O(n)$ to $O(n^2)$, but this 
is necessary: 
 see Figure \ref{Gamma3}.

\begin{figure}[htb]
\centerline{
\includegraphics[height=2.5in]{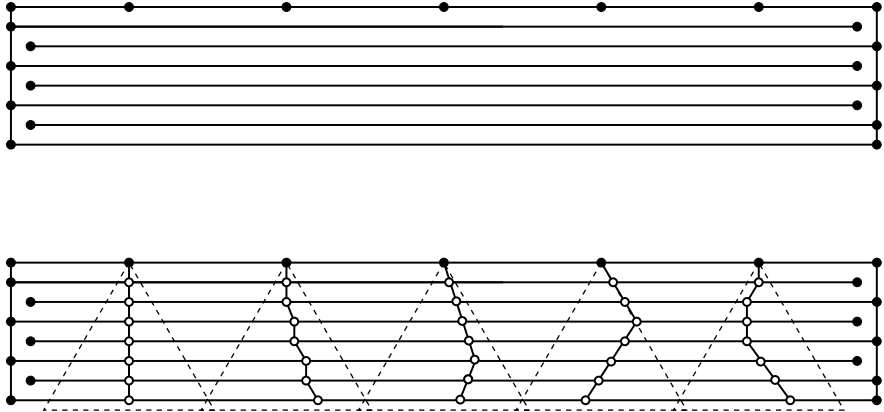}
 }
\caption{\label{Gamma3}
Consider a vertex $v$ along the top edge of 
the illustrated PSLG.
Any mesh of this PSLG with all angles
$ \leq 120^\circ$ must insert a new edge at $v$, creating
a new vertex (white) on the edge below.  This repeats
until we have a  path that  reaches the bottom edge.
If there are $n+1$ horizontal lines and we
place $n$  widely spaced vertices
on the top edge, then at least $n^2$  mesh vertices must be created.
A similar argument works if $120^\circ$ is replaced
by any bound $ < 180^\circ$.
}
\end{figure}

To save space, we  will say that a quadrilateral 
is  {\defit $\theta$-nice}  if all four interior angles are between 
$90^\circ - \theta $ and $90^\circ + \theta$ (inclusive). 
When $\theta=30^\circ$ we shorten this to saying
 the quadrilateral is {\defit nice}. 
In this paper, we mostly deal with convex quadrilaterals, 
so we will always take $0 < \theta < 90^\circ$ in this definition.
A quadrilateral mesh of a simple 
polygon is nice if all the quadrilaterals are nice.
A conforming  quadrilateral mesh 
of a PSLG will be called  nice if all the angles are 
between $60^\circ$ and  $120^\circ$, except for 
smaller angles forced by angles in the PSLG. Thus Theorem 
\ref{Quad Mesh} says that every PSLG with $n$ vertices 
 has a nice conforming quadrilateral mesh with at most $O(n^2)$ elements.

The quad-meshing result for simple polygons given 
in \cite{Bishop-optimal} is one of the main ingredients
in the proof of Theorem \ref{Quad Mesh}. We will start 
by adding vertices and  edges to the PSLG so that all of the faces 
become simple polygons. We then quad-mesh a small neighborhood
of each vertex ``by hand'' using a construction we call 
a protecting sink (Lemma \ref{sector conform sink}); 
the mesh elements near the vertex
will never be changed at later steps of the construction. 
The ``unprotected'' region is divided into  simple polygons 
with all interior angles $\geq 90^\circ$.
 We then apply the result of 
\cite{Bishop-optimal} to give a nice quad-mesh of each 
of these simple polygons.
However, these meshes might not be consistent 
across the edges of the PSLG. If the mesh elements have
bounded eccentricity (the eccentricity $\ecc(Q)$ 
 of a quadrilateral  $Q$ 
is the length of the longest side divided by the length 
of the shortest side), then we can use a device called 
``sinks'' (described below) to merge the meshes of different
faces into a mesh of the whole PSLG. 

 We say that a simple polygon $P$ is a
{\defit sink} if whenever we add an even number of vertices to the 
edges of $P$ to form a new polygon $P'$, then the 
interior of  $P'$ has a nice
quadrilateral  mesh   so that 
 the only mesh vertices on $P$ 
are the vertices of $P'$ (we say such a mesh extends $P'$).  
It is not obvious that sinks exist, but we shall 
show (Lemma \ref{quad sinks lemma 1})  that any nice quadrilateral  $Q$ 
can be made into a sink $P$ by adding 
$N=O(\ecc(Q))$ vertices to the edges of $Q$.
If we add $M$ extra points to the boundary of a 
sink, then nicely re-meshing  the sink to account 
for the new vertices will use $O(N M^2)$
quadrilaterals in general, but only $O(N M)$ 
quadrilaterals in the  important special case when 
only add the extra points to a single side of 
the quadrilateral (or to a single  pair of opposite sides). 

The name sink comes from a sink in a directed graph, i.e., 
a vertex with zero out-degree. Paths that enter 
a sink can't leave. In our construction, 
points will be propagated through a  quad-mesh (this will 
be precisely defined in Section \ref{propagation})
and propagation paths continue
until they  hit the boundary of the mesh or
until they 
run ``head on'' into another propagation path.
Sinks allow us to force the latter to happen. When 
an even number of propagation paths hit the boundary 
of a sink, we can re-mesh the interior of the sink 
so that  these paths hits vertices of the mesh
 and terminate. Thus sinks ``absorb'' propagation 
paths. Since our complexity bounds depend on 
terminating propagation paths quickly, sinks are 
a big help.

Sinks  can also be used to merge two or more 
quadrilateral meshes  that are defined 
on disjoint regions that have 
overlapping boundaries.
As a simple example of how this works, consider 
a PSLG $\Gamma$ which has several faces, $\{\Omega_k\}$,
so that $\Omega_k$ can be nicely meshed using $N_k$  quadrilaterals
 with maximum eccentricity $M < \infty$.  
Use Lemma \ref{quad sinks lemma 1} to add $O(M)$ 
vertices to the sides of each quadrilateral 
in every face of $\Gamma$, in order to make every 
quadrilateral into a sink. This 
requires $O(M\sum_k N_k)$ new vertices. Every quadrilateral 
is now a sink with at most $O(M)$
 extra vertices on its boundary
(due to the sink vertices added to its neighbors). 
We make sure that the number of extra vertices for 
each quadrilateral is even by cutting every edge in half
and adding the midpoints. By the definition of sink
we can now nicely re-mesh every quadrilateral consistently 
with all its neighbors, obtaining a mesh of $\Gamma$, i.e., 
assuming Lemma \ref{quad sinks lemma 1}, we have 
proven:

\begin{lemma} \label{simple merge}
Suppose $\Gamma$ is a PSLG,  and that  every face  of $\Gamma$ 
is a simple polygon with  a nice 
  quadrilateral mesh. Suppose   a total of $N$ elements
are used in these meshes, and every quadrilateral
has  eccentricity bounded by $M$. Then $\Gamma$ 
has a nice mesh using $O(NM^2)$ quadrilaterals.
\end{lemma}

Unfortunately,  the  proof of Theorem \ref{Quad Mesh} 
is not quite as simple as this. When we use the 
result  for quad-meshing a simple polygon 
from \cite{Bishop-optimal}, the method  will sometimes  produce 
quadrilaterals with  very large eccentricity, so the
merging argument above does not give a uniform bound.
However, the proof in \cite{Bishop-optimal} 
shows that these high eccentricity quadrilaterals
 have very special shapes 
and structure that allow us to  use a result from 
\cite{Bishop-nonobtuse} to   nicely quad-mesh the union 
of these pieces.  We then use sinks to merge this 
mesh with a nice mesh on the union of the low eccentricity pieces.
Thus the proof of Theorem \ref{Quad Mesh} rests mainly on 
four ideas: 
\newline
(1) adding edges to $\Gamma$ to reduce to the case when 
 every face of $\Gamma$ is a simple polygon, 
\newline
(2)  the linear quad-meshing algorithm
 for simple polygons from \cite{Bishop-optimal}, 
\newline
(3)  a quad-meshing result from \cite{Bishop-nonobtuse}
 for   regions with special dissections, and 
\newline
(4) the construction of  sinks  (this
takes up the bulk of the current paper).

Section \ref{PSLG defn} will review the definitions
of meshes and dissections and record some 
basic facts.
In Section \ref{connecting} we 
show how to reduce Theorem \ref{Quad Mesh}
to the case when the PSLG is 
connected and  every face  is a simple polygon. 
In Section \ref{propagation} we discuss some properties 
of quadrilateral meshes and, in particular, the 
idea of propagating a point through a quadrilateral 
mesh. 
Section \ref{sink defn} gives the definition of 
a sink and states various results about sinks that 
are proven in Sections
 \ref{inscribed polygon}-\ref{conforming sinks}. 
Section \ref{dissection} defines a dissection
by nice isosceles trapezoids and quotes a result from 
\cite{Bishop-nonobtuse} that a domain with such 
a dissection has a nice quadrilateral mesh. 
Section \ref{Thick Thin sec} will review the thick/thin 
decomposition of a simple polygon and quote 
the precise result from \cite{Bishop-optimal} 
that we will need. In particular, we will 
see that the union of ``high eccentricity'' 
quadrilaterals produced by the algorithm in 
\cite{Bishop-optimal} has the kind of dissection 
needed to apply the result from \cite{Bishop-nonobtuse}. 
In Section \ref{Proof} we will give the proof 
of Theorem \ref{Quad Mesh} using all the tools 
assembled earlier.
Actually, we will prove 
a  slightly stronger version of Theorem \ref{Quad Mesh}:
for any $\theta>0$, we can construct  a nice 
conforming mesh that  has $O(n^2/\theta^2)$ elements, and 
all but $O(n/\theta^2)$ of them are $\theta$-nice.
Thus when $\theta$ is small, 
``most'' pieces are close to rectangles.

I thank Joe Mitchell and Estie Arkin for numerous helpful 
conversations about computational geometry in general, and
about the results of this paper in particular. Also thanks to two 
anonymous referees whose thoughtful remarks and suggestions 
on two versions of the paper  greatly
improved the precision and  clarity of the  exposition.
A proof suggested by one of the referees is included in 
Section \ref{rectangular sinks}.

\section{Planar straight line graphs}  \label{PSLG defn} 
A  {\defit planar straight line graph} $\Gamma$
 (or {\defit PSLG} from now on) is a compact subset of the plane
$\reals^2$,  together with a finite set $V(\Gamma)  \subset \Gamma$ (called
the vertices of $\Gamma$)  
such that  $ E= \Gamma \setminus V$   is  a 
finite  union of disjoint, bounded,
open line segments (called the edges of $\Gamma$). 
Throughout the paper we will let $n = n(\Gamma) $ denote the 
number of vertices of $\Gamma$ and $m = m(\Gamma)$ the 
number of edges. 
The vertex set $V$  includes  both endpoints of  every 
 edge, and may include other points as well (i.e., 
isolated points of the PSLG). 
Note that  the vertex set of a PSLG is not 
uniquely determined. However, every PSLG has a minimal 
vertex set and every other vertex set is obtained
from this one by adding extra points along the edges.
It would be reasonable to define a PSLG as the pair 
$\Gamma=(V, E)$ and to let $|\Gamma|$ denote the compact planar 
set which is the union of these sets, but I have chosen 
to let $\Gamma$ denote this set, since I will most often 
be treating a PSLG as a planar set, rather than as 
a combinatorial object.


We let $\CH (\Gamma)$ denote the {\defit closed convex hull} of 
$\Gamma$, i.e., the intersection of all closed half-planes 
containing $\Gamma$. Let   $\partial\CH(\Gamma)$
denote the boundary of 
the convex hull and let  $\interior(\CH(\Gamma)) 
= \CH(\Gamma) \setminus \partial \CH(\Gamma)$  be the 
interior of the  convex hull. We say $\Gamma$ is non-degenerate
if  $\interior(\CH(\Gamma))$ is non-empty, i.e., $\Gamma$ is not contained
in a single line. 

A {\defit face} of $\Gamma $ is any of the bounded, open 
 connected components of $ \reals^2 \setminus \Gamma$.
Every PSLG  has a unique unbounded  complementary component
that we sometimes call the {\defit unbounded face}, 
but a PSLG  may or may not have faces.
We will say that a bounded, connected open set $\Omega$ is 
a {\defit polygonal domain} if it is the face of some PSLG 
(informally,   $\partial \Omega$ is a finite union of 
points and line segments). 

The {\defit polynomial hull} of a PSLG $\Gamma$ is the 
compact planar set that is the union 
of $\Gamma$ and all of its bounded faces. This 
set  is denoted 
$\PH(\Gamma)$. The name comes from complex analysis,
where the polynomial hull of a compact set $K$ 
is defined as 
$$ \PH(K) = \{ z \in \complex: |p(z)| \leq 
\sup_{w \in K} |p(w)| \text{ for all polynomials } p \}.$$
This agrees with our definition in the case
$K$ is a  PSLG. See Figure \ref{Hulls}.

\begin{figure}[htbp]
\centerline{ 
	\includegraphics[height=1.4in]{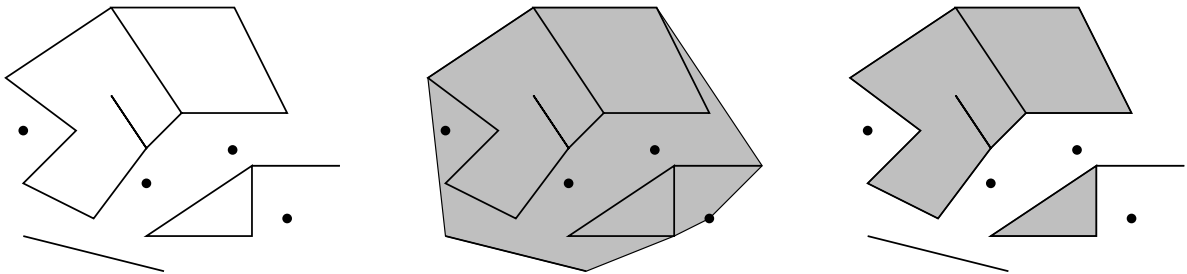}
}
\caption{ \label{Hulls}
A PSLG, its convex hull and its polynomial hull.  }
\end{figure}

A  {\defit polygon} or {\defit polygonal curve}
 is a sequence of vertices  $z_1, \dots z_n$ and 
 open  edges  $(z_1, z_2), \dots, (z_n, z_1)$.
A {\defit polygonal path} or {\defit arc} is a similar 
list of vertices, but with edges $(z_1, z_2), \dots 
(z_{n-1}, z_n)$; the last is not connected back to the first.
 A polygon 
 is {\defit simple} if the vertices are all distinct and the edges 
are pairwise disjoint.  A polygon is called {\defit 
edge-simple} if the (open)  edges are all pairwise disjoint, 
but vertices may be repeated.

 The Jordan curve theorem implies that 
a simple closed polygon  has two distinct complementary  connected 
components, exactly one of which is bounded.   This is 
the  interior (or face) of the simple polygon. 
A domain (i.e., an open, connected set)  that is 
the interior of some simple polygon
will be called a {\defit simple  polygonal domain}.
If $\Omega$ is multiply connected, but every  connected 
component of $\partial \Omega$ is a simple polygon, we 
say $\Omega$ is a {\defit simple polygonal domain with holes}.
See Figure \ref{DomainDefns} for some examples.

\begin{figure}[htb]
\centerline{
\includegraphics[height=3.0in]{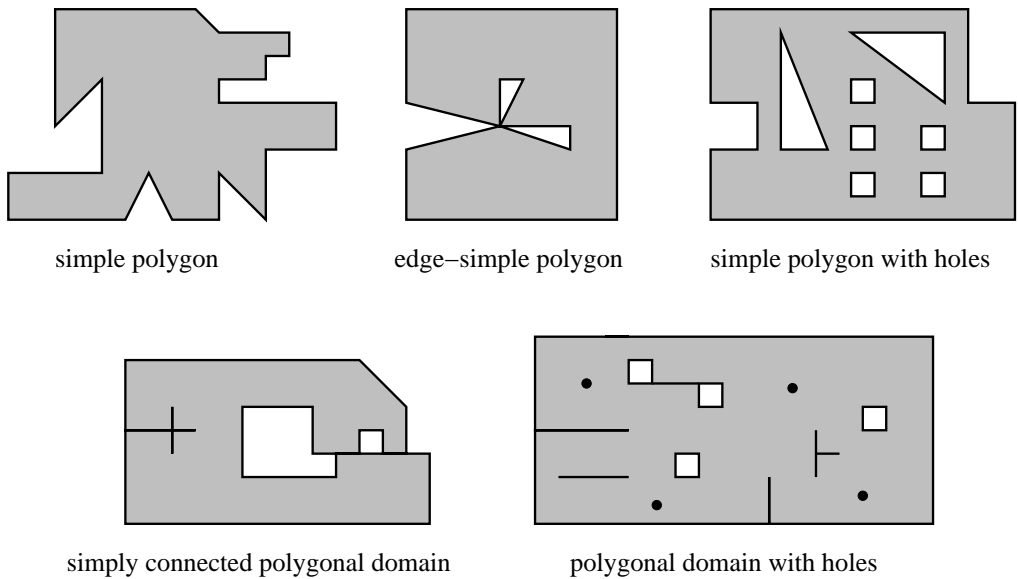}
 }
\caption{\label{DomainDefns}
Examples of the different types of polygonal domains.
}
\end{figure}

A {\defit triangle} is a simple polygon with three vertices 
(hence three edges). We say a simple polygon $P$  has 
a {\defit triangular shape} if  there is a triangle $T$ 
so that $P$ is obtained by adding vertices to the 
edges of $T$. See Figure \ref{Shapes}.
Similarly, a {\defit quadrilateral} $Q$ is a simple polygon with 
four vertices. We say a simple polygon $P$ has a {\defit 
quadrilateral shape} (or is {\defit quad-shaped}), if  $P$ is obtained 
by extra  adding vertices to the edges of a quadrilateral $Q$. 
 The  four 
vertices of $Q$ will be called the {\defit corners} of $P$ 
(they are the only vertices of $P$ where 
the interior angle is not $180^\circ$).  The other 
vertices of $P$ will be called {\defit interior edge vertices}.

\begin{figure}[htbp]
\centerline{
	\includegraphics[height=1.0in]{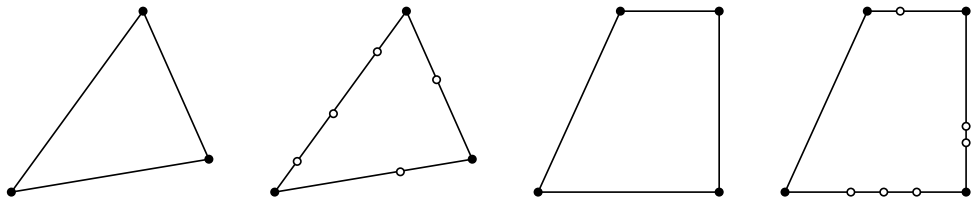}
}
\caption{ \label{Shapes}
A triangle, a triangular shaped octagon, a 
quadrilateral and a quad-shaped decagon. 
The black dots are the corners and the white 
dots are the interior edge vertices.
}
\end{figure}


A {\defit refinement} (also called a {\defit sub-division}) 
 of a PSLG  $\Gamma$ is a PSLG $\Gamma' $
so that $V(\Gamma) \subset V(\Gamma')$ and $\Gamma \subset \Gamma'$.
Informally, $\Gamma'$ is obtained from $\Gamma$ by 
adding new vertices and edges and by subdividing 
existing edges. 
A {\defit mesh} of $\Gamma$ is a
sub-division  $\Gamma'$  of $\Gamma$  such that 
$\Gamma' \subset \PH(\Gamma)$ and 
every face of $\Gamma'$ is a simple polygonal domain.
   Note that we allow the addition of 
new vertices (called {\defit Steiner points}) when we mesh a PSLG.

A mesh $\Gamma'$ is called a {\defit triangulation} if every face
of $\Gamma'$  is a triangle
 and is called a  {\defit quadrilateral mesh}  or 
{\defit quad-mesh} if every face is a quadrilateral.
We will only consider meshes by convex quadrilaterals
in this paper.
It is always possible to triangulate a PSLG without adding 
Steiner points, but this is not the case for quadrilateral meshes.
Sometimes we wish the mesh of a PSLG to cover the convex 
hull of the PSLG. In this case, we should add the boundary 
of the convex hull to the PSLG and mesh this new PSLG.

 A  {\defit  quadrilateral dissection} 
is a mesh in which every face is  a quad-shaped
polygon. (Similarly, a triangular dissection is 
a mesh where every face is triangular shaped, 
but we won't use these in this paper.)
More informally, a quadrilateral 
dissection is like a quadrilateral mesh, except 
that quadrilaterals whose boundaries intersect,  do not have 
to intersect at just points or full edges; 
two edges can overlap without being equal and 
the corner of one piece can be an interior edge 
vertex  of another piece.
A vertex where this happens is called 
a {\defit non-conforming vertex}.
 A    quadrilateral dissection is also 
called a 
{\defit non-conforming quadrilateral mesh}.
See Figure \ref{MeshDefns}. 

\begin{figure}[htb]
\centerline{
\includegraphics[height=3.5in]{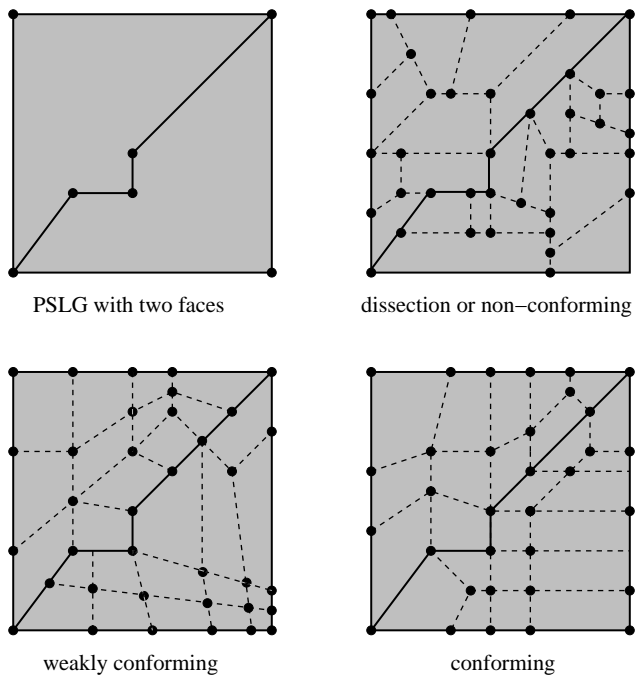}
 }
\caption{\label{MeshDefns}
On the upper left is a  PSLG with two simple polygonal 
faces, 
followed by a quadrilateral dissection, a weakly  conforming
quadrilateral mesh 
and a fully conforming quadrilateral  mesh.
}
\end{figure}

If every face of $\Gamma$ is a simple polygon, 
then a {\defit weak quadrilateral mesh}  of $\Gamma$ is a 
 quadrilateral mesh of each face, 
without the requirement that the meshes match up across 
the edges of $\Gamma$.  One of the main goals of this 
paper is to give a method of converting a weak mesh of 
a PSLG into a true mesh of similar size.

\section{Connecting $\Gamma$ with no small angles}
 \label{connecting}
We reduce Theorem \ref{Quad Mesh}
to the case when every face of $\Gamma$ is a 
simple polygon.

\begin{lemma} \label{simple poly lemma}
If $\Gamma$ is a PSLG with $n$ vertices such 
that $\PH(\Gamma)$ is connected, then 
by adding at most $O(n)$ new edges and 
vertices,  we can find a    connected refinement $\Gamma'$   
of $\Gamma$ so 
that every face of $\Gamma'$ is a simple polygon and
any angles  less than $60^\circ$ 
were already  angles in a face of $\Gamma$.
\end{lemma}

\begin{proof}
We use an idea of Bern, Mitchell and Ruppert  \cite{BMR95} (refined by
 David Eppstein in \cite{Eppstein-faster}) of adding
disks that connect different components of $\Gamma$.
  For each component $\gamma$ of $\Gamma$, except the  
component $\gamma_u$ (``u'' for {\bf u}nbounded) 
 bounding 
the unbounded complementary component of $\Gamma$, choose
a left-most vertex $v$  of $\gamma$ and consider the left
half-plane defined by the vertical line through this point,
e.g., the vertical line on the left side of
 Figure \ref{Connect1}. Consider the family
of  open disks in the  left half-plane  tangent  to this line at  $v$,
and take the maximal open  disk   that does not intersect $\Gamma$.
Its boundary must hit one or more components of $\Gamma$ that
are distinct from $\gamma$. Choose one point on the circle
from each distinct component (see the right side of
Figure \ref{Connect1}); a previously chosen disk touching 
some component counts as part of that component.
  Do this for each component of
$\Gamma$ other than $\gamma_u$,
 taking the maximal open disk that is disjoint
from $\Gamma$ and all the previously constructed disks.

\begin{figure}[htbp]
   \centerline{ 
	\includegraphics[height=2in]{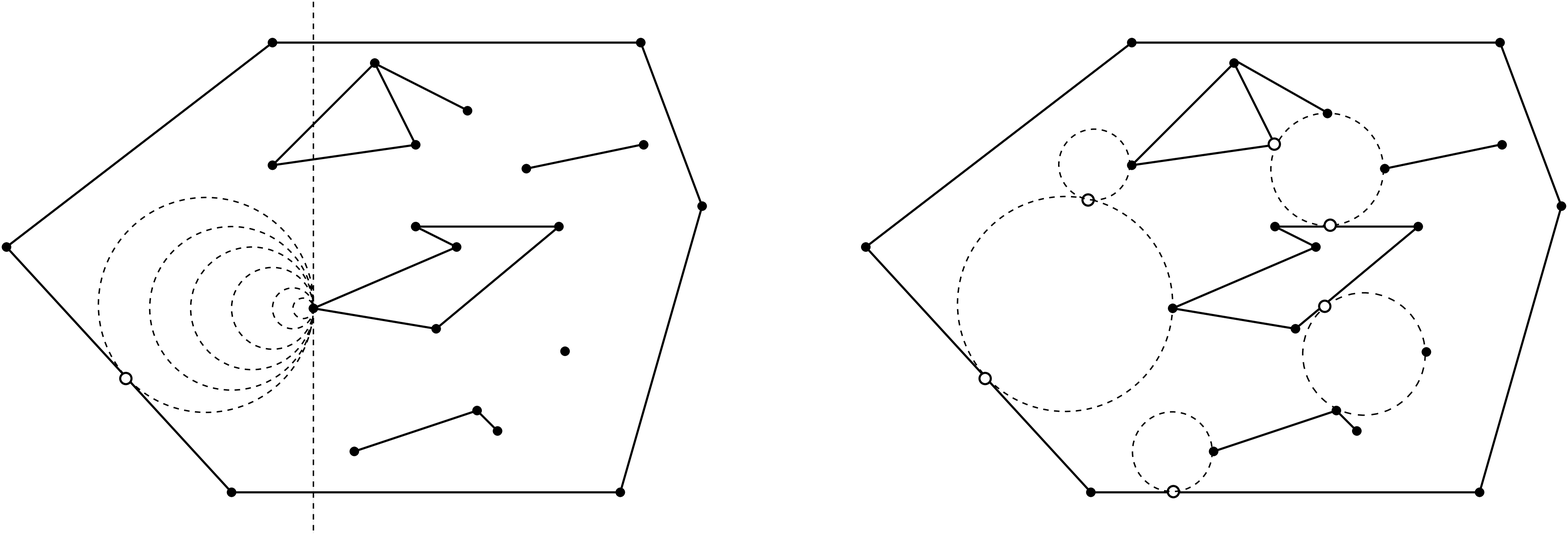}
	}
\caption{ \label{Connect1}
For each connected component of $\Gamma$ (except the component
bounding the unbounded complementary component), choose the leftmost
point and expand a disk until it contacts another component or
a previously constructed disk.
The process for one component is shown on the left; the result 
for all the components is shown on the right.
}
\end{figure}

After all the disks have been placed, we connect the
chosen points on the boundary of each disk by a
PSLG inside the disk that has no angles $< 60^\circ$,
e.g., as illustrated in Figure \ref{JoinHex}.
If the points are widely spaced, we can simply join them
all to the origin (left side of Figure \ref{JoinHex}),
 assuming this does not form an angle
$< 60^\circ$. Otherwise we place a regular hexagon around
the orgin and connect the points on the circle to
the hexagon by radial segments (right side of Figure
\ref{JoinHex}).

\begin{figure}[htbp]
\centerline{ 
	\includegraphics[height=1.5in]{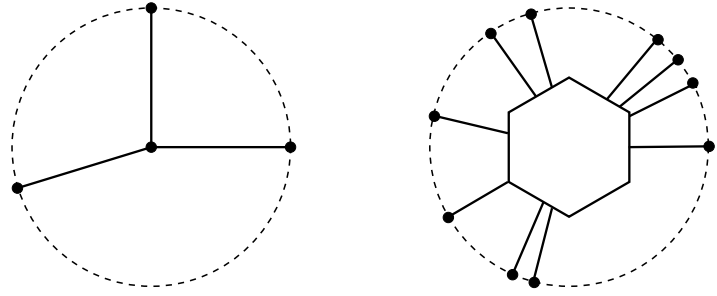}
	}
\caption{ \label{JoinHex}
Joining points on a circle without creating angles $< 60^\circ$.
}
\end{figure}

\begin{figure}[htbp]
\centerline{ 
	\includegraphics[height=2in]{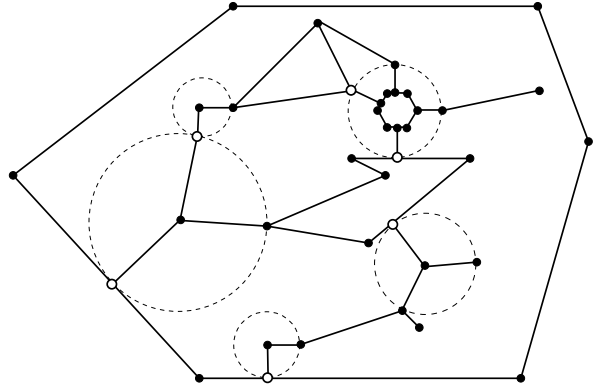}
	}
\caption{ \label{Connect5}
The PSLG in Figure \ref{Connect1} with the disks replaced by the 
connecting PSLGs from Figure \ref{JoinHex}.
One face is not a simple polygon.
}
\end{figure}

We have now replaced $\Gamma$ by another PSLG  
 $\Gamma'$ that is connected and contains no
angles $< 60^\circ$, except for those that were already
in $\Gamma$. However, the faces of $\Gamma'$ need not be
simple polygons. We will add more disks to get this property.

Each face  $\Omega$ of $\Gamma'$ is either
a simple polygon or $\partial \Omega$ has an edge or vertex
whose removal disconnects $\partial \Omega$. If there is an
edge  $e$  so that $\partial \Omega \setminus e$ has two
components, then  one of them, $\gamma_o$ (``o'' for {\bf o}uter)
 separates
the other, $\gamma_i$ (``i'' for {\bf i}nner) from $\infty$. Let
$p$ be the endpoint of $e$ that meets $\gamma_i$; after
a rotation and translation we can assume $p=0$ and $e$ lies
on the negative real axis. Choose a vertex $v = a+ib$ on $\gamma_i$
that is farthest to the right, and consider disks in the half-plane
 $\{ x+iy : x > a\}$   that are tangent to the vertical line 
 $\{x+iy: x=a\}$
at $v$. There is a maximal such open disk $D$ contained in $\Omega$
and its boundary must intersect $\gamma_o$ or a previously
generated disk.
We repeat the process until there are no separating edges left and
then connect points in the disks as above. The procedure is 
repeated at most once for each edge in the the boundary of the face, 
hence at most $O(n)$  edges and vertices are added.

\begin{figure}[htbp]
\centerline{ 
	\includegraphics[height=1.75in]{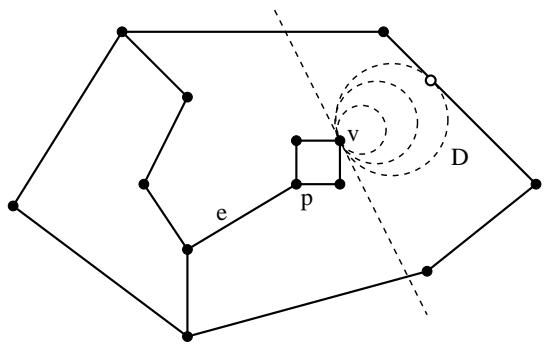}
	}
\caption{ \label{Connect3}
If $\partial \Omega$ is not an edge-simple polygon then there is at
least one edge  $e$ that divides $\partial \Omega$ into
``outer'' and  ``inner'' parts. The inner component can be
connected to the outer component by adding a disk as
described in the text. 
}
\end{figure}

\begin{figure}[htbp]
\centerline{ 
	\includegraphics[height=1.75in]{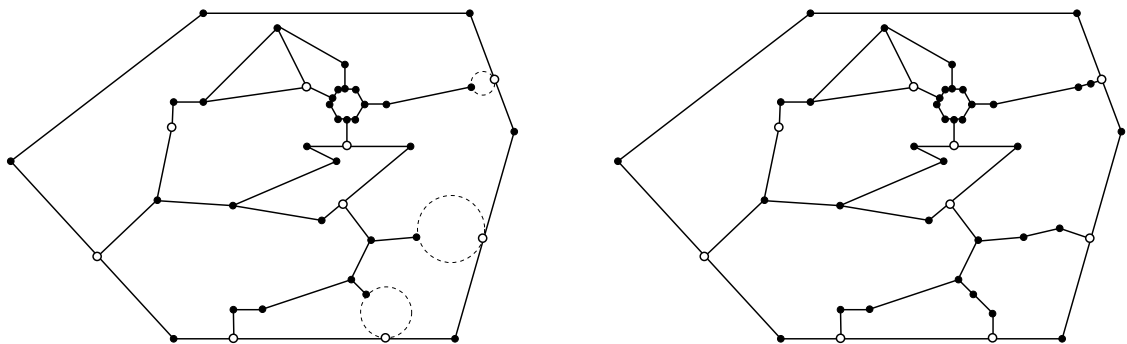}
	}
\caption{ \label{Connect6}
The PSLG in Figures \ref{Connect1} and \ref{Connect5} 
with the disks that prevent ``two-sided'' edges on the 
left and these disks replaced by PSLGs on the right.
The resulting polygon has faces that are simple polygons.
}
\end{figure}

We now have a PSLG so that all the faces are edge-simple. 
To make the faces simple, we  choose a small circle around 
each repeated vertex and add polygonal arcs inscribed in 
arcs of these circles as shown in Figure \ref{CutVertex}.
This is 
easy to do and the details are left to the reader. 
\end{proof}

\begin{figure}[htbp]
\centerline{ 
	\includegraphics[height=1.5in]{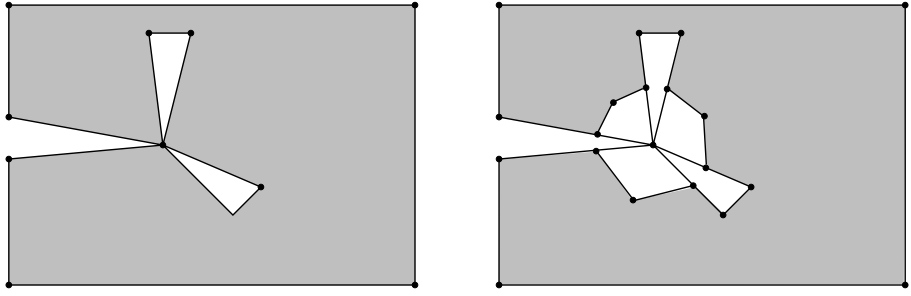}
	}
\caption{ \label{CutVertex}
If there is a vertex that disconnects
$\partial \Omega$, then we place cross-cuts around this
vertex that block access from $\Omega$. These lie
on some sufficiently small circle around the vertex and use
at most $O(\deg(v))$ new edges and vertices. Summing
over all vertices gives $O(n)$.
}
\end{figure}

\section{Propagation in quadrilateral meshes}  \label{propagation}

In this section we review a  few  helpful properties
of quadrilateral meshes.

\begin{lemma} \label{need even}
If $\Omega$ is a simply connected  polygonal domain  that has a
quadrilateral mesh, then the number of  boundary vertices
must be even.
\end{lemma}

\begin{proof}
 Let $Q$ be the number of quadrilaterals in the mesh, $I$ the number of
interior vertices, $B$ the number of boundary vertices and
$E$ the number of edges. Note that the number of edges on the boundary
also equals $B$. Each quadrilateral has four edges, and each interior edge
is counted twice,  so $4Q = 2(E-B) + B = 2 E - B .$
Thus $B = 2(E -2Q)$ is even.
\end{proof}

 Given a convex  quadrilateral $Q$
 with vertices ${a,b,c,d}$ (in the counterclockwise direction)
 and a point $x= ta +(1-t)b$, $0<t<1$,  on the edge
$[a,b]$, use a line segment to connect $x$ to
 $y= t d + (1-t)c$ on the opposite edge
of the  quadrilateral. We will call this a propagation segment.
By  replacing $x$ with $y$ and repeating the construction,
we create a path through the mesh 
 that can be continued until it either hits a boundary edge
or returns to the original starting point $x$. If $x$ is
a boundary point, the latter
 is impossible, so the path must terminate at a distinct
boundary point. Applying this to every midpoint of a boundary edge shows
that each such edge is paired with a distinct edge,  giving an
alternate proof of Lemma \ref{need even}.

We will repeatedly  use propagation lines to subdivide a
quadrilateral mesh, and so a basic fact we need is that
this process  preserves ``niceness''.

\begin{lemma} [Lemma 4.1, \cite{Bishop-nonobtuse}]
\label{split quad}
Suppose $Q$ is a $\theta$-nice  quadrilateral.
If $Q$ is  sub-divided by a  propagation segment, then each
of the resulting sub-quadrilaterals is also $\theta$-nice.
\end{lemma}

\begin{cor} \label{get even}
Suppose $\Omega$ is polygonal domain
 and every component of $\partial
\Omega$ is a simple polygon. Then every nice  quadrilateral mesh of
$\Omega$ has a  nice subdivision with
exactly twice as many vertices on each component of $\partial \Omega$.
\end{cor}

\begin{proof}
Split each quadrilateral  by two segments joining
the midpoints of opposite sides. Each boundary edge is split in
two so the number of boundary edges  on each component doubles.
\end{proof}

Alternatively, one can just split each boundary edge into two,
and propagate these vertices until they hit another boundary
midpoint. See Figure \ref{Double} for an example of both
types of subdivision.

\begin{figure}[htbp]
\centerline{
 \includegraphics[height=1.5in]{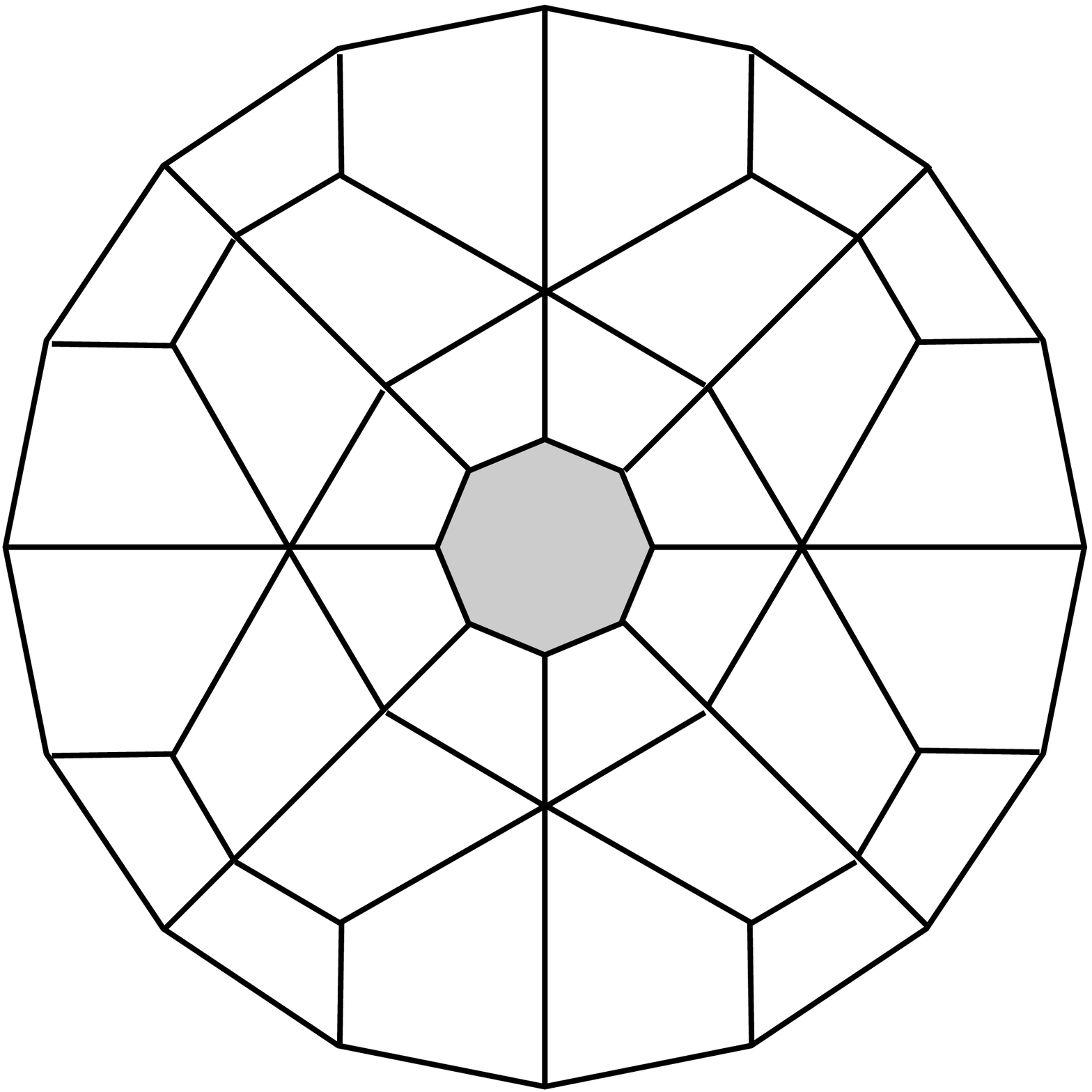}
$\hphantom{xxxx} $
 \includegraphics[height=1.5in]{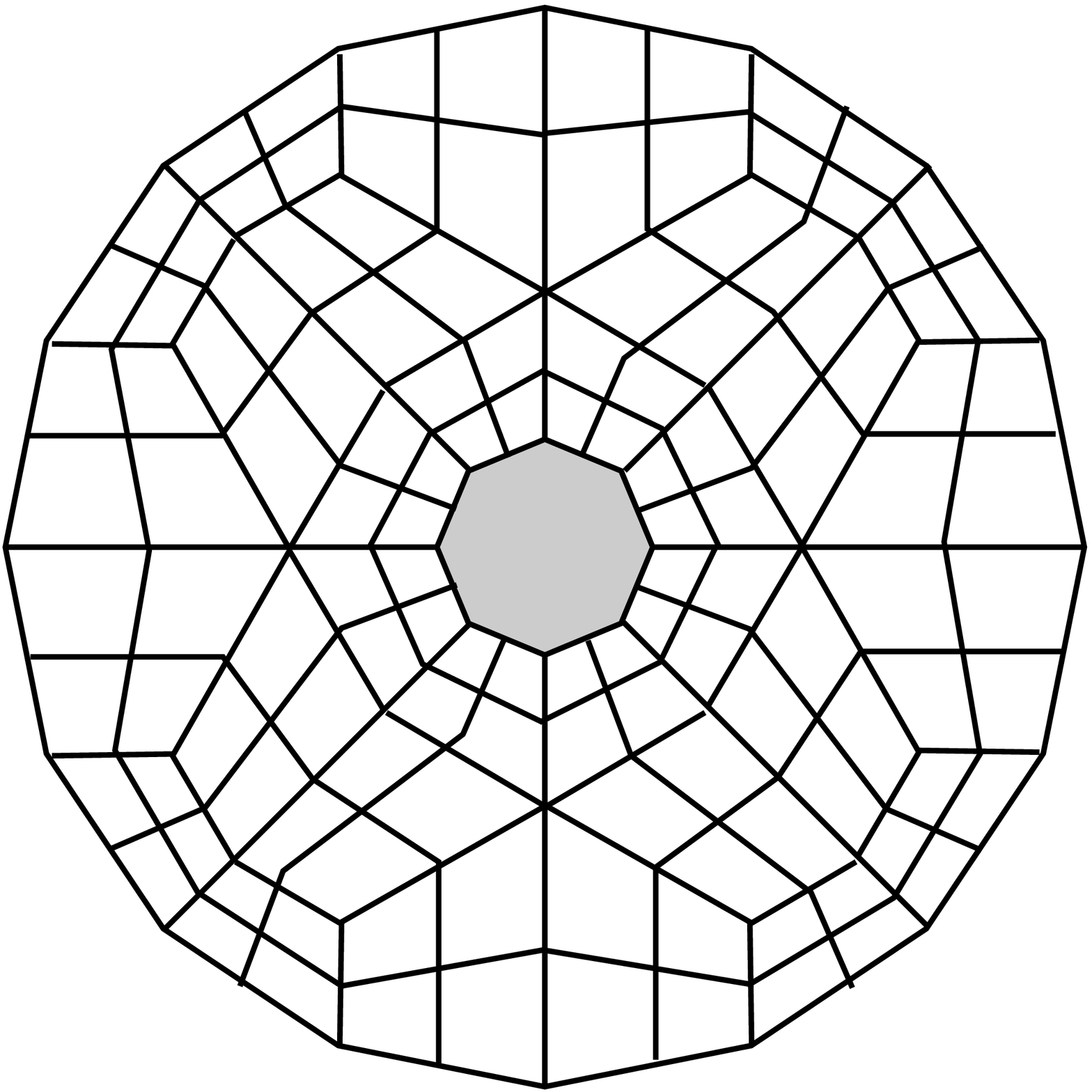}
$\hphantom{xxxx}$
 \includegraphics[height=1.5in]{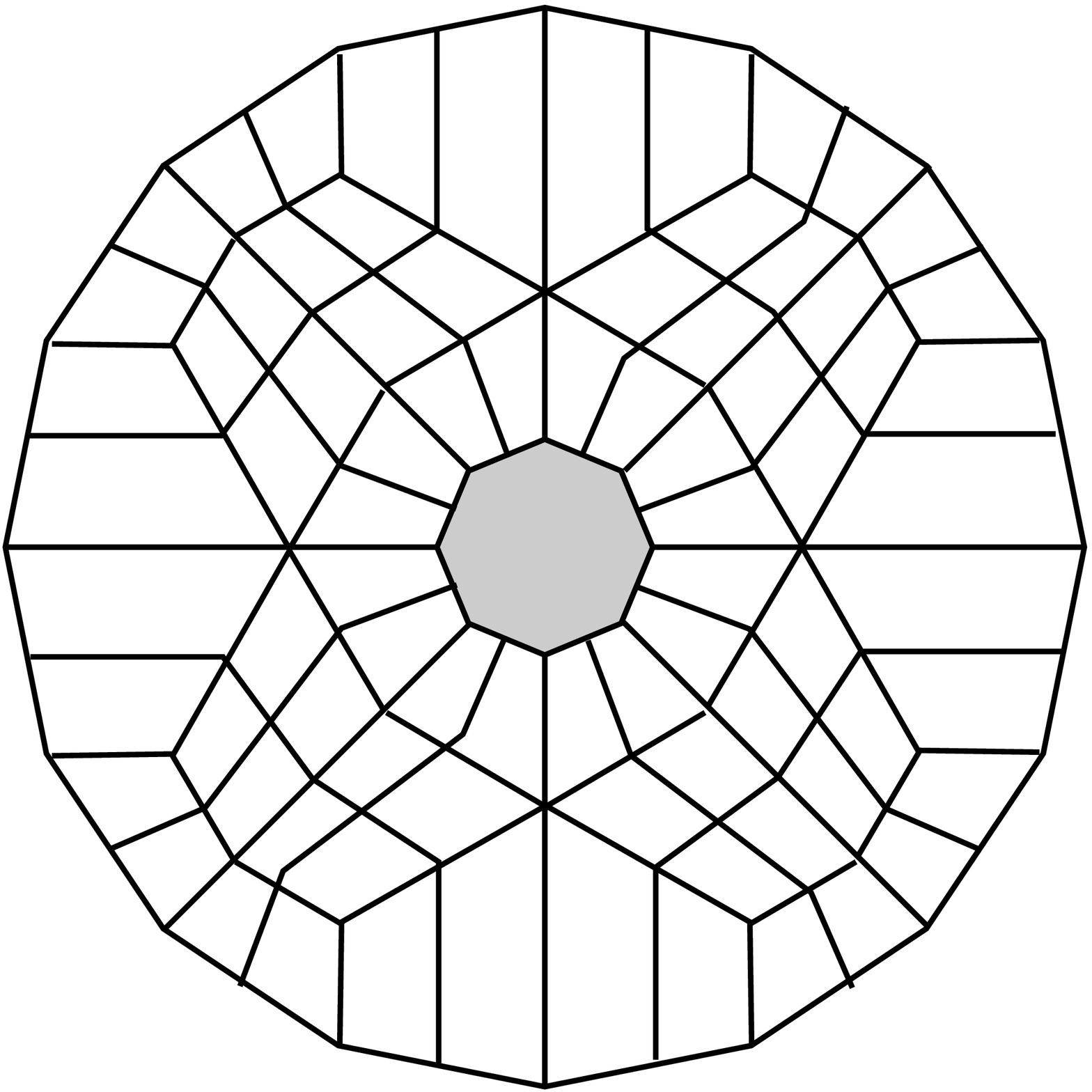}
 }
\caption{ \label{Double}
A mesh of an annular  region (left) and its
double (center).
On the right we only double the boundary edges and propagate these.
}
\end{figure}

\begin{lemma}
A propagation path in a quadrilateral mesh 
can visit each quadrilateral at most twice 
(once connecting each pair of opposite sides)
\end{lemma} 

\begin{proof}
First we show each edge is visited at most once.
Suppose the edge $e$ is visited twice by a path that crosses 
$e$ in the same direction both times. Then the path must cross 
$e$ at the same point both times (by the definition of how 
points propagate). Thus the path is really a closed loop that 
hits $e$ once.  Next suppose $e$ is  visited twice by a propagation path
that crosses in opposite directions each time. Then there is another 
edge $f$, opposite to $e$ on one of the adjacent quadrilaterals, 
so that $f$ is crossed twice by the same path before $e$ is 
crossed twice.   Iterating the argument gives a contradiction 
since the crossings of $e$ are separated by only a finite number of 
steps.  Thus no 
edge is visited twice. Since  each quadrilateral has four sides, 
and each side is visited at most once, each quadrilateral is 
visited at most twice.
\end{proof}

\section{Sinks} \label{sink defn} 

We start by reviewing the definition given in 
the introduction, and  then stating the results that
we will prove in later sections.

We say that a mesh $\Gamma$  of the interior $\Omega$ 
 of a simple polygon  $P$ {\defit   extends
$P$}  if  $V(\Gamma) \cap P = V(P)$, i.e., 
the only vertices of the mesh 
that occur on $P$ are the vertices of $P$ 
(no extra boundary vertices are added). 

As noted in the introduction, a
 {\defit sink} is  a  simple polygon $P$ with the 
property that whenever we add an even number of vertices 
to the  edges of $P$ to obtain a new polygon $P'$, then 
there is nice mesh of the interior that  extends $P'$.
By Lemma \ref{need even}, a sink must 
have an even number of vertices and  it is 
a simple exercise to check 
that neither a square nor a regular hexagon  is a sink.

\begin{lemma} \label{octagon is sink}
The regular octagon is a sink. If we add $M\geq 1$ vertices to the 
edges of a regular octagon $P$, the resulting polygon  $P'$  
can be extended by a 
nice  quadrilateral  mesh with $O(M)$  elements.
\end{lemma} 

From octagons we will construct other sinks, e.g. we can make
a square into a sink by adding 24 vertices to 
its boundary, obtaining a 28-gon.
Using such square sinks we can make any rectangle $R$ into a sink 
 by adding $O(\ecc(R))$  extra
vertices to the sides  of the rectangle
(recall that for a quadrilateral $Q$, 
 $\ecc(Q) \geq 1$ denotes its eccentricity, i.e., 
the longest side length of $Q$
 divided by  the shortest side length of $Q$).  
From the case of rectangles we will deduce:

\begin{lemma} \label{quad sinks lemma 1}
If $Q$ is a nice quadrilateral, then we can make $Q$ into 
a sink $Q'$ by adding $ N=O(\ecc(Q))$ 
vertices to the sides of $Q$. 
If we add $M$ vertices to the sides of $Q'$,  the resulting 
polygon $Q''$  can be extended by a nice mesh 
using $O(N+M^2)$ quadrilaterals. 
More precisely, if   $M=M_1+M_2$, where 
$M_1\geq 1$ is an upper bound for the 
 number of  extra points added to one 
pair of opposite sides of $Q$ and $M_2 \geq 1$ is 
an upper bound for the number of
extra points added to the other pair of opposite 
sides of $Q$ , then the 
number of elements in the nice extension of 
$Q''$  is $O(N + M_1 M_2)$.
\end{lemma} 

In particular, if all the extra vertices are added to a 
single side of the quadrilateral $Q$  (or are only  added to 
a single pair of opposite sides) then the number of
mesh elements is $O(N+M)$.

In addition to the ``regular'' sinks described above, 
we will also need some ``special'' sinks that use some angles 
less than $60^\circ$.
These sinks will be polygons inscribed in a circle or a sector.
A {\defit cyclic polygon} $P$ will refer to a  
simple polygon whose vertices $V$ are all on a circle 
that circumscribes the polygon.
The polygon is {\defit $\theta$-cyclic}  if 
every complementary arc of  $V$ on the circle has   angle measure
$\leq \theta$ (so the smaller $\theta$ is, the more 
$P$ looks like a circle).

In this paper,  an $n$-tuple will always refer
to   an ordered  list    of  $n$ distinct  points
on  a circle,  ordered counter-clockwise.
 Most commonly, we 
will take these on the unit circle  $\circle 
= \partial \disk =\{ z:|z|=1\}$ (a circle is a  1-dimensional 
torus, which is why the unit circle is traditionally denoted
with a $\circle$). 
An $(n,\theta)$-tuple is an $n$-tuple  
$X \subset \circle$ so that 
 each component of 
$\circle \setminus X$ has angle measure $\leq \theta$.
 Given an $n$-tuple $X$, 
we let $X^1$ be the $2n$-tuple obtained by adding the 
midpoint of each circular arc $\circle \setminus X$ and 
define $X^k$ to be the $2^kn$-tuple obtained by repeating 
this  procedure $k$ times. 
If $X$ is an $n$-tuple on $\circle$ and $\lambda >0$, then 
$\lambda X$ is the image of $X$ under the dilation
$ z \to \lambda z$, so is an $n$-tuple on the circle 
$\lambda \circle =\{ z:|z|=\lambda\}$. 

If   $X \subset \circle$ is an
$n$-tuple, then a {\defit sink conforming to $X$} 
is a polygon $P$  inscribed in the unit circle, so that
\begin{enumerate}
\item the vertices of $P$ contain $X$, 
\item if we add vertices to the edges of $P$ to get a 
      new polygon $P'$, then there is quadrilateral    
      mesh that extends $P'$, 
\item
the edges of this mesh cover all the radial segments connecting
the origin to points of $X$,
\item every angle in the mesh
 is between $60^\circ$ and $120^\circ$, except for any angles less 
than $60^\circ$   at the orgin formed by the
radial segments corresponding to $X$;
such angles remain undivided in the mesh.
\end{enumerate}
This mesh  actually 
conforms to the PSLG consisting of the closed radial 
segments connecting the origin to the points of $X$.
We abuse notation slightly by saying it conforms
to $X$.
We will prove in Section \ref{conforming sinks} that

\begin{lemma} \label{sector conform sink}
Given any $d$-tuple $X$ on the unit circle, 
there is a cyclic polygon $P$ with $O(d)$ vertices that 
is a sink conforming  to $X$.  
If $M$ points are added to $P$,
then the corresponding nice  conforming mesh has 
at most $O(M d + d^2   )$ elements.
For any $\theta>0$ the number of mesh elements
that are not $\theta$-nice is $O(Md + d/\theta)$.
If at most $O(K)$ extra points 
are added to each boundary edge of the sink, 
then at most $O(Kd/\theta)$ of the quadrilaterals 
are not $2 \theta$-nice.
\end{lemma}

This lemma  is the device that we will use in the proof of
Theorem \ref{Quad Mesh}  to ``protect'' small
angles in the original  PSLG, i.e., to prevent these
angles from being subdivided.
When the vertex $v$ is on the boundary of a face of 
the PSLG, then instead of using the whole mesh we 
will only use the part that lies inside the face (a
sector at $v$).
This lemma is also one of two places where the worst
case $O(n^2)$ estimate comes from; the other is 
Theorem \ref{quad mesh lemma}. 


\section{The regular octagon is a sink} \label{inscribed polygon}

In this section we prove Lemma \ref{octagon is sink}. We start with:

\begin{lemma} \label{annulus quad lemma} 
Suppose $\theta_1, \theta_2, \theta_3, \theta_4 
\in [0^\circ,360^\circ]$, that $\theta_1 < \theta_3$, 
$\theta_2 < \theta_4$, and that 
$$|\theta_1-\theta_3|, |\theta_2-\theta_4| < 45^\circ, \quad
|\theta_1-\theta_2|, |\theta_3-\theta_4| <  1^\circ.$$
Then the quadrilateral with vertices 
$$ 
z_1 = e^{i \theta_1} = \cos(\theta_1) +i \sin(\theta_1),  \qquad 
z_2 = 2 e^{i \theta_2}=2( \cos(\theta_2) +i \sin(\theta_2)),
$$
$$
z_3 =  e^{i\theta_3} = \cos(\theta_3) +i \sin(\theta_3), \qquad 
z_4 = 2 e^{i \theta_4}= 2(\cos(\theta_4) +i \sin(\theta_4)),
$$
has all its interior angles between $60^\circ$ and 
$120^\circ$.
(Note that $z_1,z_3$ are both on $ \circle$ with $z_1$ 
clockwise from $z_3$, and $z_2, z_4$ are on $2 \circle$ 
with  $z_2$ clockwise from $z_4$; see Figure \ref{AnnulusQuad}). 
\end{lemma} 

\begin{proof}
This is a straightforward trigonometry calculation. We
carry it out in detail for the corner located at $z_1$; 
the other three corners are very similar, and are left 
to the reader.
The situation is illustrated in Figure \ref{AnnulusQuad}. 
Consider 
the segments  $[z_1, z_3]$ and $[z_1, z_2]$.
The angle formed by these two sides at $z_1$   is  the 
sum or difference of the angles that  each of these sides makes 
with the radial line  $L_1$ from the origin through $z_1$.
Set $\tau = \theta_3 - \theta_1$. Then considering 
the isosceles triangle formed by $0, z_1, z_3$, we 
see that  the angle between $[z_1, z_3]$ and $L_1$ 
is   between $90^\circ$ and $90^\circ + \tau/2$
 (see Figure \ref{AnnulusQuad}). 
Since $\tau \leq 45^\circ$, this  is between 
$90^\circ$ and  $112.5^\circ$.

\begin{figure}[htbp]
\centerline{
	\includegraphics[height=3.5in]{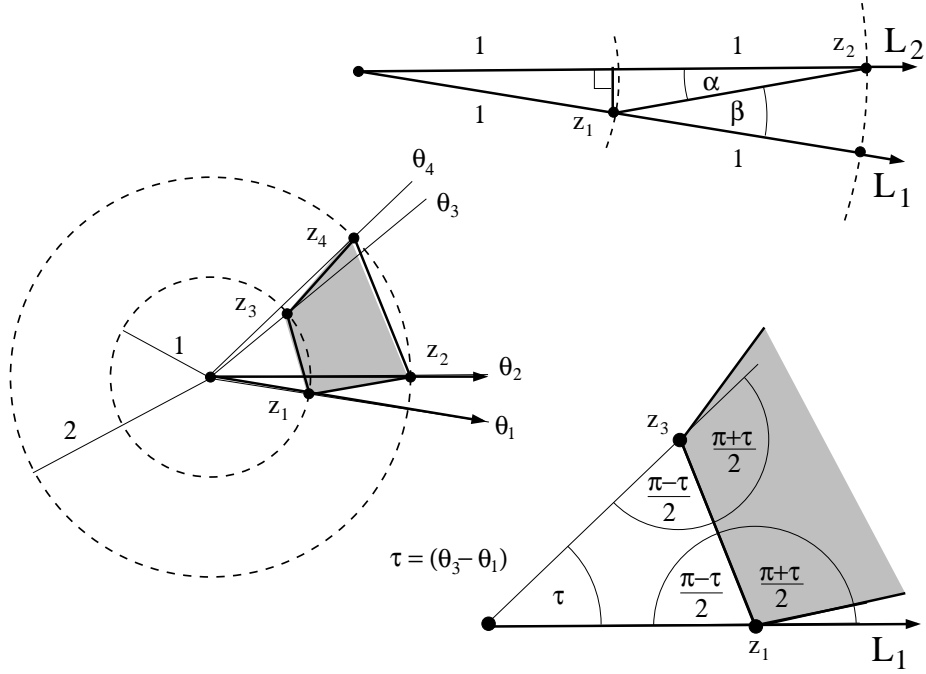}
	}
\caption{ \label{AnnulusQuad}
The diagram in the proof of Lemma \ref{annulus quad lemma}. The quadrilateral 
formed by the points  $z_1,z_2, z_3, z_4$ 
has all four angles between $60^\circ$ and $120^\circ$ if 
$|\theta_1-\theta_3|, |\theta_2-\theta_4|$ are both smaller 
than $45^\circ$ and $|\theta_1-\theta_2|$, $|\theta_3-\theta_4|$
are both small enough, say less than $1^\circ$ (these numbers 
are not sharp). 
}
\end{figure}

On the other hand,  if $\beta$ denotes 
the angle between $[z_1, z_2]$ and
$L_1$ (as  labeled in Figure \ref{AnnulusQuad}) 
and $\alpha$ denotes the angle formed by 
$[z_1, z_2]$ and the ray $L_2$ from  $0$ through $z_2$, 
then $\beta = (\theta_2 - \theta_1) + \alpha$. 
We claim $\alpha \leq \theta_2-\theta_1$. 
To prove this, consider the segment perpendicular 
to $[0, z_2]$  through $z_1$. This intersects 
$[0, z_2]$ closer to $0$ than to $z_2$  (see Figure 
\ref{AnnulusQuad}) and hence $\tan(\alpha) 
\leq \tan(\theta_2- \theta_1)$, giving the claimed
inequality.  Thus 
$$\beta =|\theta_2-\theta_1| + \alpha \leq 
    2|\theta_2-\theta_1| \leq 2^\circ. $$
This proves the angle bounds at $z_1$ with room 
to spare.
The argument for $z_3$ is identical and the arguments
for $z_2$, $z_4$ are  almost the same (we only need to 
estimate $\alpha$, not $\beta$).
\end{proof}

\begin{lemma} \label{homotopy}
Suppose $\theta>0$ 
and that  ${\bf z} =\{z_1, \dots, z_n\}$ and $ {\bf w}=
\{w_1, \dots, w_n\}$ are
both $(n,\theta)$-tuples  on the unit circle (in particular, 
they are both ordered on the circle in the same direction).
Given any $\tau >0$  there is an integer $t$ (depending on $\tau$, 
but not on ${\bf z}$ or ${\bf w}$) and 
  a sequence of $(n,\theta)$-tuples
$\{z_1^k, \dots, z_n^k\}_{k=1}^t$  so that 
$z_j^0 = z_j$ and $z_j^t = w_j$ for $j=1,\dots n$ and so that 
$ |z_j^k-z_j^{k+1}| < \tau$ for all $k=0, \dots , t-1$ and 
$0=1, \dots n$.   
In other words, we can discretely deform ${\bf z}$ into ${\bf w}$
through a sequence of $(n,\theta)$-tuples 
that move individual points by less than $\tau $ at each step.
If $z_j = w_j$ for some $j$, then  the points $z_j^k$, $k=1,\dots
t$  are all the same.
\end{lemma} 

\begin{proof}
Choose arguments $\theta_1, \phi_1  \in (0^\circ, 360^\circ]$  for 
 $z_1$ and $w_1$. Then choose arguments $\{\theta_j\}_2^n$ 
 for $z_2, \dots z_n$ 
in $(\theta_1, \theta_1 + 360^\circ)$; note that the arguments 
increase since we assume the $n$-tuple ${\bf z}$ is ordered
in the counter-clockwise direction. Similarly choose arguments
$\{\phi_j\}_2^n \subset (\phi_1,\phi_1+360^\circ)$ for the elements of 
${\bf w}$.
Now use linear interpolation on the angles, i.e., 
$$ \theta_j^k = (1-\frac {k }{t}) \theta_j +  \frac {k }{t} \phi_j,
\qquad k=0,\dots, t, $$   
to define points $ z_j^k = \exp( i \theta_j^k)$.
See Figure \ref{ArgInterpolate}. 
Since both $n$-tuples have the same orderings, none of the 
lines in Figure \ref{ArgInterpolate}
cross each other, hence these intermediate points define $n$-tuples 
with the correct ordering. Moreover, 
the  $0 \leq  \arg(z_k^j)-\arg(z_k^{j+1}) < 360^\circ/t$ 
is as small as we wish if $t$ is large enough. In particular, 
it is smaller than $\tau$ if $t$ is large enough, depending 
only on $\tau$.
\end{proof}

\begin{figure}[htbp]
\centerline{
 \includegraphics[height=2.2in]{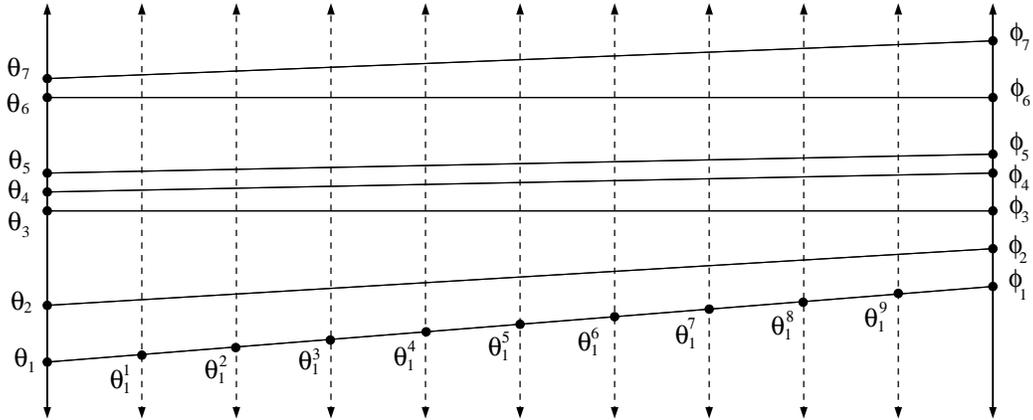}
 }
\caption{ \label{ArgInterpolate}
Linear interpolation of the arguments ($n=7$, $t=10$).
}
\end{figure}

\begin{cor} \label{annular mesh angles}
Suppose  that  ${\bf z} =\{z_1, \dots, z_n\}$ and $ {\bf w}=
\{w_1, \dots, w_n\}$ are $(n, 45^\circ)$-tuples. There is 
an integer $s$ so that the annular  region  bounded between $P_{\bf z}$ 
and $2^{-s} P_{\bf w}$ can be meshed with $ n s$ quadrilaterals
using only angles between $60^\circ$ and $120^\circ$.
(Recall that $P_{\bf z}$ is a cyclic polygon inscribed
on the unit circle $\circle$, and $2^{-s}P_{\bf w}$  
is a cyclic polygon inscribed on $2^{-s}\circle = \{|z| =
2^{-s}\}$.)
If $z_j = w_j$ for some $j$, then the mesh can be chosen 
so that its edges cover the radial segment from $z_j$ to 
$2^{-s} w_j$.
\end{cor} 

\begin{proof}
This is immediate from Lemma \ref{annulus quad lemma}
and Lemma \ref{homotopy}.
See Figure \ref{Spider1}.
\end{proof}

\begin{figure}[htbp]
\centerline{
 \includegraphics[height=3.0in]{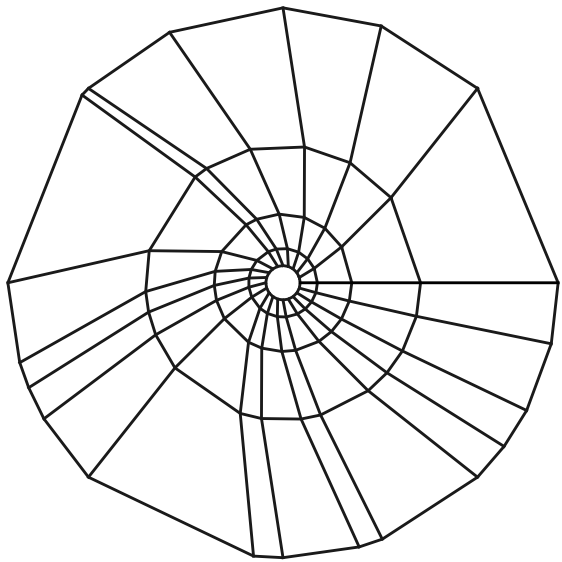}
 }
\caption{ \label{Spider1}
A quadrilateral mesh of the annular region bounded 
by two $20$-gons inscribed on concentric circles (the 
outer  $20$-tuple is a regular octagon with 12 randomly 
chosen points added; the inner one is a regular icosagon).
This illustrates Corollary \ref{annular mesh angles}
}
\end{figure}


\begin{lemma} \label{rounding lemma} 
Suppose ${\bf z} = \{ z_1, \dots , z_n \}$
is an $(n, 45^\circ)$-tuple on the unit circle and 
suppose $V$ is a set of $t$   points on 
the edges of the  inscribed polygon $P_{\bf z}$
 corresponding to ${\bf z}$. Let 
${\bf w}$ be the $(n+t)$-tuple consisting of ${\bf z}$ 
and the radial projection of $V$ onto the unit circle.
Let $A$ be the annular region bounded between 
$P_{\bf z}$ and $\frac 12 P_{\bf w}$ and consider the 
quadrilateral mesh formed by joining each point of 
$\frac 12 {\bf w} \subset \{|z|=1/2\}$
 to its radial projection onto $P_{\bf z}$.
Then all angles used in this mesh are between 
$67.5^\circ$ and $112.5^\circ$. 
\end{lemma}

\begin{proof}
Each  angle on $ \frac 12  P_{\bf w}$  is the supplement 
of  the  base
angle of an isosceles triangle with vertex at the origin and 
vertex angle $\leq 45^\circ$, hence is between
$90^\circ$ and $112.5^\circ$.  
Each angle on $P_{\bf z}$ is clearly   between $67.5^\circ$ 
and $112.5^\circ$. See Figure \ref{Rounding1}.
(The figure shows ${\bf z}$ as a regular octagon, but this need
not be the case in general.)
\end{proof} 

\begin{figure}[htbp]
\centerline{
 \includegraphics[height=2.0in]{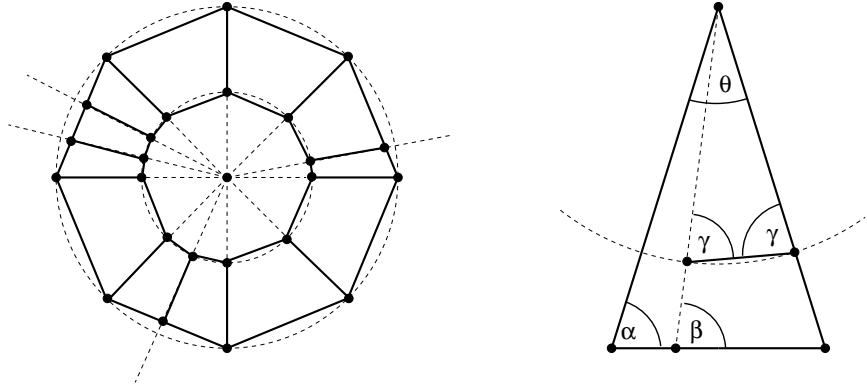}
 }
\caption{ \label{Rounding1}
The annular region between a octagon with four 
extra points added to the sides  and the ``rounded''
version of this polygon. Since $\theta \leq 45^\circ$, 
we  have $\alpha = 90^\circ - \theta/2$ is 
between $67.5^\circ$ and $90^\circ$. Similarly 
 $  67.5^\circ \leq  \alpha \leq \beta \leq 180^\circ - \alpha
 \leq 112.5^\circ$. Finally, $ 90^\circ \geq\gamma 
\geq 90^\circ- \theta/2 \geq 67.5^\circ$ as claimed in 
the text. 
}
\end{figure}

\begin{lemma} \label{easy mesh}
Suppose  $P$  is a regular octagon, centered 
at the origin with  four of its vertices on the
coordinate axes.
Nicely mesh the interior using twelve quadrilaterals, 
as shown in Figure \ref{OctoMesh3}.
 Place $t$ distinct points
on the sides of octagon that touch the  $x$-axis, and 
connect each of these to another boundary point by 
a propagation path in the given mesh of the octagon.
 The resulting refinement is a nice mesh with  at most
 $12+4t$ elements.
\end{lemma}

\begin{proof}
The niceness  is immediate from Lemma \ref{split quad}.
The complexity bound follows from the fact that the 
propagation paths don't intersect unless they agree 
(i.e., two of the points are endpoints of the same path),
and each passes through four mesh elements.
See Figure \ref{OctoMesh3}. 
\end{proof}

\begin{figure}[htbp]
\centerline{
 \includegraphics[height=1.6in]{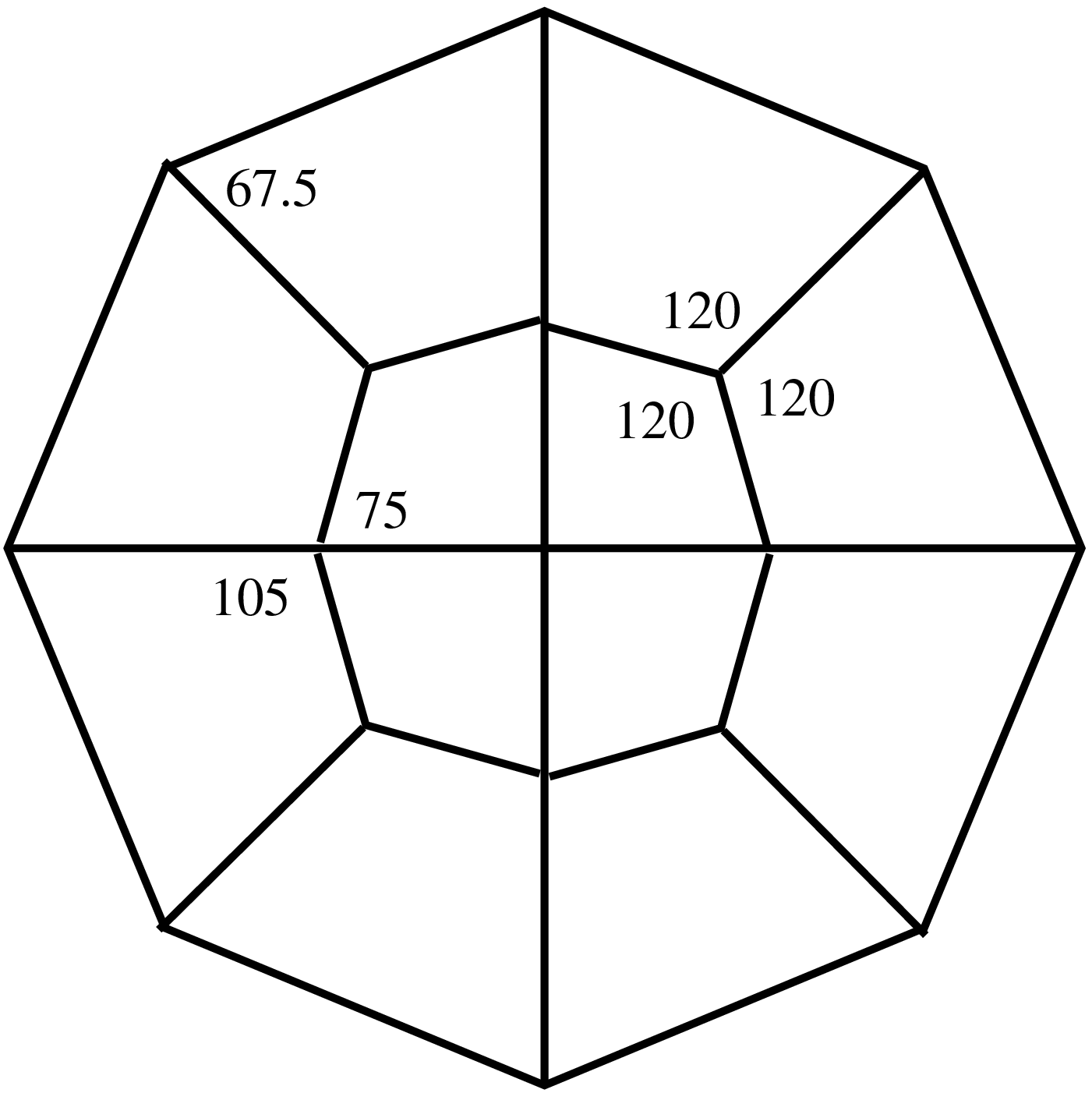}
$\hphantom{x}$
 \includegraphics[height=1.6in]{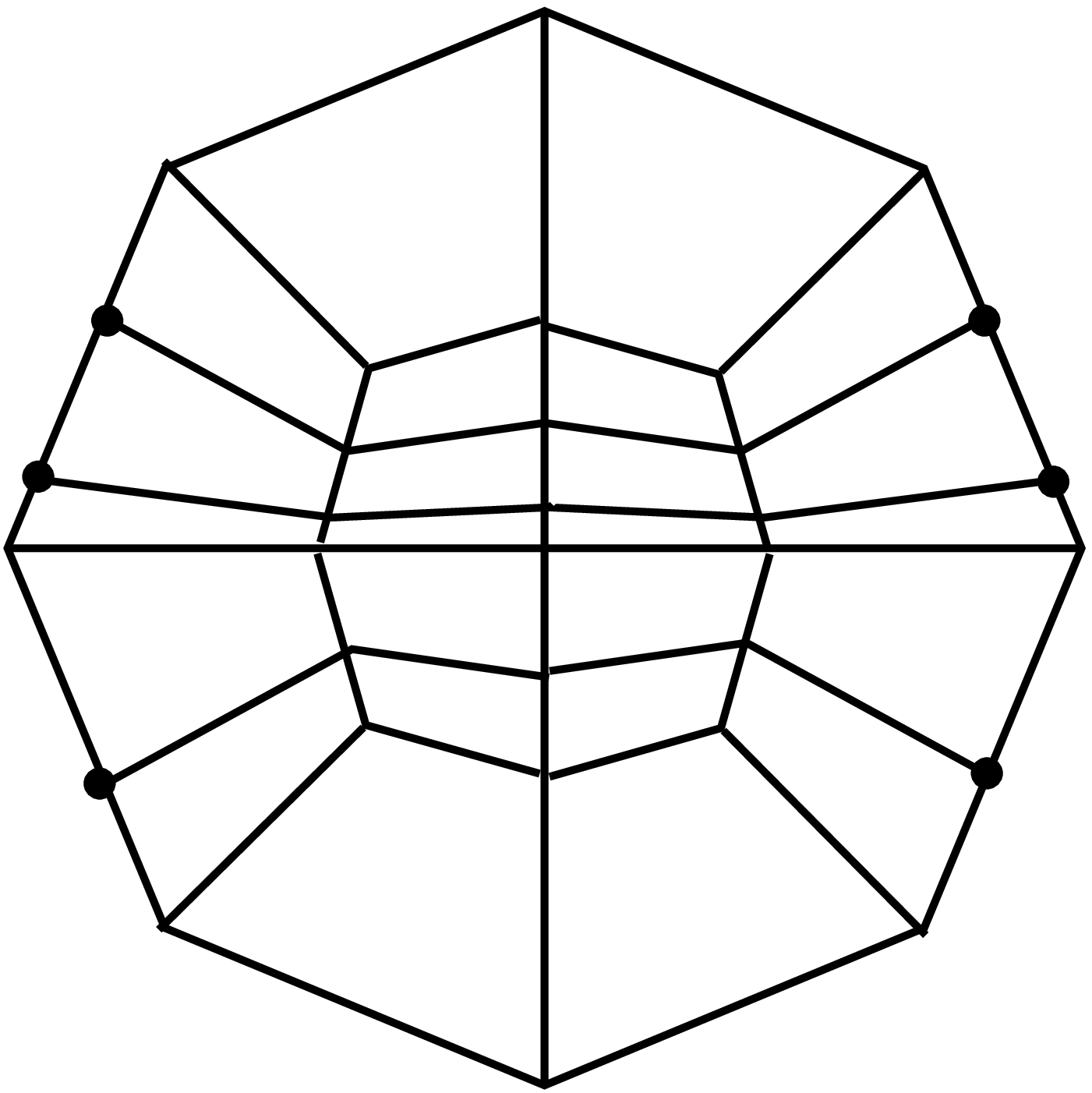}
$\hphantom{x}$
 \includegraphics[height=1.6in]{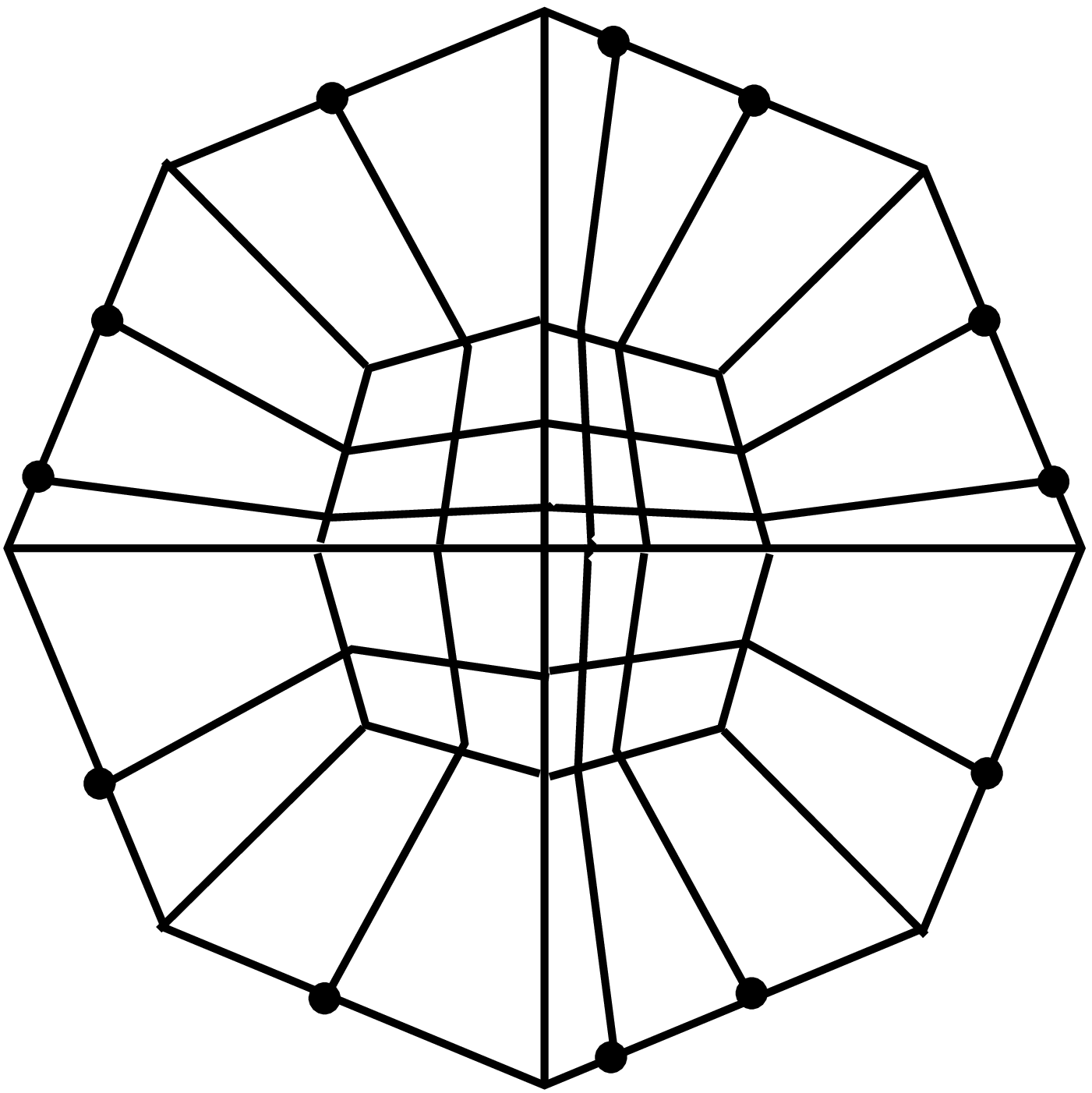}
 }
\caption{ \label{OctoMesh3}
On the left is the  ``standard mesh'' of a
regular octagon. In the center is a mesh 
formed by non-intersecting propagation paths; 
given $t$ extra points, $4t$ new elements are
created. 
In general, $\simeq t^2$ new elements might be formed
(right side) by intersecting propagation paths.
}
\end{figure}

\begin{lemma}  \label{cyclic mesh} 
Suppose ${\bf z}$ is an $(n,45^\circ)$-tuple on the 
unit circle and $P$ is the cyclic polygon with these 
vertices. Suppose we  form a new polygon $P'$ by adding 
$k$ distinct points
to $P$ (all distinct from ${\bf z}$), so that $n+k$ is even.
Then $P'$ has a nice mesh with $O(n+k)$ elements. 
In particular,  this holds if  $P$ is a regular octagon, 
so this result contains Lemma \ref{octagon is sink} as
a special case.
\end{lemma}

\begin{proof}
Let ${\bf z''}$ be the radial projection of ${\bf z}$ onto 
$\frac 12 \circle = \{ |z| = \frac 12\}$  and let 
$P''$ be the cyclic polygon with these vertices.
Clearly $n \geq 8$, so $n+k \geq 8$. Hence $n+k = n+2t$ for 
some $t \geq 0$. 
Let ${\bf w}$ be some $(8+2t)$-tuple as  described  in 
Lemma \ref{easy mesh}:
eight points are evenly spaced on the unit circle and
the other $2t$ are the endpoints of $t$   non-intersecting
propagation paths for the standard mesh of the
octagon shown in Figure \ref{OctoMesh3}.
Let $s$ be the constant from Corollary 
\ref{annular mesh angles}. 
Let $P_{\bf w}$ be the cyclic polygon on $\circle$ 
with vertices ${\bf w}$ and let $2^{-s-1} P_{\bf w}$ be
the dilation of this by a factor of $2^{-s-1}$ (so it 
is a cyclic polygon on $2^{-s-1} \circle$). 

 Nicely mesh 
$2^{-s-1} P_{\bf w}$ by cutting the standard mesh 
of the octagon by the $t$ propagation paths 
as in Lemma \ref{easy mesh}. 
Let ${\bf z''}$ be the rounded version of ${\bf z'}$.
Apply Lemma \ref{annular mesh angles}   to nicely 
mesh the region between $2^{-s-1}P _{\bf w}$ and 
$\frac 12 P_{\bf z''}$. Then  use Lemma   \ref{rounding lemma}
to nicely mesh the region between $ \frac 12 P_{\bf z''}$
and $P'$. Thus the $P'$ has been nicely meshed 
without adding any extra boundary points, as claimed.
See Figure \ref{OctoMesh7} (but note that scales are distorted 
in this figure to make the smaller layers more visible).
\end{proof}

\begin{figure}[htbp]
\centerline{
 \includegraphics[height=3.0in]{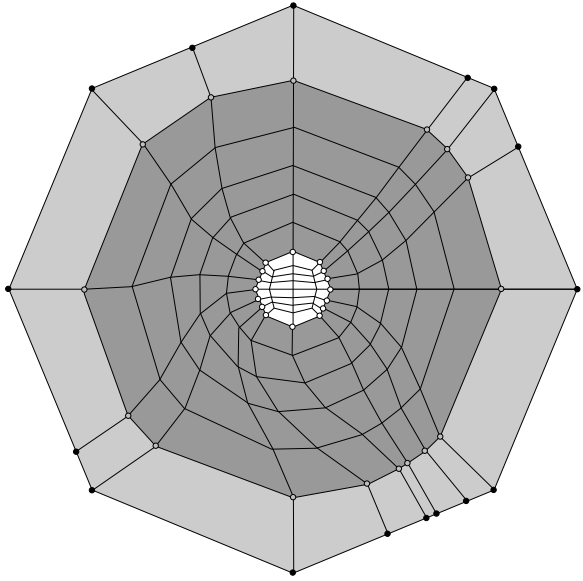}
 }
\caption{ \label{OctoMesh7}
This illustrates the proof of Lemma \ref{cyclic mesh}  
with $n=8, k=8, t=4, s=5$. The back dots are ${\bf z'}$, 
the gray dots are $\frac 12 {\bf z''}$ and the white dots 
are $2^{-s-1}{\bf w}$.  The light gray region is meshed 
by Lemma \ref{rounding lemma}, 
the dark gray by Corollary \ref{annular mesh angles} and the 
white by Lemma \ref{easy mesh}. These results ensure the 
angles in each region satisfy the desired bounds. 
Note that the cyclic polygons should be 
inscribed on  circles   with 
radii $ 1, \frac 12 , \dots, \frac 1{32}$,  but to make
the combinatorics  more visible we have made the circles 
larger than this (which distorts the angles).
}
\end{figure}

\section{Square shaped sinks} \label{square sinks}

Given octagonal sinks, 
we  can  build sinks with other shapes.
In this section we show that there are ``square shaped sinks'':

\begin{lemma}
We can  add 24 vertices to the sides of 
a square $S$ to make it into a sink.
If we then add $2M>0$ extra vertices to the 
sides of the sink, the interior can be 
re-meshed using $O(M^2)$  nice quadrilaterals.
If  at most $M_1 \geq 1$ points are   added to  one 
pair of opposite sides of $S$, and at most $M_2 \geq 1$ are added
points are added 
to the other  pair of opposite 
sides, then the number of quadrilaterals 
needed is $O(M_1 M_2)$. 
\end{lemma} 

\begin{proof}
We want to place a copy of a regular octagon  
inside a square in a certain way. 
Suppose we have a octagon centered at the origin, with
sides of length 1 and two sides parallel to each coordinate
axis. Let $a$ denote the origin. 
Let   $b= (\frac 12 + \frac 1{\sqrt{2}},\frac 12),$
and $c = (\frac 12 , \frac 12 + \frac 1{\sqrt{2}})$ 
be the  vertices of the octagon in the first 
quadrant.
 See Figure \ref{SquareSink6}.
 Now form an equilateral triangle $T$ with one
side parallel to the $x$-axis and passing through $b$,
the  vertex $e$ opposite this side on the same vertical line as $b$, 
and a side passing from $e$  through $c$ ($T$ is the 
 dashed triangle in 
Figure \ref{SquareSink6}). Let $d$ be the vertex of the 
triangle inside the octagon. Let $f$ be the third
vertex of the triangle. Let $S$ be the axis-parallel 
square centered at the origin and passing through $e$.

\begin{figure}[htbp]
\centerline{
 \includegraphics[height=2.5in]{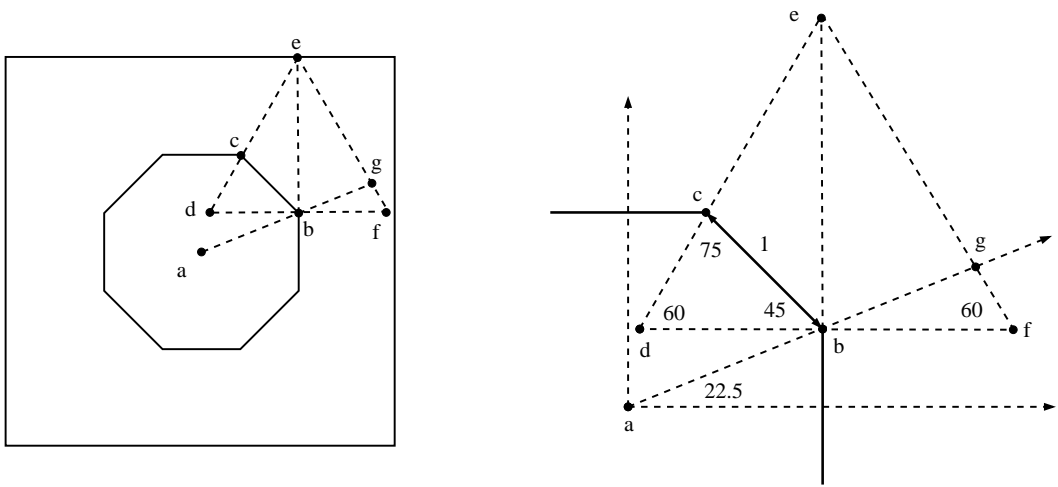}
}
\caption{ \label{SquareSink6}
Placing the octagon within a concentric square. 
}
\end{figure}

 The Law of Sines implies that 
$$ |f-b|=|d-b| = |b-c|  \frac {\sin 75^\circ}{\sin 60^\circ} = 1 \cdot 
 \frac {1+\sqrt{3}}
{\sqrt{6}} \approx 1.11536,$$
and hence
$$ |e-b| = |d-b| \tan 60^\circ = \frac {1+\sqrt{3}}{\sqrt{2}} 
\approx 1.93185.
$$
Therefore 
$$ e =( \frac 12 + \frac 1{\sqrt{2}},  \frac 12 + \frac {1+\sqrt{3}}{\sqrt{2}})
     \approx (1.20711 , 2.43185),
$$
and 
$$ f = (\frac 12 + \frac 1{\sqrt{2}} + \frac {1+\sqrt{3}}
{\sqrt{6}}  ,\frac 12 )  \approx (2.32246   , .5).
$$
Hence  $f$ is inside the square $S$. Therefore the line 
through $a$ and $b$ intersects the segment $[e,f]$ at a point 
$g$  strictly inside the square, as shown in Figure 
\ref{SquareSink6}.

We now mesh the region between the square and the octagon as 
in Figure \ref{SquareSink7}.

\begin{figure}[htbp]
\centerline{
 \includegraphics[height=2.2in]{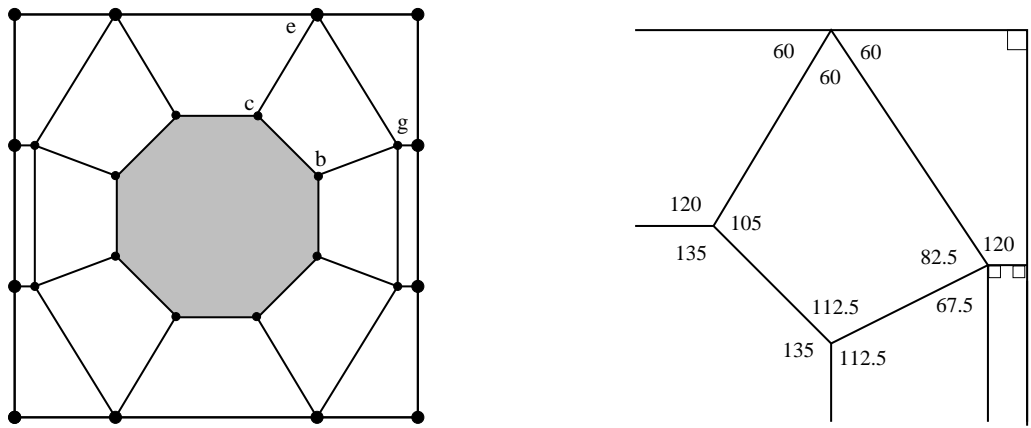}
}
\caption{ \label{SquareSink7}
The mesh of the region between the octagon and the square
using the points defined in the text. The mesh is symmetric 
with respect to the $x$ and $y$ axes, so it suffices the 
check the angles in the upper right corner, as done in 
the right side of the figure.
}
\end{figure}

Consider boundary points of the square that propagate through 
this mesh. See Figure \ref{SquareSink8}. Every point on 
the two vertical sides of the square hits the octagon
under propagation, and this is true for some points on the 
horizontal sides of the square, but there are some points 
that propagate to the other side of the square without hitting 
the octagon (e.g., the dashed line on the far right in 
Figure \ref{SquareSink8}). 

\begin{figure}[htbp]
\centerline{
 \includegraphics[height=2.0in]{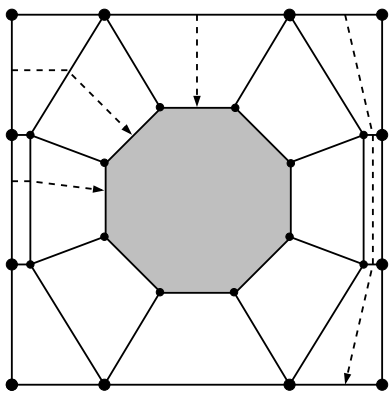}
}
\caption{ \label{SquareSink8}
Some boundary points propagate to the octagon, but other 
propagate past it.
}
\end{figure}

To fix this, we place two small squares inside $S$
 as shown in 
Figure \ref{SquareSink9}. Inside each square 
we place an octagon and mesh the region between the 
small square and the small octagon as shown in Figure 
\ref{SquareSink6}, but rotated by $90^\circ$, so that 
every propagation path hitting the small square from 
above or below propagates to the small octagon inside
it. Paths hitting the small square on the vertical 
sides  either propagate to the small octagon, or 
past it and then to the large octagon. In either case, 
every boundary point of the large square $S$ now 
propagates to one of the three octagons.

\begin{figure}[htbp]
\centerline{
 \includegraphics[height=1.8in]{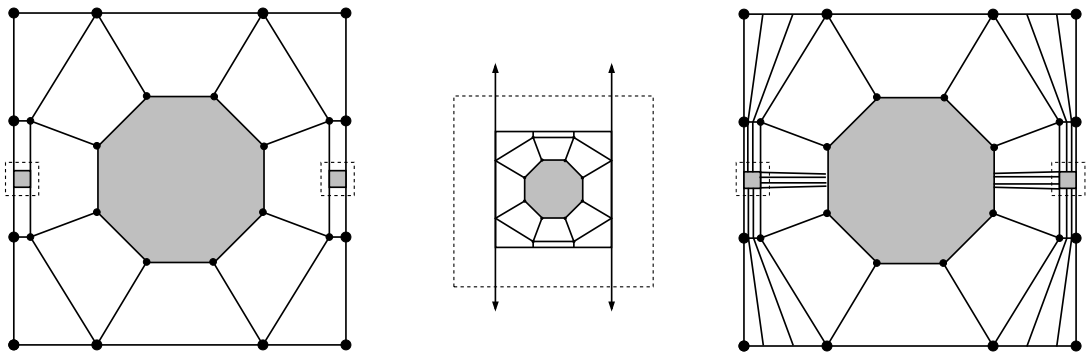}
}
\caption{ \label{SquareSink9}
Now all boundary points of the square propagate
to one of three octagons inside the square. After 
re-meshing to account for the new vertices on the 
small squares, the  mesh  has 28 vertices on 
the boundary of the large square.
}
\end{figure}

If we add an even number of new vertices to the boundary 
of the large square, we can propagate these until they 
each hit one of the three octagons. If each octagon 
gets an even number of points, then, since each octagon 
is a sink, we can re-mesh the interior of the octagons 
using these points and without adding any new points to 
the boundaries of the octagons.  If some octagon gets
an odd number of extra  points, then exactly two of them 
do. If these two are the central octagon and one of 
the smaller side octagons,  we  connect them  a propagation 
path, thus adding one more point to each boundary. 
If the two  side octagons have an odd number of extra 
points each, then we  connect them both to the central octagon 
by a propagation path; this  adds one point to each 
of the smaller octagons 
and adds two points  to the 
central octagon.
In either case,  all three octagons  now 
have an even number of boundary 
points and we can proceed as before.
This proves a square can be turned into a sink; counting 
points in Figure \ref{SquareSink9} shows that the sink has  seven edges
on each side of the original square, hence it has 28 vertices 
in all (including the corners of the square).

If we add $M>0$  points to the boundary of the square, and 
propagate these points until they hit one of the three 
octagons, then $M$ points are created on the boundaries 
of the octagons. After add extra paths (for parity) 
and re-mesh the octagons, $O(M)$ quadrilaterals are 
created inside the octagons. However, up to $O(M^2)$ 
elements might be formed between the octagons and the 
boundary of the square. For example, if $M_1, M_2 \geq 1$ points 
are added to two edges that are both adjacent to the same 
corner of the square $S$, then the every propagation path 
from one group crosses every path from the other group,
generating $ M_1 \cdot M_2$ quadrilaterals. 
See Figure \ref{SquareSink11}.  This can 
only happen when extra points are added to adjacent sides
of $S$; if all the extra points are added to one side, or 
to a single pair of opposite sides of $S$, then only 
$O(M)$ quadrilaterals are needed to re-mesh the interior 
of $S$.  
\end{proof} 

\begin{figure}[htbp]
\centerline{
 \includegraphics[height=2.0in]{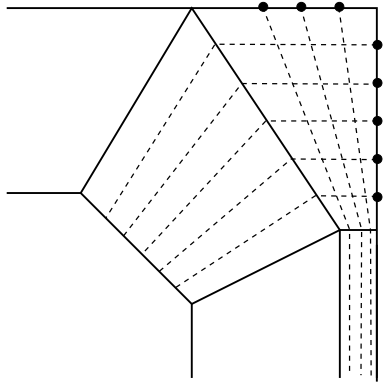}
}
\caption{ \label{SquareSink11}
$M_1, M_2$ extra points placed on two edges   
 adjacent to a corner of  a square sink
can create $M_1 \cdot  M_2$ mesh elements.
}
\end{figure}

\section{Rectangular Sinks} 
 \label{rectangular sinks}

Let $\lceil x \rceil = \min\{n \in \integers: n \geq x\}$. 

\begin{lemma} \label{rect from square} 
For any rectangle $R$ we can add  at most $28 \cdot  
\lceil\ecc(R) \rceil$ vertices to 
the boundary and make $R$ into a sink.
If we then add $M_1 \geq 1$ extra vertices to one pair
of opposite sides,
and add $M_2 \geq 1$ extra vertices to the other 
pair of opposite sides, then the interior can be
re-meshed using $O(\ecc(R)+M_1 M_2)$  nice quadrilaterals.
\end{lemma} 

\begin{proof} 
The idea for this  proof was suggested by
one of the referees and 
modifies an earlier proof of the author.

We claim the construction of square sinks in the 
previous section can 
be adapted to  also handle rectangles with  
$ \ecc(R) \leq \sqrt{2}$.
See Figure \ref{SquareSink10}. This figure  
replicates the left side of Figure \ref{SquareSink9} and 
shows that if the two vertical sides of the square are moved 
outward and the small shaded squares are expanded 
accordingly, the same construction gives a rectangular sink.
This works as long as  the expanding  shaded squares do not 
cover the point  $g$. This holds as long 
as $\Re(p)-\Re(g) \leq 2 \Im(g)$ ($\Re$ and $\Im$ denote 
the real and imaginary parts of a complex number). 

\begin{figure}[htbp]
\centerline{
 \includegraphics[height=1.8in]{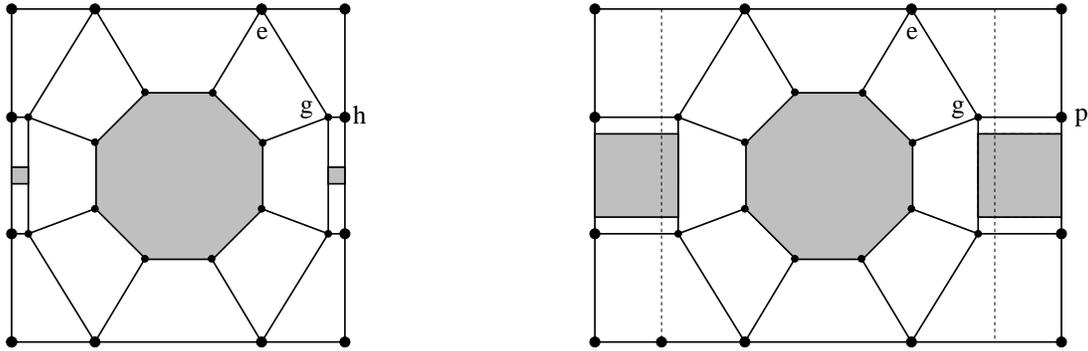}
 }
\caption{ \label{SquareSink10}
The construction for squares also works for rectangles 
with eccentricity sufficiently close to $1$.  The 
left side is the same as in Figure \ref{SquareSink9}. 
The right side is obtained by moving the vertical sides
of the big square outwards and expanding the small shaded 
squares as shown (we also expand the mesh inside these 
squares accordingly).
}
\end{figure}

Using the  formulas from Section \ref{square sinks}  we can 
show that 
 \begin{eqnarray*}
    |b-f| &=&  (1+\sqrt{3})/\sqrt{6}  \approx 1.11536, \\
    |b-g| &=& |b-f| \sin 60^\circ/\sin 97.5^\circ \approx .974261 ,\\
     \Re(g) &=& \Re(b) + |b-g| \cos 22.5^\circ \approx 2.10721,\\
     \Im(g) &=& \Im(b) + |b-g|\sin 22.5^\circ \approx .872833.
\end{eqnarray*} 
These estimates   imply
$A= \Re(g) + 2 \Im(g) \approx 3.85287$. 
Since $ \Im(e) \approx 2.43185$,  this says 
the construction works for rectangles with eccentricity  
up to $A/\Im(e)  \approx  1.58434$. 
This number is larger  that
$\sqrt{2}\approx 1.41421$, proving the claim.

It is not hard to see that 
a  $1 \times r$ rectangle can be sub-divided
into at most $\lceil r \rceil$  rectangles with eccentricity 
in $[1,\sqrt{2}]$ (it suffices to consider $1<r< 2$ and 
note that a $1\times \sqrt{2}$ rectangle can be split into 
two $ 1 \times \frac 1{\sqrt{2}} $ rectangles that also 
have eccentricity $\sqrt{2}$). 
 Therefore we can place 
a modified square sink in each sub-rectangle. 
When we add new boundary points to the 
large rectangle, 
every such point  is on the boundary 
of one of the sub-rectangles.  
If every  sub-rectangle gets  an even number of boundary points, 
we simply re-mesh all the  rectangular sinks.  
Otherwise, we can add points to the common sides 
of the sub-rectangles so that they all end up with 
an even number of new points (see Figure  \ref{Evenness2})
and then re-mesh. 

\begin{figure}[htbp]
\centerline{
 \includegraphics[height=0.75in]{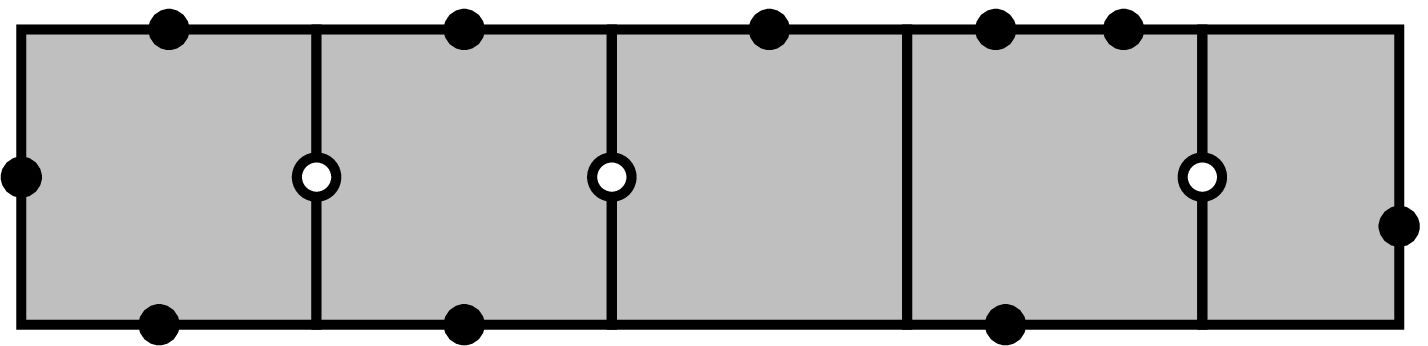}
 }
\caption{ \label{Evenness2}
Any $1 \times r$  rectangle can be sub-divided into 
at most $ \lceil r  \rceil$ 
 sub-rectangles with eccentricity 
bounded by $\sqrt{2}$, and each such sub-rectangle is 
a sink. When we add extra boundary points, it is easy 
to add internal points to make sure each sub-rectangle
ends up with an even number of boundary points. Then 
each can be re-meshed. 
}
\end{figure}

The bound on the number of mesh elements needed 
when we add extra boundary points 
follows immediately from the similar bounds for 
the square and near-square sub-sinks. 
\end{proof}

\section{Nice quadrilateral sinks} 
 \label{quad sinks sec}

Finally, we are ready to prove Lemma \ref{quad sinks lemma 1}
(nice quadrilateral sinks exist).


\begin{proof}[Proof of Lemma \ref{quad sinks lemma 1}]
As shown in Figure \ref{QuadSink1},  
we can mesh a  nice quadrilateral $Q$ using nine elements, 
four of which are rectangles, each rectangle  touching one side
of $Q$ and with the sides of the rectangle   parallel or perpendicular 
to that side of $Q$. 
It is  easy to check that each  of the five  
remaining regions
is $\theta$-nice  if $Q$ is $\theta$-nice.
It is also    easy 
to see  that we can choose the rectangles 
with eccentricity bounded by  $O(\ecc(Q))$.

\begin{figure}[htbp]
\centerline{
 \includegraphics[height=1.0in]{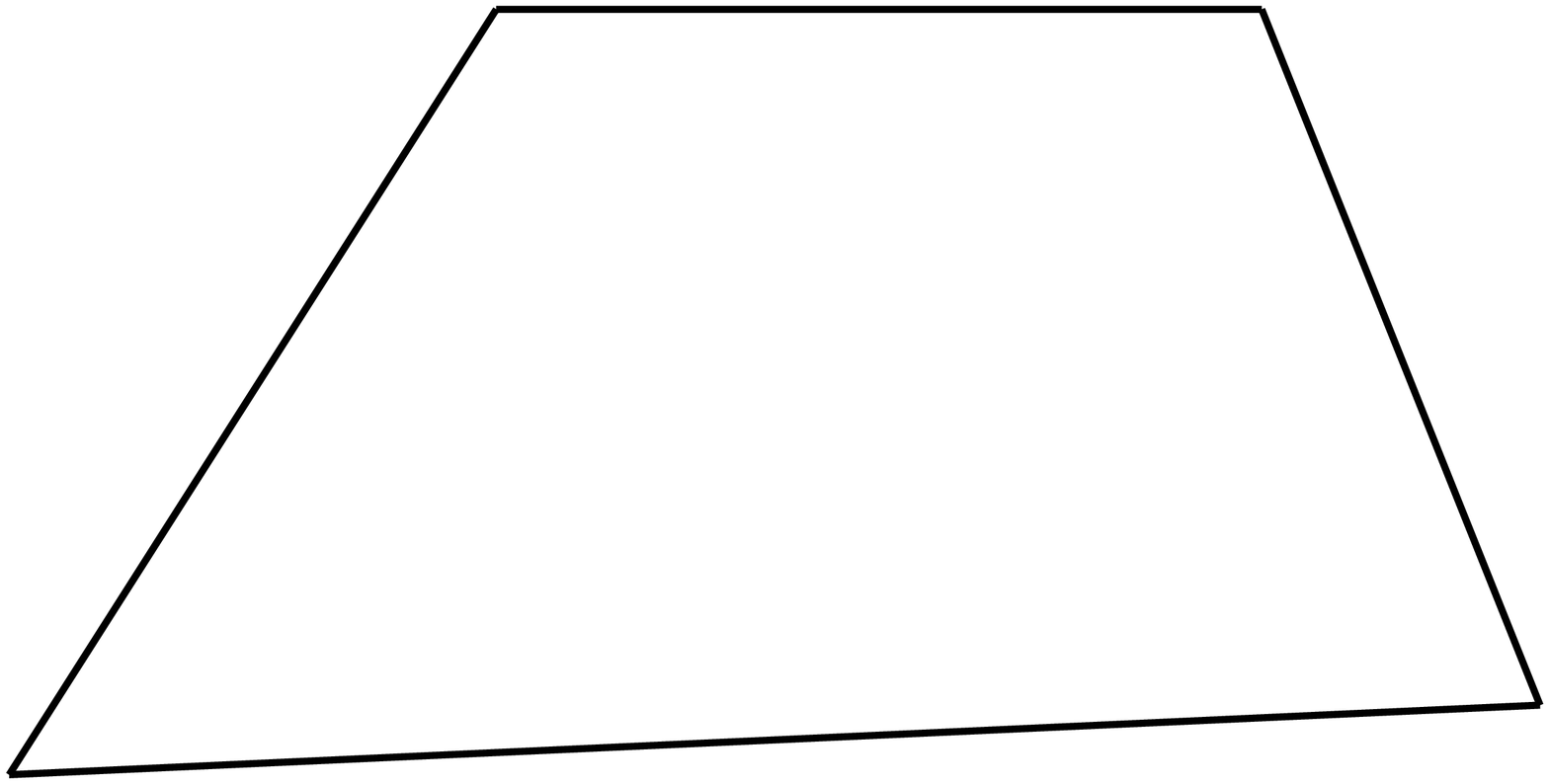}
$\hphantom{xx}$
 \includegraphics[height=1.0in]{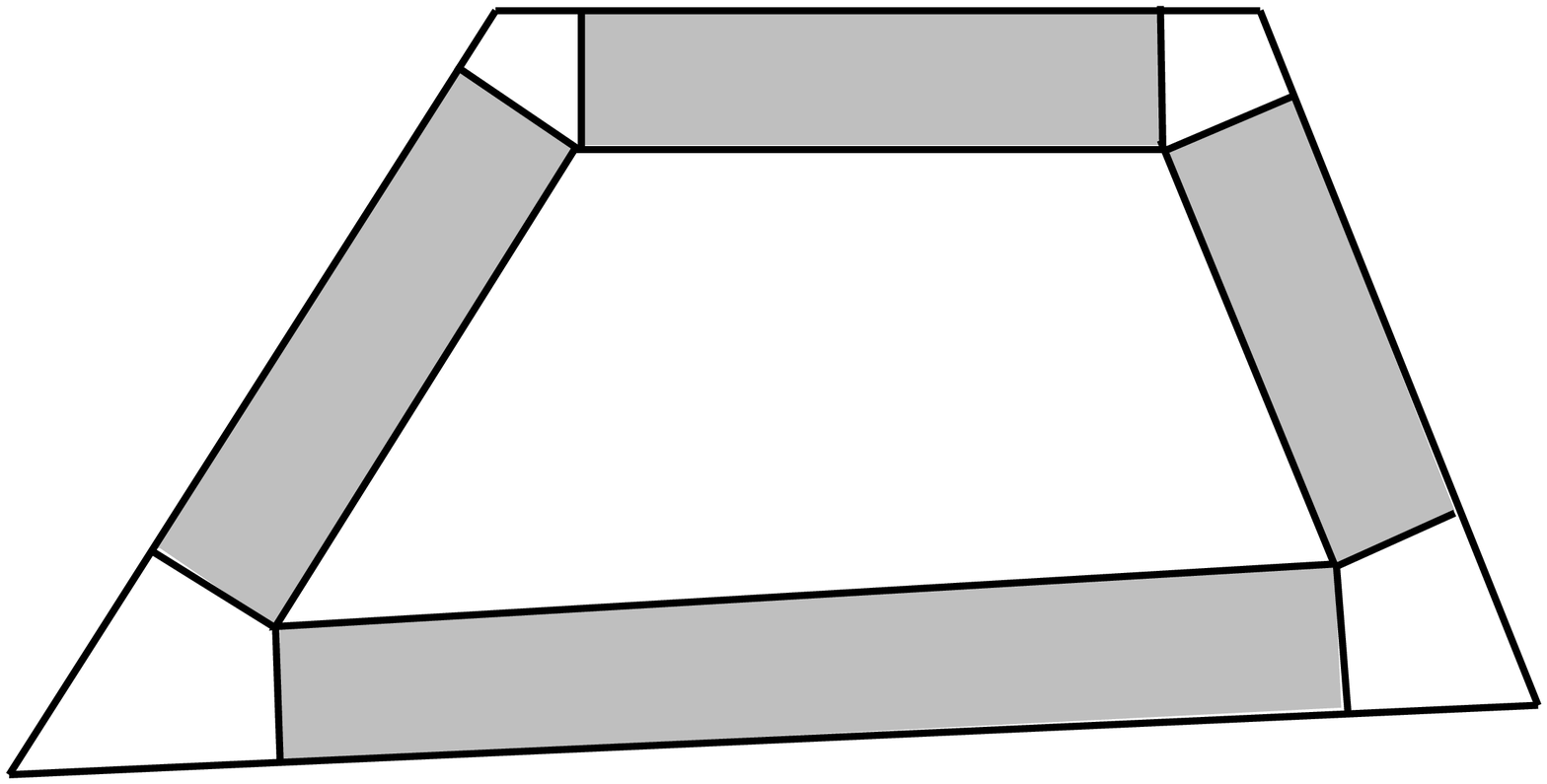}
 }
\caption{ \label{QuadSink1}
The nice quadrilateral $Q$ (left) is meshed with   4 rectangles 
(shaded) and 5 smaller quadrilaterals (white). Every  boundary point of 
$Q$ is either on one of the rectangles, or it propagates
through one of the white  corner quadrilaterals to a rectangle.
}
\end{figure}

When we place extra points on the boundary of $Q$, each 
point is either on the side of a rectangle or on the 
boundary of  
one of the four corner regions. In the latter case, 
we propagate the point through the corner to a rectangle.
If every rectangle ends up with 
 an even number of the new boundary points, 
then we can re-mesh each rectangle since they are all sinks. 
Otherwise, either all four rectangles get an odd number of 
boundary points, or exactly two do.  In the first case, we can 
connect opposite sides by a propagation segment as shown
on the top right  of Figure \ref{QuadEven2}. If exactly two rectangles 
get an odd number of boundary points,
 and these two rectangles 
touch opposite sides of $Q$, then we again connect them 
by a propagation segment as shown on the top left in Figure 
\ref{QuadEven2}. 

\begin{figure}[htbp]
\centerline{
 \includegraphics[height=1.0in]{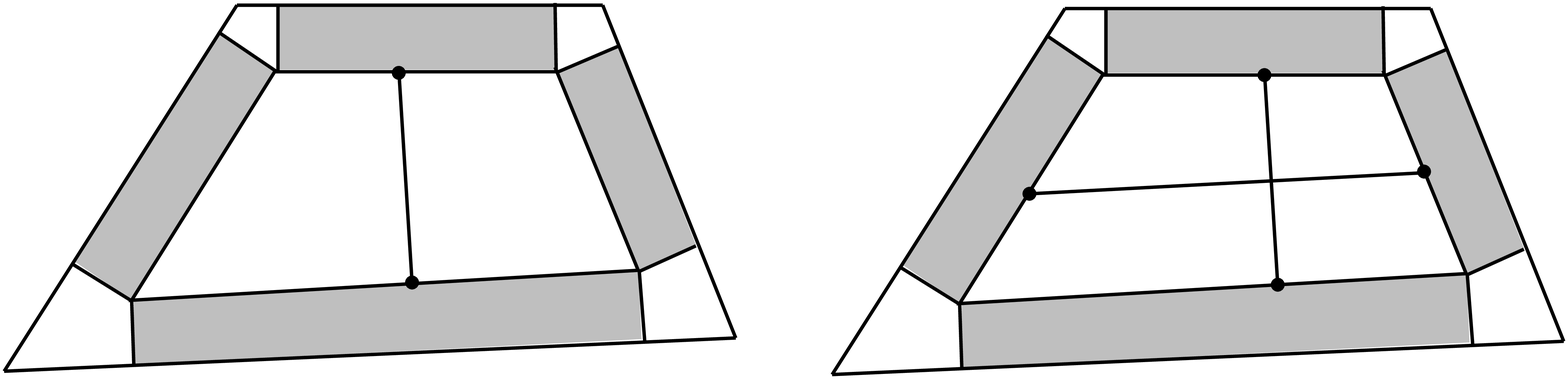}
 }
\vskip.5in
\centerline{
 \includegraphics[height=3.0in]{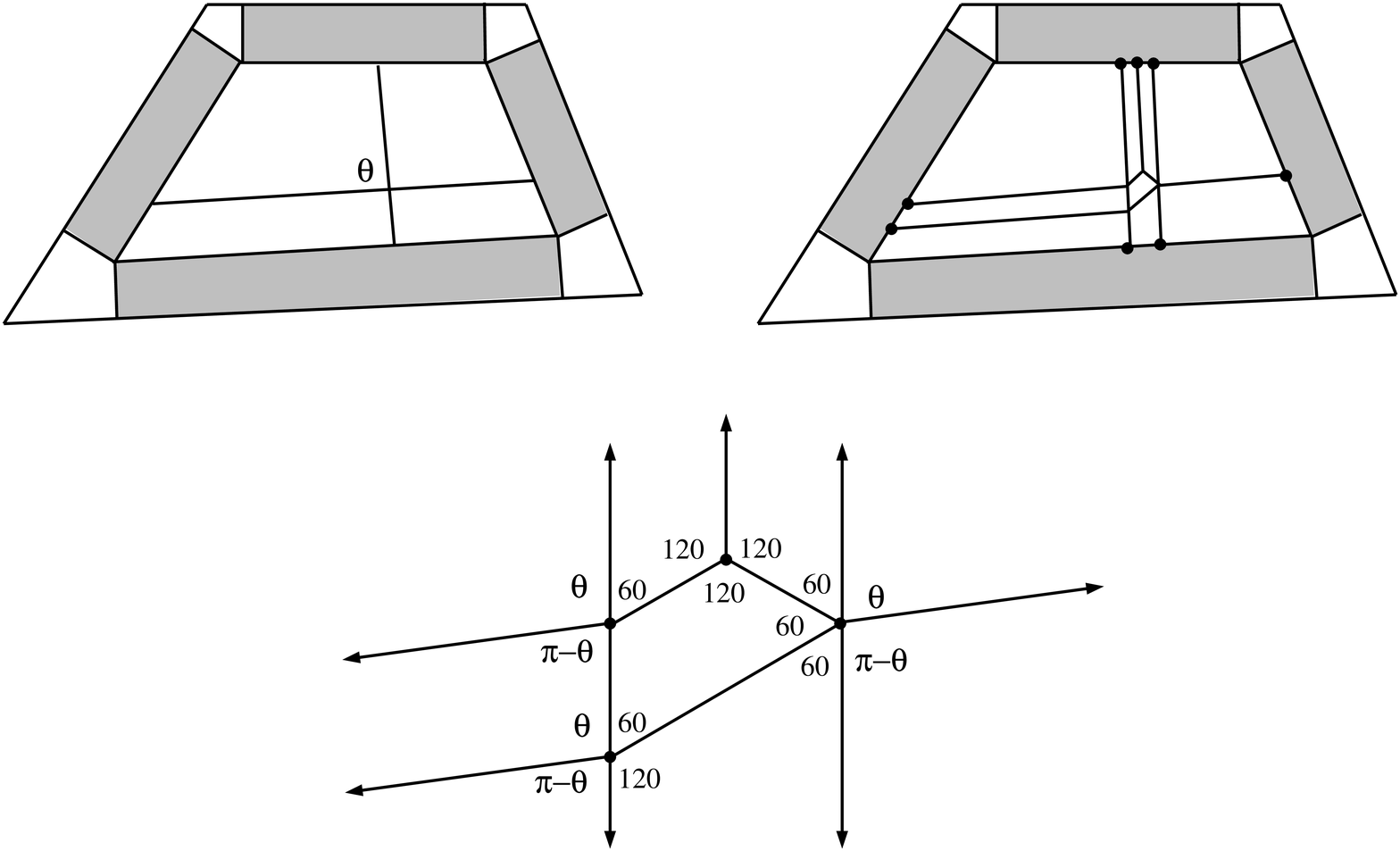}
 }
\caption{ \label{QuadEven2}
The top figures show what to 
do if opposite pairs of rectangles get an odd number 
of new boundary points: we connect the pair by 
a propagation segment to make the number of 
boundary points even again.
The second row shows the construction when 
a pair of adjacent rectangles each get 
an odd number of extra boundary points.
The bottom figure verifies the necessary angle 
bounds.
}
\end{figure}

 The final possibility is that exactly 
two rectangles get an odd number of new boundary points and 
these rectangles touch adjacent sides of $Q$. In this case, we make a 
connection as shown  in the middle
of Figure \ref{QuadEven2}. Here we connect
both pairs of opposite rectangles by propagation 
segments for $Q$, and these segments must 
 cross at some angle $\theta$ that is 
between $60^\circ$ and $120^\circ$.
 We then replace
these segments by parallel segments as shown in
Figure \ref{QuadEven2}. The exact 
angles near the crossing point are show in the lower
part of Figure \ref{QuadEven2}. These angles are 
all between $60^\circ$ and $120^\circ$. After adding 
these segments, all the rectangles have an even 
number of new boundary points and they can then 
be re-meshed.
Thus the quadrilateral with the extra $O(\ecc(Q))$  boundary 
points due to the four rectangular sinks is also a sink.

As with square sinks and rectangular sinks, certain placements
of $M$ extra points   near the corners  of $Q$  
can create 
$\simeq M^2$ quadrilaterals in the  ``corner sub-regions''
of $Q$. However, if we place at most $M_1 \geq 1$ points 
on one pair of opposite sides of $Q$, and at most 
 $M_2 \geq 1$ points on the 
other pair of opposite sides, then each of the four 
sub-rectangles gets (either directly or by propagation 
through a corner region)  at most $M_1$ points on one pair of 
opposite sides and $M_2$ points on the other pair of 
opposite sides. Thus  
the sink  can be re-meshed using 
$O( \ecc(Q) + M_1 M_2)$ elements. 
\end{proof}

\section{The standard sector mesh}  \label{sector sec}

A sector of angle $\theta$ and radius $r$   is a region 
in the plane  isometric to 
$$ S = \{ z = t e^{i\psi}: 0< t < r,  |\psi| < \theta/2\}.$$ 
The point corresponding to the origin is called the vertex of 
the sector.
We shall use the term ``inscribed sector''  to refer to polygons
where the circular arc in $S$ has been replaced by a polygonal arc
with the same endpoints and possibly other vertices placed (in order) 
along the circular arc. See Figure \ref{InscribedSector3}. 
A ``truncated sector''  with inner radius $r_1$ and outer radius 
$r_2$ is a region similar to 
$$ S = \{ z = t e^{i\psi}: r_1 < t < r_2,  |\psi| < \theta/2\}, $$ 
 and we define an inscribed truncated sector in the analogous way. 
See Figure \ref{InscribedSector3}.

\begin{figure}[htbp]
   \centerline{
	   \includegraphics[height=2.25in]{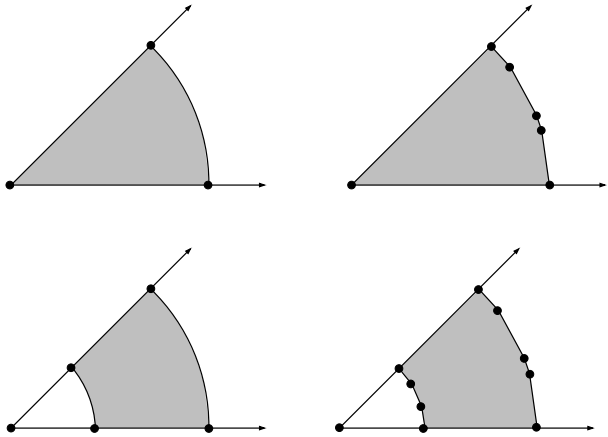}
	   }
\caption{ \label{InscribedSector3} 
The definition of a sector and truncated sector (left top and bottom)
and inscribed sector and inscribed truncated sector (right top and 
bottom).
} 
\end{figure} 

Given a sector of angle $\theta \leq 120^\circ$, 
the ``standard sector mesh'' is illustrated 
in Figure  \ref{thin-vertex}. Note that every angle
(except possibly the  angle $\theta$  at the  
vertex of the sector), 
is between $60^\circ$ and $120^\circ$. This can be 
verified by  the following calculations  ($\theta_1$, 
$\theta_2$ are as labeled in Figure \ref{thin-vertex}):
\begin{eqnarray*}
 && 0 < \theta \leq 120 \\
 &&  60 \leq \theta_1 = 180 - 60 - \frac 12 \theta \leq 120 \\
 && 60 \leq \theta_2 = \frac 12(180-\frac 12 \theta)  \leq 90 
\end{eqnarray*}
Note that the mesh covers an inscribed sector, where the circular arc
is replaced by a polygonal  arc with two segments of equal length.

\begin{figure}[htbp]
   \centerline{ 
	\includegraphics[height=2.0in]{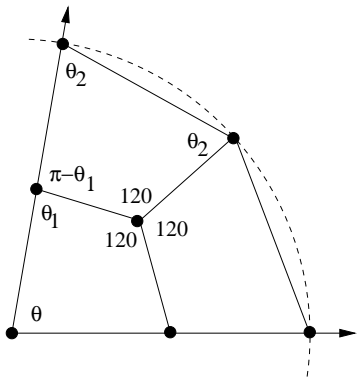}
	}
\caption{ \label{thin-vertex} 
The standard sector mesh. 
If $0 < \theta \leq 120^\circ$
this uses only angles between $60^\circ$ and $120^\circ$, except 
possibly for the angle $\theta$ at the vertex of the sector.
} 
\end{figure} 

We define a similar, but more complicated ``standard
mesh'' of a truncated sector as is illustrated in 
Figure \ref{TrunSecMesh}. Here we assume 
$0^\circ < \theta \leq 60^\circ$.
 In the figure, we take the vertex
of the sector to be the origin, one radial edge along the positive 
real axis and the other along the ray  passing through $0$ and 
$e^{i\theta}$.
Define the points 
$$ z_1 =1, z_2, =2, z_3 =4, z_4=8, z_5 = e^{i \theta},$$
$$ z_8 = 4 e^{i \theta}, z_9 = 8 e^{i \theta}, z_{10} = 
8 e^{i \theta/2}, z_{11} = \frac 34 z_3 + \frac 14 z_8.$$
The point $z_{12}$ is the intersection of the horizontal 
line through $z_{11}$ and the line of slope $\sqrt{3}$ 
through $z_2$. The point $z_6$ is the intersection of the upper radial 
edge of the sector with the line of slope $- \sqrt{3}$,
through $z_2$. The point  $z_7$ is the intersection of the same 
radial edge of the sector with the line of slope $- \sqrt{3}$ 
through $z_{12}$.  

\begin{figure}[htbp]
 \centerline{ 
	\includegraphics[height=3.7in]{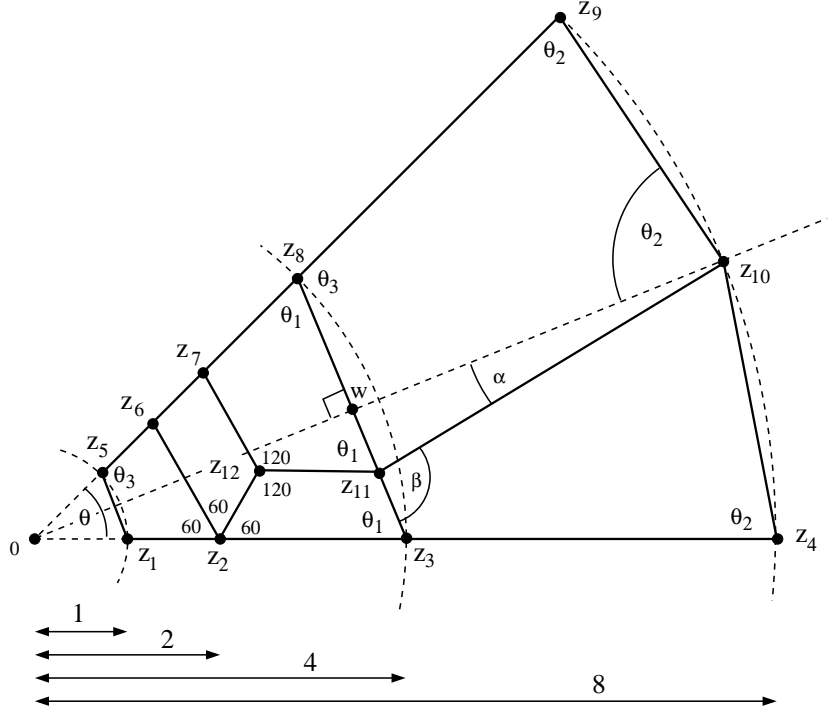}
	}
\caption{ \label{TrunSecMesh} 
The standard  truncated sector mesh.  Angle bounds are
given in the text.
} 
\end{figure} 

It is easy to check that $\theta \leq 60^\circ$ implies 
$$ 60^\circ \leq  \theta_1=90^\circ - \theta/2 \leq 90^\circ,$$
$$ 75^\circ \leq  \theta_2 = 90^\circ - \theta/4 \leq 90^\circ,$$
$$ 90^\circ \leq  \theta_3 = 180^\circ - \theta_1= 90^\circ + \theta/2 \leq 120^\circ.$$
From the diagram, we also see that if $w$ is the intersection of 
the segments $[z_3,z_8]$ and $[0,z_{10}]$, then  
\begin{eqnarray*}
 \tan \alpha 
=  \frac {|w-z_{11}|}{|w-z_{10}|} 
=  \frac { |w-z_3|/2}{|z_{10}|-|w|}
= \frac { (|z_3|/2) \sin \frac \theta2 }{8-|z_3|\cos\frac \theta 2}
= \frac { 2\sin \frac \theta2 }{8-4\cos\frac \theta 2}.
\end{eqnarray*} 
Since $\theta/2 \leq 30^\circ$, this means 
$$ \tan \alpha \leq \frac  1{2(4-\sqrt{3})}
   \approx .220463 < .267949 \approx 2 - \sqrt{3} = \tan 15^\circ.$$
Thus $ 0 \leq \alpha \leq 15^\circ$. Hence the two angles 
at $z_{10}$ are  $\theta_2 \pm \alpha$,  and satisfy
$$75^\circ \leq  \theta_2 \pm \alpha = (90^\circ - \frac \theta 4)
     \pm \alpha \leq 105^\circ.$$
Therefore
\begin{eqnarray*}
 \beta 
&=& 360^\circ - \theta_2 - (180^\circ - \theta_1) -( \theta_2 - \alpha)\\
&=& 180^\circ - 2 \theta_2  + \theta_1  + \alpha \\
&=& 180^\circ - (180^\circ -\theta/2)  + (180^\circ -\theta)/2  + \alpha\\
&=&  90^\circ +\alpha,
\end{eqnarray*}
so 
$  90^\circ \leq \beta  \leq 90^\circ + 15^\circ \leq 105^\circ.$
Clearly, the supplementary angle to $\beta$ also satisfies the 
desired bounds. Finally, we should check that the 
point $z_{12}$ is actually in the lower half of the sector 
as shown in Figure \ref{TrunSecMesh}.    
Note that 
$$ \Im(z_{12}) = \Im(z_{11}) \leq |z_3-z_{11}| 
            = \frac 12 |z_3| \sin \frac \theta 2 
            \leq 2 \tan  \frac \theta 2,$$
and the right hand side is the length of the vertical segment 
connecting $z_2$ to the segment $[0,z_{10}]$ (the dashed ray 
in Figure \ref{TrunSecMesh}). Since $\Re(z_{12}) > \Re(z_2)$, this 
implies $z_{12}$ is below the dashed line, as claimed.

\section{Layered sector meshes} 
 \label{conforming meshes}


In this section we combine the standard meshes for a 
sector and truncated sector to form a ``layered'' sector 
mesh.

First consider two line segments that meet at angle $\theta
 \leq 120^\circ$.
Without loss of generality we can assume these are 
  the radial segments $[0, 1]$ and $[0,e^{i \theta}]$.  
Mesh the sector of radius 1 between 
these lines using the standard sector mesh. 
Then mesh the truncated sector between $[1,8]$ and 
$[e^{i\theta/2}, 8 e^{i\theta/2}]$ using the standard 
truncated sector mesh. Then reflect this mesh over the segment 
$[0,  8 e^{\theta/2}]$. The union of these 
 three components gives a 2-layer  mesh of the 
sector of angle $\theta$. 
The important point is that the vertices on the radial edges 
of the sector are at radii  that are
independent of the angle of the sector. Therefore, if we mesh disjoint 
sectors at the same vertex that share a radial edge, then the 
meshes match up along that radial edge. 
See Figures   \ref{secmesh1layer} and \ref{3layers}.
\begin{figure}[htbp]
   \centerline{
 \includegraphics[height=1.75in]{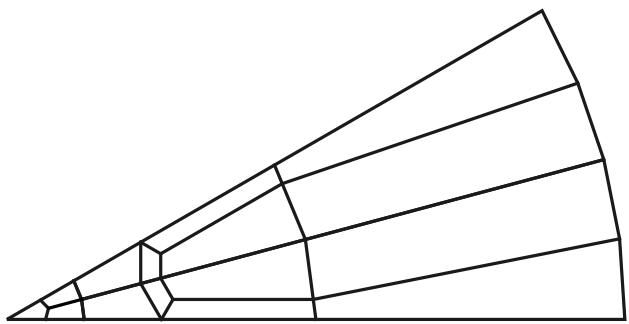}
}
\caption{ \label{secmesh1layer} 
A 2-layer sector mesh.
At the vertex of the sector 
is the standard sector mesh,
 and between radii 1 and 8 are two copies of the standard 
truncated sector mesh for angle $\theta/2$ (they 
are symmetric with respect to the sector bisector). 
} 
\end{figure}

\begin{figure}[htbp]
   \centerline{
 \includegraphics[height=1.5in]{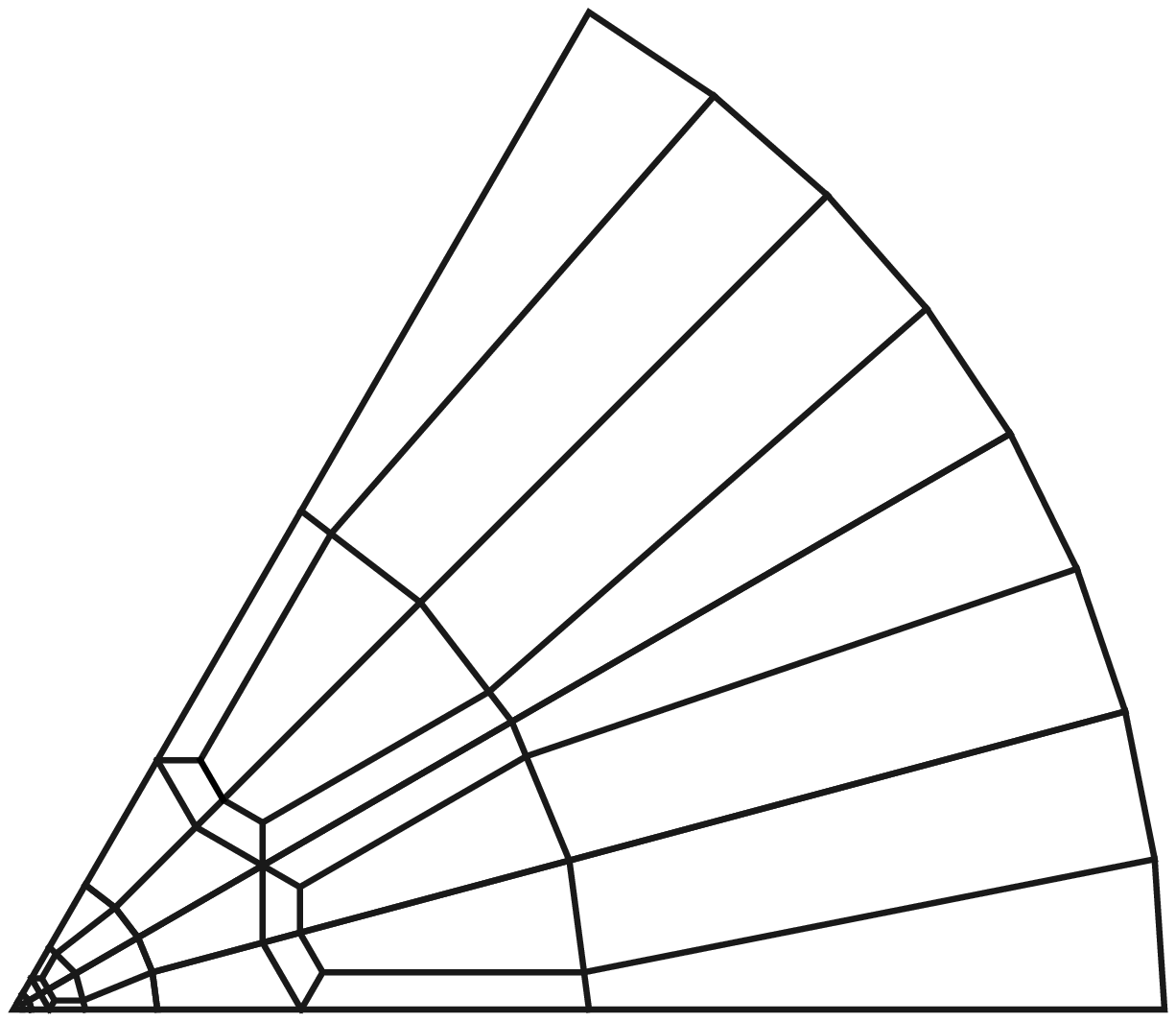}
  $\hphantom{xxx}$
\includegraphics[height=1.5in]{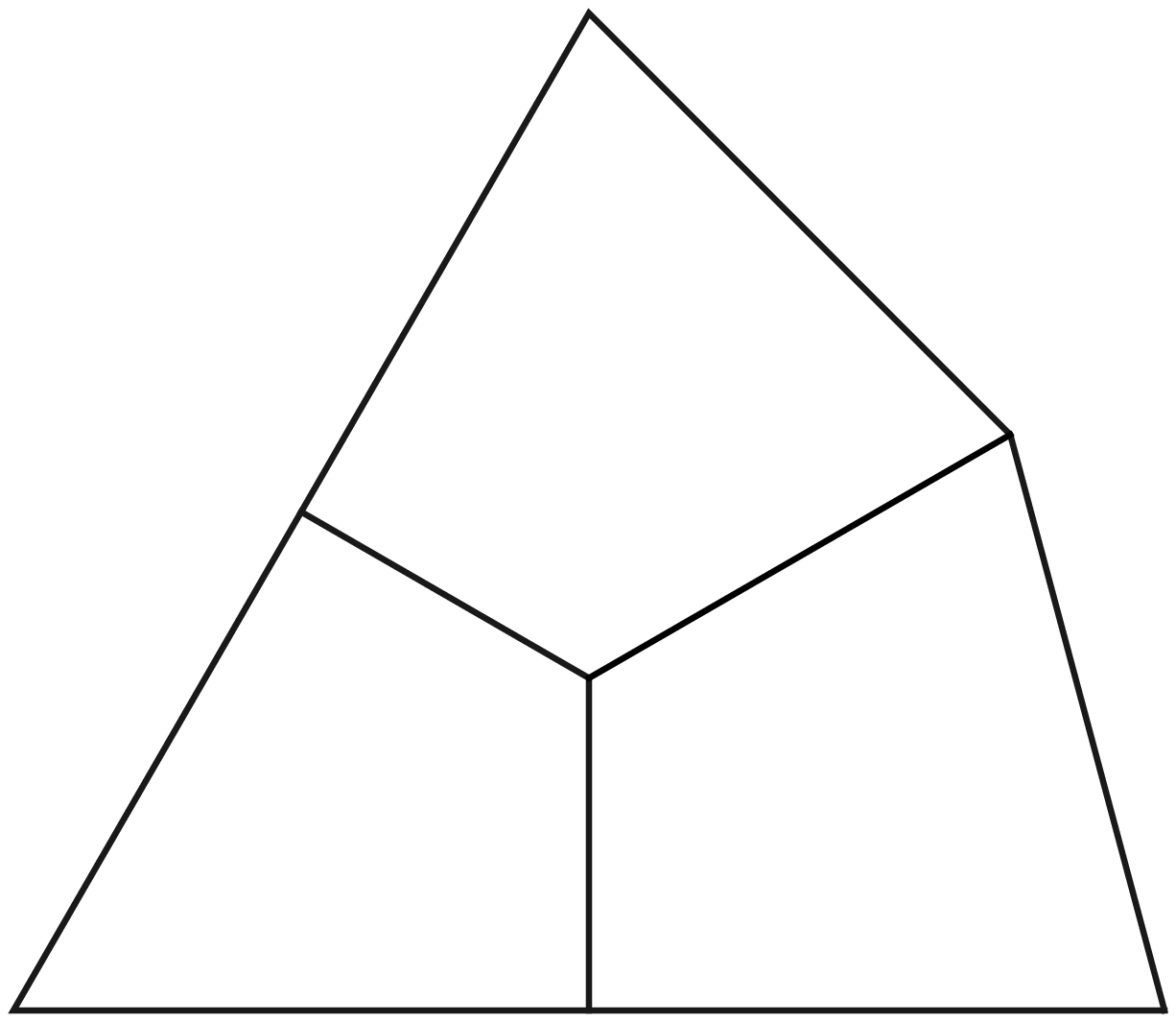}
}
\vskip.2in
\centerline{
\includegraphics[height=1.5in]{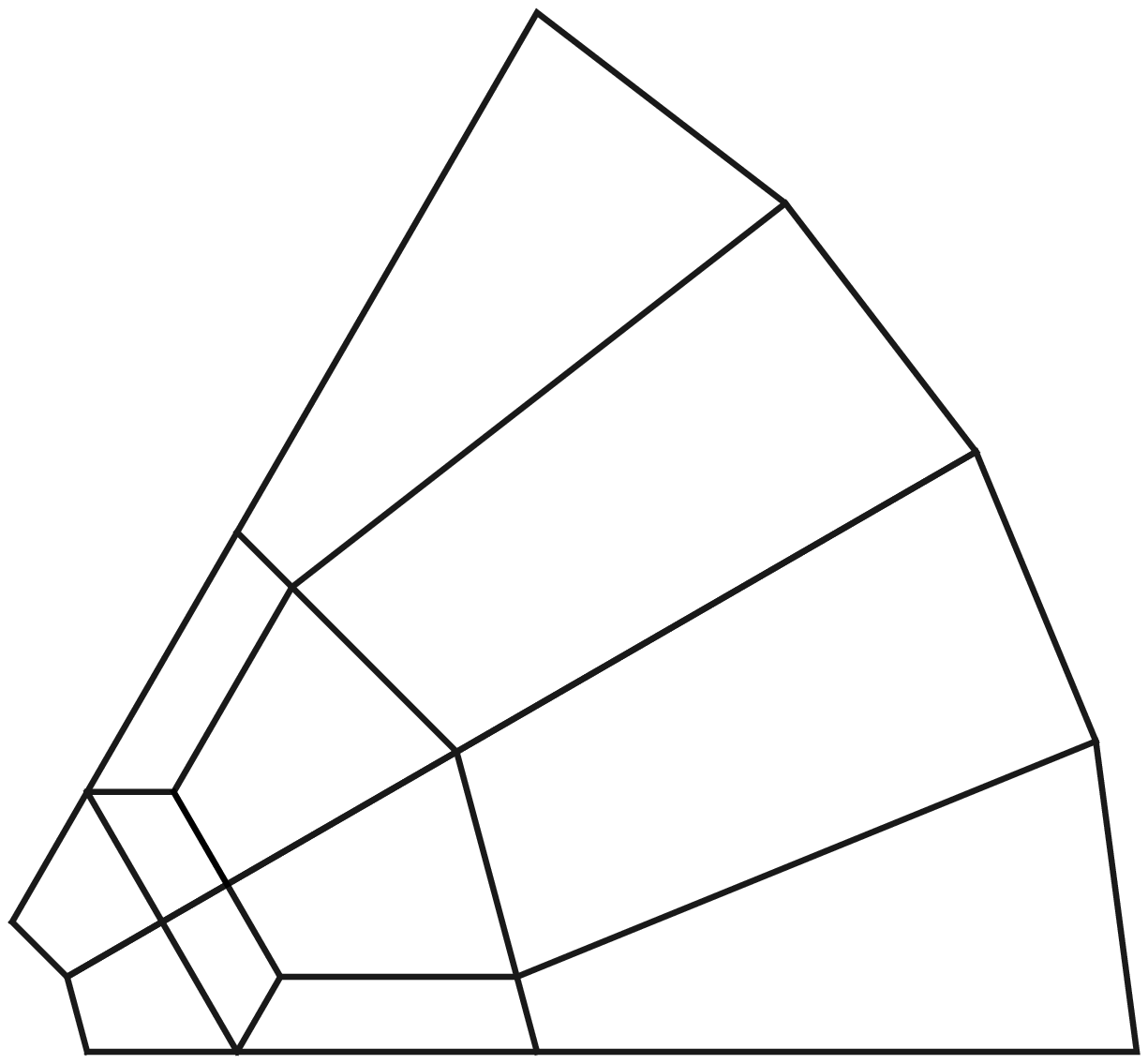}
  $\hphantom{xxx}$
\includegraphics[height=1.5in]{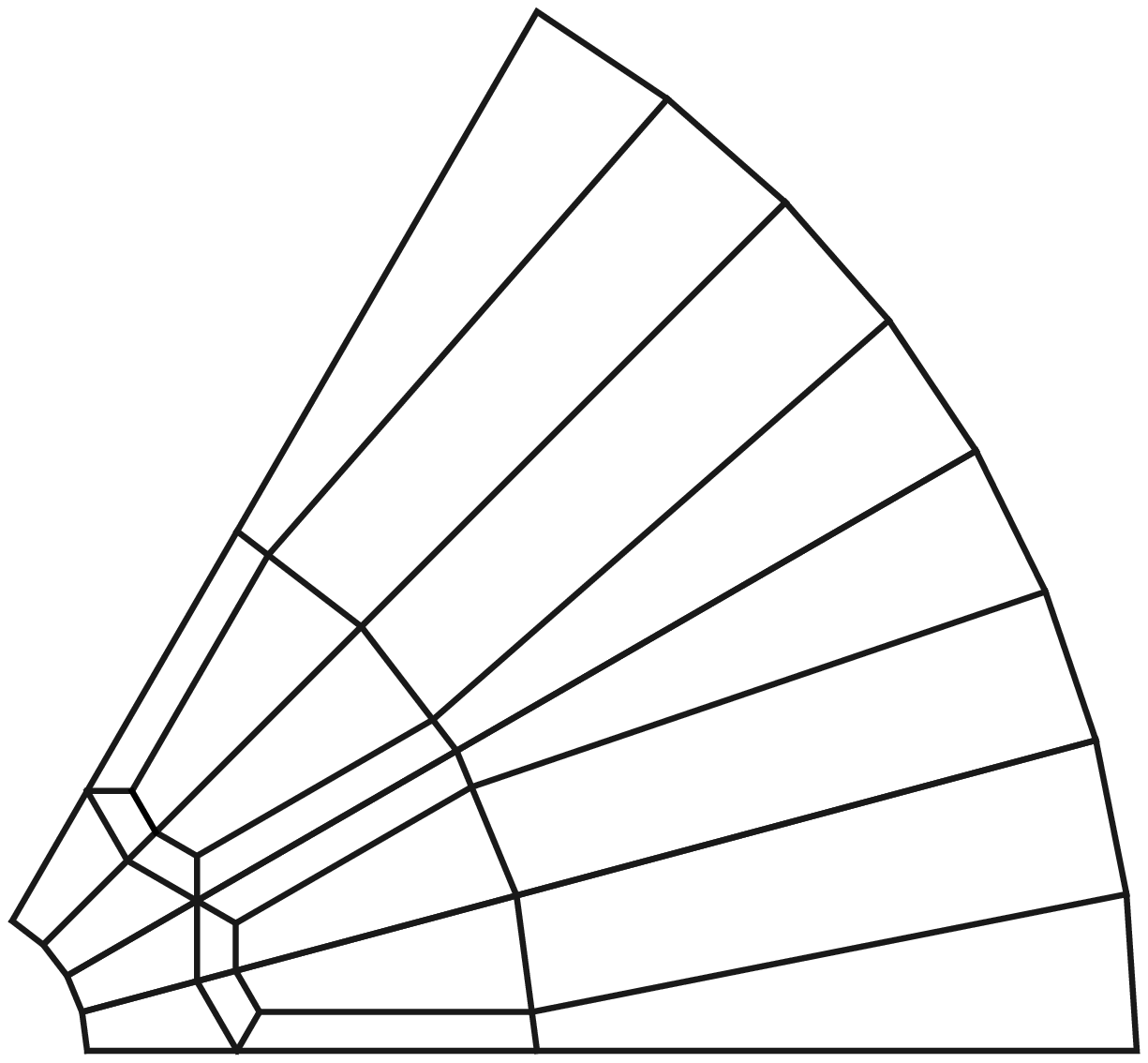}
}
\caption{ \label{3layers} 
A 3-layer sector mesh.
The upper left shows a sector mesh with three layers 
and the other three pictures show rescalings of 
each layer. At top right is  layer 1, the standard sector 
mesh for some angle $\theta$  (this layer is not visible in the picture 
at upper left); at lower left is layer 2, two copies of a standard
truncated sector mesh for angle $\theta/2$; and at lower right is 
layer 3, four copies of a standard truncated mesh for 
angle $\theta/4$. 
More layers could be added to give as many points as desired
on the outer boundary of the sector (the number of 
such points doubles with each layer).
} 
\end{figure} 

By repeating the construction on larger truncated 
sectors we can bisect each of these angles again, and 
continue until the angles are as small as we wish. 
If we start with an $n$-tuple $X =X^0$ on the unit circle, 
the standard mesh  in each sector creates a mesh of 
the cyclic polygon with vertices  $X^1$. Using  
a pair of symmetric truncated sector meshes as above, 
  we quad-mesh the region between the cyclic polygon 
for $ \frac 18X^1$ and $X^2$. In general, using $k$ 
layers we can mesh the cyclic polygon for $X^k$ (this 
was defined in Section \ref{sink defn}, by repeatedly 
bisecting the components of $\circle \setminus X$).

\begin{figure}[htbp]
   \centerline{
    \includegraphics[height=3.0in]{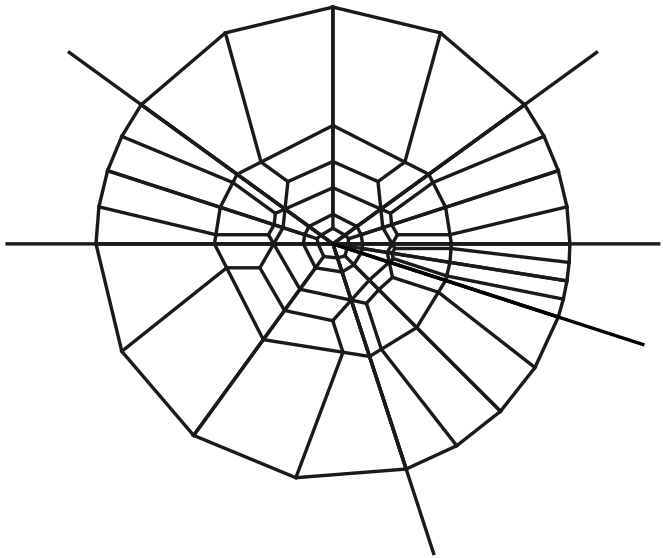}
    }
\caption{ \label{EmptyPSLG} 
Meshing each sector using two layers
 as described in Section \ref{conforming meshes}
 gives a conforming mesh 
of the PSLG consisting of six line segments 
meeting at a point.
} 
\end{figure} 

We now do a similar construction in each given sector.
See Figure \ref{EmptyPSLG} for an example of a two layer 
mesh in six sectors.  The 
mesh points that occur on the given segments occur at 
fixed distances from the origin, independent of the angle, 
so the sector meshes fit together to form a mesh of a polygon 
inscribed in the circle. 
We summarize our conclusions as:

\begin{lemma} \label{protect}
Given any $(d,120^\circ)$-tuple $X$ on the circle
and  a non-negative 
integer $k$,  there is a nice 
conforming  quadrilateral mesh for $X^k$  
that uses $O(2^k d)$ elements and so that the 
edges of the mesh cover each radial segment connecting 
the orgin to a point of $X$.
\end{lemma} 

This mesh will form the center part of   
the protecting sinks constructed next.

\section{Sinks protecting a vertex} 
 \label{conforming sinks}


\begin{proof}[Proof of Lemma \ref{sector conform sink}]
 
Consider  $d$ segments meeting at a single point,
as shown on the left side of 
Figure \ref{Layers1}. Assume the meeting point is the origin and 
the segments all have length $1$.
If any of the angles formed are $\geq 45^\circ$ then we 
add extra radial segments until all the angles are 
$\leq 45^\circ$ (the dashed lines in the top left 
of Figure \ref{Layers1}).  When finished, we have $d +O(1) = O(d)$ 
segments.

\begin{figure}[htbp]
   \centerline{ 
	\includegraphics[height=5.2in]{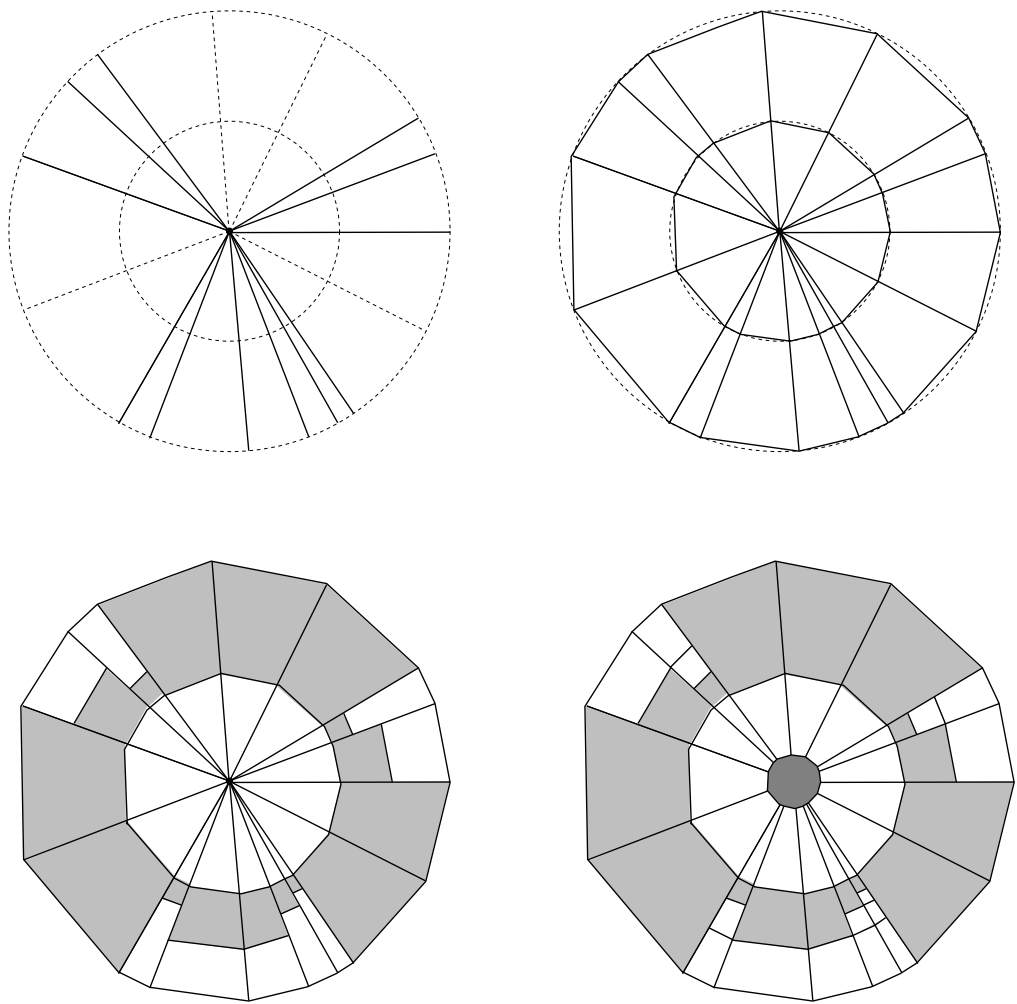}
	}
\caption{ \label{Layers1}
The construction of a vertex protecting sink.
We  first add edges  so every sector 
has angle $\leq 45^\circ$ (upper left). When then 
create cyclic polygons 
$P$ and $\frac 12 P$ (upper right) and place bounded eccentricity 
quadrilaterals $\{Q_j\}$  between these polygons
(shaded region in lower left picture). 
We then propagate the corners of these 
quadrilaterals until they run into another quadrilateral
 (bottom right). 
Inside a much smaller polygon (dark gray, lower right) we place
the  layered sector mesh given by Lemma \ref{protect}.	
The white annular region between this and $\frac 12 P$ is meshed
using Corollary \ref{annular mesh angles}.
We show the case when the sink surrounds the vertex $v$, 
but if $v$ is on the boundary of a face of the PSLG, then 
the picture above  may be  restricted to the sector with 
vertex $v$ that lies inside the face.
}
\end{figure}

Intersect the radial segments with the circles 
of  radii $1$ and $\frac 12$ centered at the origin, 
and inscribe cyclic polygons $P$ and $\frac 12 P$ 
 on each circle using these points.
Let $A$ be the topological annulus bounded by these two 
polygons. 
 The edges of $A$ on $ \frac 12P$  will 
be called the inner edges and the opposite sides on 
$P $ will be called the outer edges; these outer edges will 
be the boundary edges of the sink.

Place  quadrilaterals  $\{ Q_j\}$ in $A$
 as shown by the shaded regions  on 
 the bottom, left  of Figure \ref{Layers1}.  Each quadrilateral has 
one side that is an inner side of $A$ and 
two sides that are sub-segments of the radial segments.
The  fourth side 
is parallel to the first and chosen so that the quadrilateral 
has bounded eccentricity (for large angles, the fourth side 
might be a side of  $P$, but 
for small angles it is in the interior of  $A$.
Later, it will be convenient to assume  
the radial sides of the quadrilaterals $\{Q_j\}$ 
only take certain values, e.g., powers of $2$.
This is easy to do while keeping bounded eccentricity.


Next, propagate any vertices  of the quadrilaterals  $\{Q_j\}$
around the annular region, until  the propagation paths 
 run into another quadrilateral.  
Since there are $O(d)$ quadrilaterals,  and each 
path might travel through $O(d)$ sectors, 
this will create at most $O(d^2)$ new vertices. 

Now place a sink inside each of the quadrilaterals $\{Q_j\}$.
Since all these quadrilaterals have bounded 
eccentricity, there is a  uniform upper bound   
$K$ for the number  of sink vertices that occur on each inner edge. 
Choose an integer $k$ so that $2^{k-1}  <    K  \leq 2^k$ and add
at most $2^k-K$ extra points to each inner edge so 
that the number of vertices on the inner edge of 
each quadrilateral sink is exactly $2^k$. Call the 
resulting polygon $P'$ (this is a subdivision of $\frac 12 P$). 
Let $P''$ be the rounded 
version of $P'$ inscribed on the circle $\{|z|= 1/4\}$.
  Then choose 
$ r = 2^{-s-2} \ll \frac 12$ using Lemma \ref{protect}, and 
 place a copy of a $k$-layer 
conforming mesh for $X$ in the disk of radius $r$. 
The value of $s$ is chosen using 
Corollary \ref{annular mesh angles} so that 
the annular region between the outer boundary 
of the $k$-layer mesh (this is a  cyclic 
polygon on $\{|z| = r\}$ with $2^k d$ vertices) 
and  $P'$ (this is a cyclic polygon with vertices on 
$\{ |z| = \frac 14\}$  with $2^k d$ 
vertices) can be nicely meshed.

At most $2 \pi /\theta$ of 
the sectors in the sinks can have angle $\geq \theta$, 
so at most $O(d/\theta)$ of the quadrilaterals 
created in the mesh of the  protecting sink can fail to 
be $2\theta$-nice. 

If we add $M=M_1 + \dots M_d$ points to the boundary of the 
sink, with $M_j$  points 
being added inside  the $j$th sector,
 these propagate inwards to the quadrilateral 
sinks. The propagation produces at most $O(Md)$ 
elements  in the region 
 outside the ring of quadrilateral sinks, 
and it the adds $M_j$ points to the outer edge 
of the  quadrilateral sink $Q_j$  in the $j$th sector. Suppose 
$Q_j$ had $d_j$ points added to its radial sides by the 
propagation of the corner points around the annular region. 
Since $M$ is even and $\sum_j d_j$ is also even, extra vertices 
can be added (if needed) between adjacent quadrilaterals $\{Q_j\}$
to ensure every sink has an even number of extra vertices (see
Figure \ref{Evenness2} where this was done for rectangles). 
Then  each $Q_j$  can be re-meshed with $O(1+d_j M_j)$ elements
(but see remark after proof).
Summing over $j=1,\dots,d$ shows that at most $O(d + Md)$ mesh 
elements are used in the re-meshing of all these 
quadrilaterals.
The mesh elements created are nice, but we have no
better control on their angles. 

However, we do have some control on angles used 
outside the ring of quadrilateral sinks.
The angles in these new elements depend on the sectors
they belong to; if they all belong to a sector 
with angle greater than $\theta$ then they may 
all fail to be $\theta$-nice so the number of 
non-$\theta$-nice elements can be as large 
as $O(Md + d /\theta)$.
But if we only add $K$ extra  points in each sector, then 
since there are only $ O(1/\theta)$ sectors with 
angle $> \theta$, we will only create $O(Kd/\theta)$
non-$\theta$-nice elements.
\end{proof} 

It is important to note  that  when 
we add extra vertices to the boundary of a
protecting sink, and then re-mesh  the interior
of the sink, the quadrilaterals
touching the center are never changed. All changes 
required by the new vertices take place in 
the  region $A$ between $P$ and $\frac 12 P$.
In the proof of Theorem 
\ref{Quad Mesh}, this  prevents us from subdividing 
any angles that are less than $60^\circ$.

The argument above shows that each quadrilateral $Q_j$ 
is remeshed with at most $O(1+d_j M_j)$  elements, but 
a more careful argument shows only $O(1+d_j +M_j)$ elements
are needed. Since we won't use this stronger estimate, 
we merely sketch the proof. 
Recall that each quadrilateral $Q_j$ is made into a sink 
by taking four rectangular sinks inside $Q_j$, one 
touching each side of $Q_j$ (see Section  \ref{quad sinks sec}). 
 The point is that we have chosen the radial 
sides of the quadrilaterals to be powers of two and 
this implies that when we propagate the corners of  the
quadrilaterals, only $O(1)$ extra vertices are added
to the sub-rectangles of $Q_j$ that  lie along its 
radial sides; the remainder are added to the corner
regions that are adjacent to $\frac 12 P$ and then 
these points 
propagate through the corner region 
 to the sub-rectangle of $Q_j$ that lies 
along $\frac 12 P$. This rectangle does not get any 
extra points via propagation from the boundary  of 
the sink. These observations easily imply the $O(1+d_j+M_j)$ 
estimate.

\section{Isosceles trapezoid dissections}
 \label{dissection}

In this section we recall a result from \cite{Bishop-nonobtuse}
that will be needed in the proof of Theorem \ref{Quad Mesh} to 
deal with high eccentricity quadrilaterals.

The  construction of propagation paths 
for quadrilateral meshes also
works for quadrilateral dissections, although it 
is then possible that a path will terminate at an interior vertex
of the dissection, rather than continuing to the boundary of the 
dissected region. 
Moreover, in contrast to propagation paths in a quadrilateral 
mesh, a propagation path in a quadrilateral 
dissection can return to the same quadrilateral 
arbitrarily often; see Figure \ref{ManyReturns} 
for two such examples.  In many cases, a dissection of 
a polygonal region $\Omega$  can  be
turned into a mesh by propagating  non-conforming  vertices 
through the dissection until they hit 
the boundary  or another non-conforming vertex 
(recall from Section \ref{PSLG defn} that a non-conforming 
vertex is a corner of one quad-shaped piece of 
the dissection, but an interior edge vertex 
of another piece).
 In general, 
there is no bound on how many times such a 
propagation path can revisit the same piece, 
and hence no bound  the number of mesh 
elements generated in terms of the number 
of dissection elements. 

\begin{figure}[htbp]
\centerline{
 \includegraphics[height=2.5in]{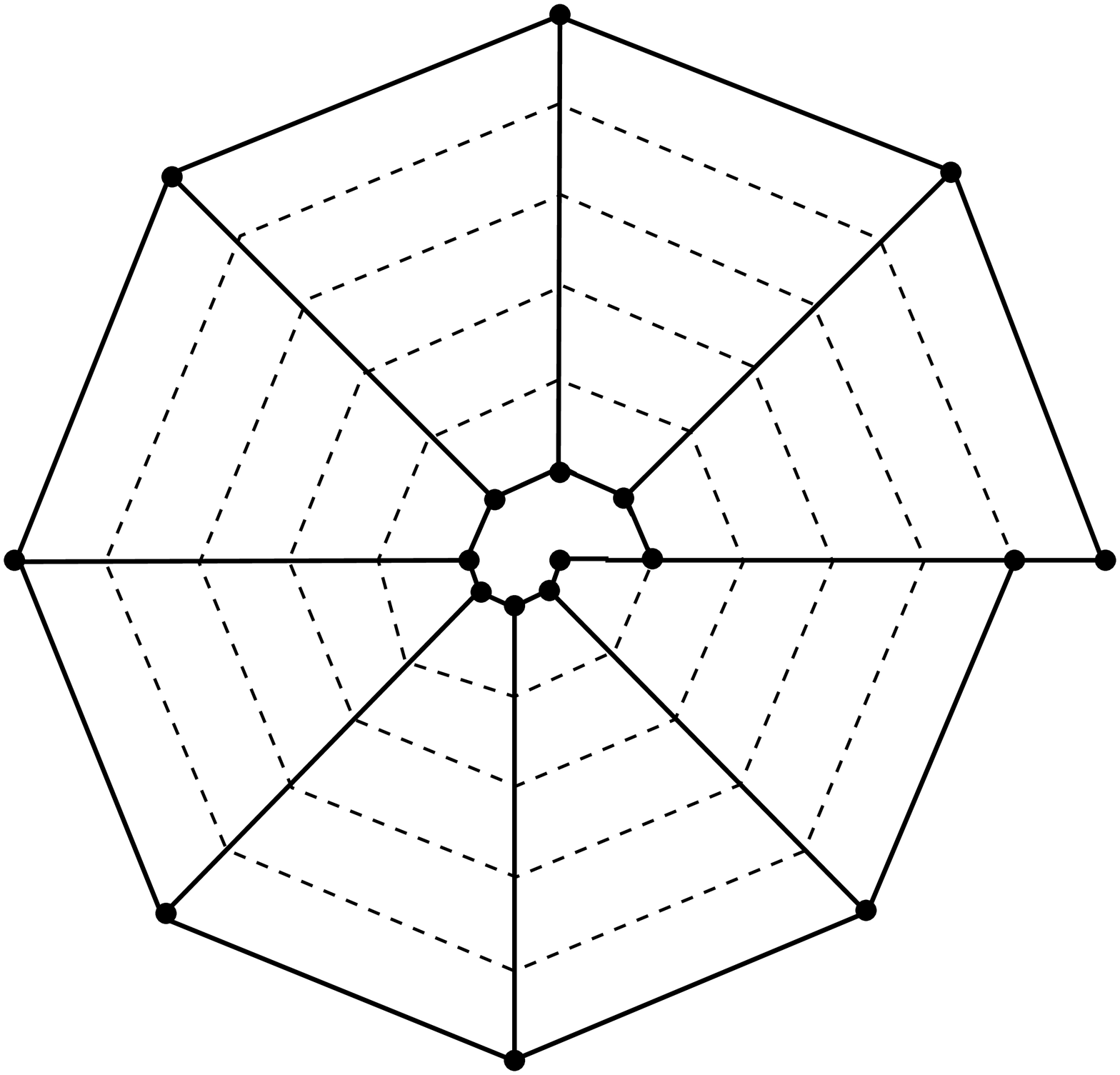}
  $\hphantom{xxx}$
 \includegraphics[height=2.5in]{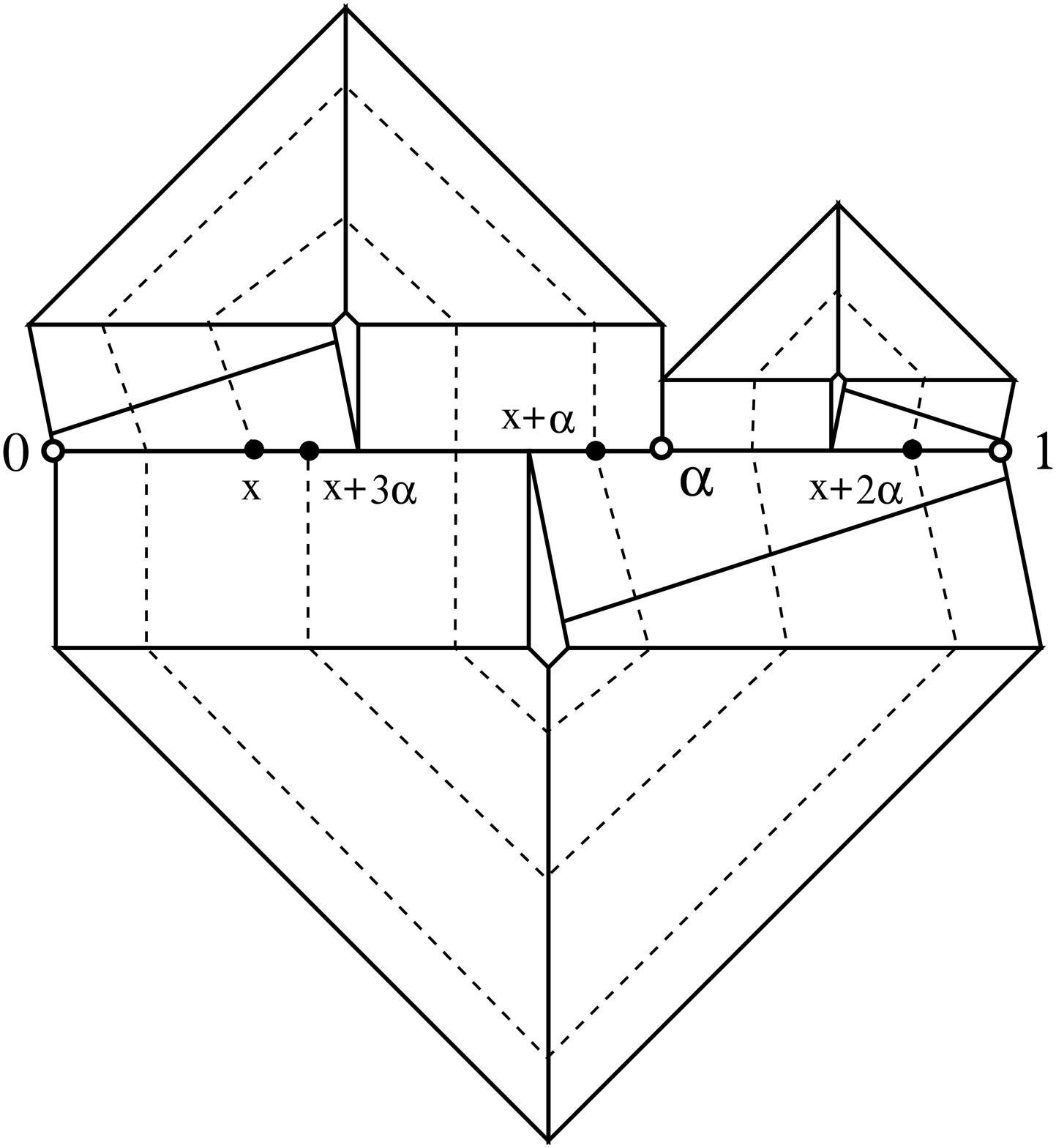}
 }
\caption{ \label{ManyReturns}
 Unlike in a mesh,
 a propagation path in a   dissection 
can  cross   the same quadrilateral many times.
On a left is a spiral where this happens 
finitely often.
  In the picture on the right, it is 
straightforward to check that  
the propagation   path starting at a  point $x \in (0,1)$
will reach $x+\alpha \mod 1$ on its second crossing of 
$(0,1)$ (three such iterates are shown as the dashed 
line connecting black dots). 
If $ \frac 12 < \alpha < 1$ is irrational, this 
implies that each propagation path will be dense in the 
dissected region. Thus propagation paths can revisit 
a trapezoid infinitely often.
}
\end{figure}

 However, in \cite{Bishop-nonobtuse}  I proved there is
such a bound if we ``bend'' propagation
paths  and we introduce
some holes in the domain that the bent paths can 
run in to.
We state the result more carefully below, but first we
give some definitions and notation.

For our application, it suffices to consider a
very special kind of quadrilateral dissection.
An {\defit isosceles trapezoid}
 is  a quadrilateral that is symmetric with 
respect to the line that bisects two opposite sides (called
the  bases of the trapezoid).
Propagation  segments in the trapezoid parallel to the 
base side will be called {\defit $P$-segments} (``P'' 
for {\bf p}arallel or {\bf p}ropagation). 
A {\defit $P$-path} is a 
polygonal path made up of $P$-segments.   The two base 
sides will also be called {\defit $P$-sides} of the trapezoid.
 Propagation 
segments connecting the base  sides will be called 
{\defit $Q$-segments}
and the two non-base sides of the trapezoid will be 
called the {\defit $Q$-sides}.

An {\defit  isosceles trapezoid
 dissection} is a dissection of a region $W$  into 
simple polygons (or pieces) such that 
\newline 
(1)   all the polygons  have the shape
of  an isosceles trapezoid, and 
\newline 
(2)  any  two trapezoidal pieces can only meet
along  boundary segments of  the same type (along two $Q$-sides
or two $P$-sides).
\newline
The second condition implies   we can keep adding   $P$-segments 
to form a $P$-path until the path 
either hits a non-conforming vertex or leaves the region 
covered by the dissection. 

The domains in Figures \ref{ManyReturns}
and  \ref{PropPath4}  are drawn with 
 isosceles trapezoid dissections.
We say that the dissection is {\defit $\theta$-nice} if all 
the pieces are $\theta$-nice.
If $W$ is a region with an isosceles trapezoid dissection, 
then every boundary edge of $W$ is  a subset of 
either a $P$-side or a $Q$-side of some trapezoid in
the dissection. We call the union of the former edges 
the {\defit $P$-boundary} of $W$ and the union of the latter 
edges the {\defit $Q$-boundary} of $W$.

\begin{figure}[htbp]
\centerline{
 \includegraphics[height=2.5in]{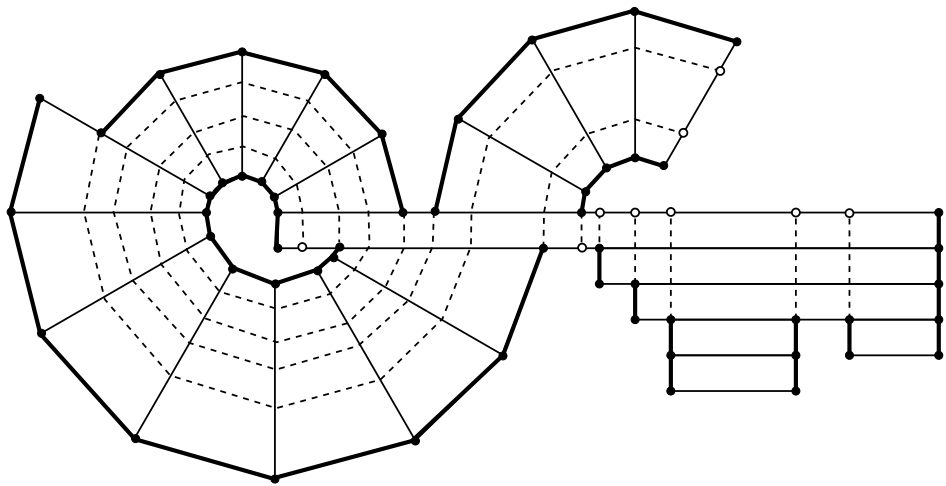}
 }
\caption{ \label{PropPath4}
A domain $W$ dissected by isosceles trapezoids and a
quadrilateral mesh constructed by propagating non-conforming 
vertices along $P$-paths
(the $P$-sides of the isosceles trapezoids are drawn with a thicker
line).
 However, in general,
we need to allow edges that are only ``close to parallel''
to the base edges. It is proven in \cite{Bishop-nonobtuse},
that if we allow $\theta$ deviation from parallel, then
we can always  mesh with $O(n^2/ \theta^{2})$  
quadrilaterals that are $2 \theta$-nice.
}
\end{figure}

A {\defit chain}  in a dissection is a maximal collection 
of distinct trapezoids $T_1, \dots, T_k$ so that for
 $j=1, \dots k-1$, 
$T_j$ and $T_{j+1}$ share a $Q$-side (i.e., the $Q$-sides 
are identical, not just overlapping; in other words, the 
chain of trapezoids forms a quad-mesh of their union).
If a trapezoid in the dissection does not share a 
$Q$-side with any other trapezoid, we consider it as a chain 
of length one.
 For example,
the dissection on the left of Figure \ref{ManyReturns} consists
of a single chain of length 8.
The dissection on the right of Figure \ref{ManyReturns} consists
of a  three chains of length 5 (each is similar to the others).
 The dissection in 
Figure \ref{PropPath4} has chains of length 2, 4, 5 and 7 
and four chains of length 1.
The ends of a chain  consists of two segments: the  $Q$-side of $T_1$
 that is not shared with $T_2$,  and the 
 $Q$-side of $T_k$ that is not shared with $T_{k-1}$. 
The corners of a chain are the endpoints of the ends
of the chain.
Usually there are four distinct corners, but there may 
only be two in the case when $T_1$ and $T_k$ also share
a $Q$-side. In this case,  the chain forms
 a closed loop,  and we will 
call this a closed chain. This case 
is not of much interest to us since no propagation 
paths will ever occur inside a closed chain.

The precise result that we will use is  Lemma 11.1 of
 \cite{Bishop-nonobtuse}:

\begin{thm} \label{quad mesh lemma}
Suppose
 that $W$ is a polygonal  domain  with an isosceles trapezoid
 dissection with $n$  pieces.
Suppose also that   $0^\circ \leq \theta \leq 15^\circ$ and
 that  every dissection  piece is  $\theta$-nice. Finally, suppose
 the number of chains in the dissection is  $M$.
Then we can remove
$O(M/\theta)$ $\theta$-nice quadrilaterals of uniformly bounded
eccentricity from $W$ so that
the remaining region $W'$  has a $2\theta$-nice quadrilateral mesh
with $O( n M / \theta)$ elements.
At most $O(M/\theta)$ new vertices are created on the
$Q$-boundary of $W'$.
At most $O(M)$ vertices are created on the $P$-boundary
of $W'$, and no more than $O(1)$ vertices are placed in
any single $P$-side of dissection piece of $W'$.
For this quad-mesh,
any boundary point  on a $Q$-side of $W'$  propagates to another
boundary point after crossing at most  $O(n)$ quadrilaterals.
\end{thm}

Very briefly,  the proof in \cite{Bishop-nonobtuse} 
creates  $O(M)$ sub-regions of $W$ called return regions that
are chains of isosceles trapezoids with the property 
that any  $P$-path crossing $5n+1$ trapezoids  must  hit
 at least one return region. Each  return region 
is sub-divided  by $P$-paths into at most $O(1/\theta)$ parallel
chains, called tubes, that are each 
approximately  $1/\theta$ times ``longer'' than they 
are ``wide'' (this is made  precise  in \cite{Bishop-nonobtuse}).
 This property  implies that  
we can place two quadrilaterals inside each tube with 
the property that any
propagation path entering either end of tube
can be $\theta$-bent so that it hits a $Q$-side of 
 one of the two quadrilaterals.
See Figure \ref{ReturnRegion1}.
 This terminates 
the path.  The $O(M/\theta)$  quadrilaterals  inside the tubes 
are the ones that are removed from $W$ to give $W'$.
There are $O(M/\theta)$ tubes and the corners of these 
tubes must be $P$-propagated until the hit $\partial W'$; 
this gives the $O(M/\theta)$ boundary vertices that 
may be created on the $Q$-sides of $W'$. The corners of 
the $O(M/\theta)$ removed quadrilaterals also have to propagated 
as $Q$-paths, giving points on the $P$-sides of $W'$, 
but by choosing the quadrilaterals appropriately we 
can arrange for only  $O(M)$  vertices to be created 
on the $P$-sides of $W$, and at most  $O(1)$
vertices  are added   to any single  $P$-side of $W'$.
Since these  details of the proof Theorem \ref{quad mesh lemma} 
 are not needed to  apply the 
result, we refer the reader to \cite{Bishop-nonobtuse} for 
further information. 

\begin{figure}[htb]
\centerline{
\includegraphics[height=2.25in]{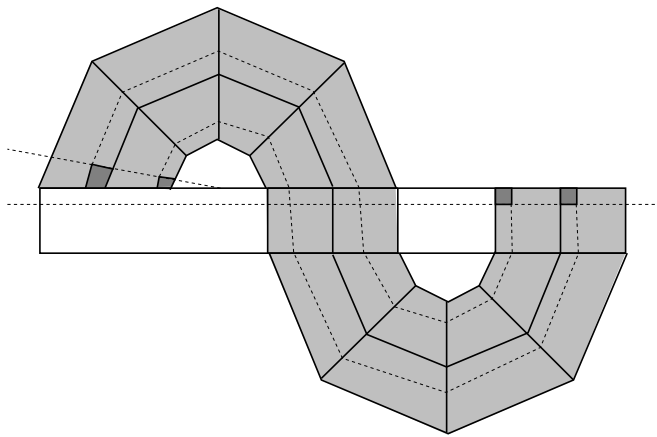}
 }
\caption{\label{ReturnRegion1}
An example of a return region divided into 
two tubes  and of  quadrilaterals placed 
at opposite  ends of these tubes. The corners
of these quadrilaterals  are connected by a path (dashed)
in the tube, dividing it into $2\theta$-nice  quadrilaterals.
Every path entering the tube will either hit one 
of the quadrilaterals immediately, or propagate to 
hit the quadrilateral at the other end of the tube.
}
\end{figure}

\section{Thick/thin decompositions and meshing a simple polygon }
 \label{Thick Thin sec}

In this section, we review some 
facts  from \cite{Bishop-optimal} about meshing a simple polygon.
Most of the discussion is  background that motivates 
the result, but is not needed to apply the result. Theorem 
\ref{simple quad mesh}, at the end of the section, 
 will contain the precise statements that we will use. 
  
The papers \cite{Bishop-time} and \cite{Bishop-optimal} make use 
of a decomposition of a polygon into pieces called 
the thick and thin parts. The name is motivated
by the thick/thin decomposition of a  hyperbolic Riemann surface;
 the $\epsilon$-thin parts of a surface   are the 
points that lie on a   closed curve  of hyperbolic  length  $< \epsilon$ that 
is not deformable to a point (i.e., a homotopically non-trivial 
loop). A famous theorem of Margulis says 
that if $S$ is a  hyperbolic Riemann surface, 
then there is an $\epsilon_0 >0$ (independent of the surface), 
so that for $\epsilon < \epsilon_0$, the $\epsilon$-thin 
parts are all disjoint. Moreover, each  thin part
is either a punctured disk
(curves around the puncture can be deformed to loops of 
arbitrarily short length) or a cylinder (there is a positive 
lower bound for the length of any non-trivial loop in  such 
a thin part).
See Figure  \ref{surfaces}.
 These two cases are called parabolic and hyperbolic 
thin parts respectively.
The names come from the classification of  the
group elements  that represent these loops when the surface 
is represented as the hyperbolic unit disk 
quotiented by a discrete group of M{\"o}bius transformations.

\begin{figure}[htbp]
\centerline{
 \includegraphics[height=2.0in]{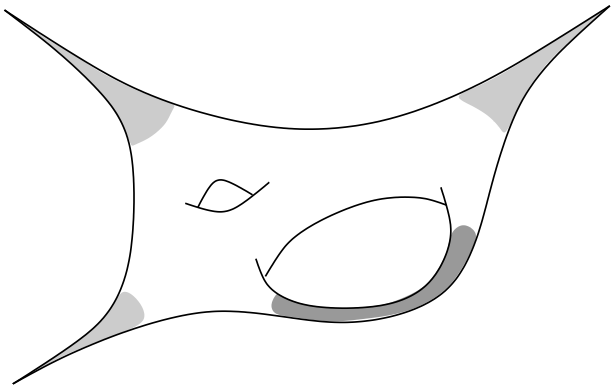}
}
\caption { \label{surfaces}
A surface with one hyperbolic thin part
  (darker) and three parabolic thin parts (lighter).
   }
\end{figure}

The precise definition of the thin part of a polygon is 
given in Section 12 of \cite{Bishop-time}; the 
construction of the thick and thin parts of 
a polygon is somewhat involved 
and makes use of an approximation of the conformal map from the 
interior of the polygon to the unit disk. For readers  who
know the terminology, each thin part corresponds to a
 pair of edges of $P$ whose extremal distance inside  
$P$ is less than $\epsilon$.
Briefly, two sides $e,f$  of a polygon have extremal 
distance $< \epsilon$ if and only if  there is a non-negative 
function $\rho$ on the interior of $P$ so that:
\newline   (1)
$\int_\gamma \rho ds \geq 1$ for any curve 
in the interior connecting $e$ to $f$, and
\newline   (2) 
$\iint \rho^2 dxdy < \epsilon$.
\newline
Extremal distance is a conformal invariant that plays
a fundamental role in modern complex function 
theory and dynamics; for its basic properties 
 see \cite{Ahlfors-QCbook}, \cite{Garnett-Marshall}. 
We will not need to understand  it to use 
the results from \cite{Bishop-time} and
\cite{Bishop-optimal}, but we will use the following 
fact: if $e$ and $f$ are sides of a polygon $P$ that have 
small extremal distance in $P$, then they can be 
joined inside $P$ by a curve $\gamma$ whose length 
is much less than $\min(\diam(e), \diam(f))$. The 
converse need not be true, i.e., a relatively short 
joining curve does not imply the extremal distance 
is small. However, two sides of $P$ that touch at a
common vertex must have extremal distance zero; 
thus there will be a  thin part associated to every 
pair of adjacent edges (parabolic thin parts),
 and  possibly other thin parts   associated to non-adjacent 
pairs of edges (hyperbolic thin parts).

A {\defit cross-cut} of a domain is a Jordan arc in the 
domain that has both its endpoints on the boundary 
of the domain. If the domain is bounded by a Jordan curve, a cross-cut 
$\gamma$ cuts the domain into two pieces that have 
$\gamma$ as the intersection of their boundaries.
The thick/thin decomposition of a polygon is a 
special  partition of its interior by cross-cuts.
A parabolic thin part is  bounded by two sub-segments 
of $P$ meeting at vertex $v$  of $P$ and by a single 
circular cross-cut (a subset of the circle centered
at $v$). The diameter of the piece is small compared 
to the distance from $v$ to the nearest distinct 
vertex or non-adjacent edge of $P$ (as measured in the internal 
path metric).  See Figure \ref{ThickThinParts2}. 

A hyperbolic thin part  is bounded by two segments on 
non-adjacent sides of $P$ and two disjoint cross-cuts that 
connect the endpoints of these segments. The segments 
have the same length, and the internal path distance 
between them is much smaller than this length. 
The cross-cuts are either both  concentric circular arcs, 
or each is the union of two circular arcs. 
See Figure \ref{ThickThinParts2}. 
The thick parts are everything  else. Thin parts can 
only share a cross-cut boundary with a thick part (never an 
another thin part).

\begin{figure}[htbp]
   \centerline{ 
	\includegraphics[height=4.0in]{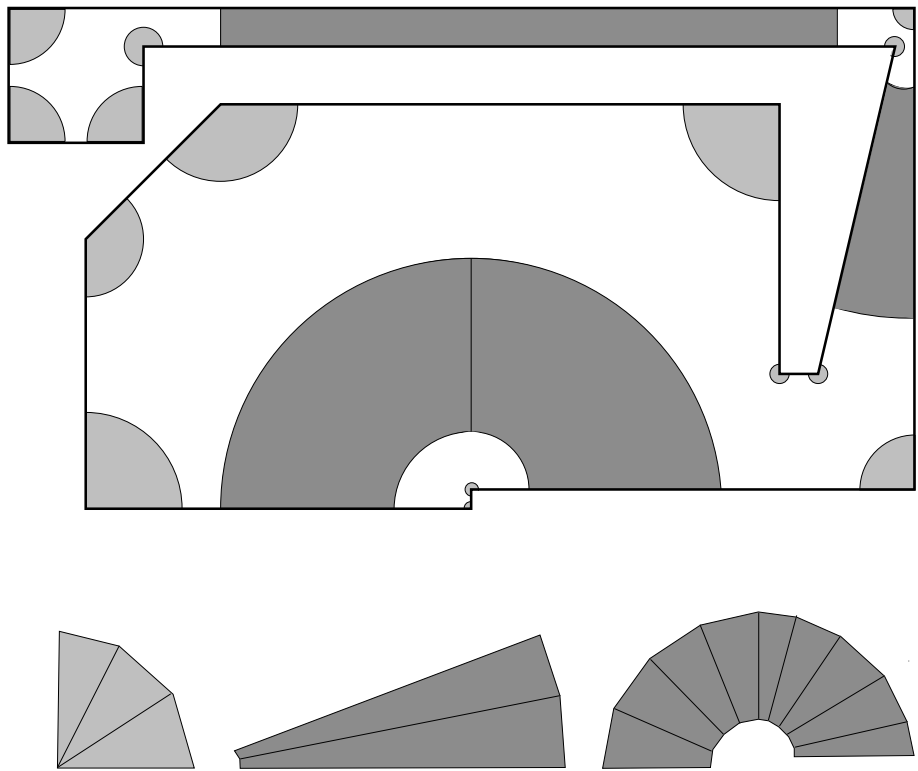}
	}
\caption{ \label{ThickThinParts2}
Parabolic thin parts are shown in light gray, hyperbolic thin 
parts in dark gray and the thick part is the remaining white 
part of the polygon.  
It is also possible to inscribe polygonal cross-cuts
into these arcs as illustrated on the bottom.
}
\end{figure}

  We can  replace the circular cross-cuts by 
inscribing polygonal arcs into these
cross-cuts  to get polygonal thick and thin parts. This 
gives thin parts that are simple polygons and it 
is   this
type of thin part that we want to use in the proof of Theorem \ref{Quad Mesh}.
In fact,  \cite{Bishop-time} shows that by 
inscribing $O(1/\theta)$   vertices  on each 
cross-cut, we can ensure that each polygonal hyperbolic 
thin part is meshed by a chain of $\theta$-nice  isosceles  trapezoids.
However, the eccentricity of these
trapezoids may be arbitrarily large, depending on $P$.

In \cite{Bishop-optimal} the thick/thin decomposition is 
applied to quad-meshing as follows. The thin parts come in 
a small number 
of very simple shapes and all of these are meshed ``by hand''.
The thick parts have complex geometry, 
but the conformal map of a polygon  to the unit disk is very 
well behaved on the thick parts (i.e., there are 
good estimates of the distortion) and it can be used 
to transfer nice meshes on the disk (constructed using 
hyperbolic geometry)  back to 
a  nice mesh of the thick part
that agrees with the meshes of the thin parts on the common 
cross-cut boundaries.

The following result includes various facts that 
are needed in the proof of Theorem \ref{Quad Mesh}.
The result is not stated like this in \cite{Bishop-optimal}, but 
all the claims made here are established as part of the proofs
in \cite{Bishop-time} and \cite{Bishop-optimal}. 

\begin{thm} \label{simple quad mesh}
Suppose $P$ is a simple polygon with $n$ vertices.
Then $P$ has a decomposition into $O(n)$ 
 simple polygonal pieces 
of three types: thick pieces, parabolic thin pieces and 
hyperbolic thin pieces. The total  number of edges used 
in $O(n)$.  Suppose $\theta >0$.
Then $P$ has a quadrilateral 
mesh  with $O(n/\theta^2)$ elements
that sub-divides the thick/thin decomposition. 
The mesh quadrilaterals in the thick parts will be 
     called {\defit thick quadrilaterals}. Similarly for  the
     {\defit parabolic quadrilaterals} and {\defit hyperbolic
     quadrilaterals}.  
The quad-mesh restricted to each type of piece satisfies 
the following properties:

\noindent 
{\bf Thick pieces:} 
     The interior angles 
     of the thick pieces are all either in $[90^\circ, 90^\circ + \theta]$
     or in $[180^\circ, 180^\circ + \theta]$.
      If a thick piece has $k$ cross-cuts boundary arcs,
      then  the quadrilateral mesh in this piece 
     has $O(k/\theta^2)$ elements and $O(k/\theta)$ vertices 
     occur on the boundary of the thick piece. The eccentricity 
     of every quadrilateral mesh element in a thick piece is bounded by a 
     constant $M$, independent of anything else.
     All the mesh elements are nice (angles between $60^\circ$    
      and $120^\circ$). 
     Each quadrilateral mesh element in a  thick piece has  at most one 
     of its sides on $P$.

\noindent
     {\bf Parabolic thin pieces:}
       A parabolic thin piece contains exactly one vertex
      $v$ of $P$. It is bounded 
       by two sub-segments of  $P$ and one polygonal cross-cut with
      $O(1/\theta)$ elements. The mesh in this 
      piece  has $O(\theta^{-2})$ elements and $O(\theta^{-1})$ vertices 
      on the boundary of the thin piece. Every angle of the mesh,
     except possibly one, has angle measure  
       in $[60^\circ, 120^\circ]$. The 
      possible exception occurs if the vertex $v$ 
      has interior angle $ \alpha < 60^\circ$ in $P$. In that case, 
      exactly one  quadrilateral $Q$  in the mesh of the thin part
           has $v$ as a vertex
      and the angle of $Q$ at $v$ is $\alpha$. This quadrilateral is 
      a kite with eccentricity $O(1/\alpha)$; all other quadrilaterals 
      in the mesh of the parabolic thin piece have eccentricity
       that is  uniformly 
      bounded. 
     Each parabolic  quadrilateral  $Q$ that does not contain a 
     vertex $v$ of $P$  has at most one 
     of its sides on $P$.  If $P$ has 
     angle $\leq 120^\circ$ at  $v$ then the quadrilateral $Q$  having 
      $v$ as a vertex has  two sides on $P$. If the angle of $P$ 
     at $v$ is greater than $120^\circ$ then the angle is subdivided
     and there are two or more quadrilateral mesh elements that have
    $v$ for a vertex,  and these all have at most one side  on $P$.
 
\noindent 
    {\bf Hyperbolic thin pieces:}
     Each hyperbolic thin part is bounded by two
       sub-segments $S_1, S_2$ of non-adjacent 
      sides of $P$. The  endpoints of $S_1, S_2$ 
        are connected by two polygonal cross-cuts 
      with $O(1/\theta)$ edges each.  The mesh of the thin part is 
      a chain of $O(1/\theta)$  isosceles trapezoids.
      The ends of the chain are the segments $S_1, S_2$.
      Each $P$-side of each trapezoid is also the side of a 
      thick  quadrilateral  mesh element.
\end{thm}

\section{The proof of Theorem \ref{Quad Mesh} } \label{Proof}

We can now give the proof of Theorem \ref{Quad Mesh}.
We will actually prove the following result that implies 
Theorem \ref{Quad Mesh} if we  set $\theta = 15^\circ$ (the 
proof simplifies slightly in this case and the reader may wish to 
first read it with this in mind):

\begin{thm} \label{near 90}
Suppose $\Gamma $ is a PSLG  with $n$ vertices
and suppose $ 0 < \theta  < 90^\circ$.
 Then there is a  nice conforming quadrilateral mesh of 
$\Gamma$ such that 
\begin{enumerate}
\item  every angle of every quadrilateral is 
$\leq 120^\circ$,
\item  every  new mesh  angle is $\geq 60^\circ$, i.e., 
all angles are $\geq 60^\circ$  
unless the angle occurs at a mesh vertex that 
is a vertex of $\Gamma$ 
and the edges of the corresponding  quadrilateral lie on edges 
of $\Gamma$ that make an angle $\alpha < 60^\circ$ with each other.
In this case,
the angle of the mesh element there is also $\alpha$ (small 
angles of the PSLG are not sub-divided),
\item   the mesh uses $O(n^2/ \theta^{2})$
quadrilaterals,
\item all the mesh angles actually lie  in the smaller interval 
    from $90^\circ - 2\theta$    to 
    $90^\circ + 2\theta$ with at most $O(n/\theta^2)$ exceptions.
\end{enumerate}
\end{thm}

Taking $\theta$ small, we see that 
in the worst case there are quadratically many 
mesh elements, but they are all close to rectangles with 
only a linear number of exceptions.
Taking  $\theta = 15^\circ$ we see that 
 a PSLG with $n$ vertices has a conforming mesh with $O(n^2)$ 
elements that  are all nice with at most $O(n)$ exceptions 
(one for each  interior 
angle of measure less than $60^\circ$  of some 
face of  $\Gamma$.

\begin{proof}[Proof of Theorem \ref{near 90}] 
We give the construction of the mesh as a series of steps, 
and then we will count the number of mesh elements created.

{\bf Step 1:}
   Using   Lemma \ref{simple poly lemma} we reduce to 
   the case when $\Gamma$ is already meshed by $O(n)$
   simple polygons and has $O(n)$ vertices and edges. 

{\bf Step 2:}
   For each vertex $v$ of $\Gamma$, we choose a disk $D$ centered at $v$
    with   diameter much smaller than   the distance to the nearest
    distinct vertex or edge of $\Gamma$. Inside  $D \cap \PH(\Gamma)$ 
  we place a $5^\circ$-conforming sink using Lemma \ref{sector conform sink}. 
     The interiors of these conforming sinks are 
    called the protected regions. 
    The only quadrilaterals in our final mesh that will  
     have angles less than $60^\circ$
    are those at the center  these  protecting sinks, and only 
    when there are edges of $\Gamma$ that touch $v$ making 
     angle $< 60^\circ$. These quadrilaterals  touching $v$ 
     will not be altered by any later step of the construction.

     Note that when we define the boundary of the 
     protecting sink for $D$, we can assume that 
     each connected component of $\partial D 
     \setminus \Gamma$ is divided into equal sized 
     intervals  by the boundary vertices of the sink
     (when a component has angle measure less than 
      $5^\circ$  it contains a single boundary edge of 
      the sink; otherwise the boundary of the sink 
      inscribed on this arc consists of  subarcs of
       equal  angle measure $\leq 5^\circ$).  
     This fact will be used later
     to show that protecting sinks do not touch 
     hyperbolic thin parts in the unprotected region.

{\bf Step 3:}
     Note that the unprotected region is now meshed by simple 
     polygons where all the angles are either close to $90^\circ$ or 
     close to $180^\circ$. The former occurs where boundaries 
     of the protecting sinks meet the edges of $\Gamma$, and 
     the latter occur at other boundary vertices of the 
     protecting sinks.
     More precisely, 
    the non-protected region of $\PH(\Gamma)$ is now meshed 
    by simple polygons $\{ P_j\}$  with all angles either in 
     $[90^\circ, 95^\circ ]$ or $[180^\circ , 185^\circ ]$. 
      By Theorem  \ref{simple quad mesh} each 
     such face in the unprotected region can
      be nicely meshed using  $O(k/\theta^2)$ 
      quadrilaterals having the properties listed in 
     Theorem \ref{simple quad mesh}, 
     where $k$ is the number of sides of that face.   Summing 
    over all the simple polygons   shows that 
     $O(n/\theta^2) $ quadrilaterals are used overall.

{\bf Step 4:} 
     By Theorem \ref{simple quad mesh},  the thick quadrilaterals 
     have uniformly bounded eccentricity independent of the polygon. 
     Each such quadrilateral can share at most one side with 
     $P$. We place a sink in each thick quadrilateral.

{\bf Step 5:} 
     Similarly, the parabolic quadrilaterals have bounded eccentricity 
     if the angles of the polygon are bounded away from zero, which 
     occurs in our case (as noted in Step 3, we only apply the 
      thick/thin decomposition  to a polygon where all the angles
      are  $\geq 90^\circ$).
      Thus all the parabolic quadrilaterals 
      have uniformly bounded eccentricities. 
     We place a sink in every    parabolic quadrilateral.

{\bf Step 6:}
     By Theorem \ref{simple quad mesh},  
     the hyperbolic quadrilaterals  are all  $\theta$-nice
     isosceles trapezoids and each hyperbolic thin part consists of a chain 
     of such trapezoids. The $P$-sides of these trapezoids
      agree with sides of thick 
     quadrilaterals and their $Q$-sides either agree with 
     the $Q$-sides of other  elements of the chain or they lie on $P$.
     Thus the union $W$ of all the closed hyperbolic quadrilaterals has 
     a dissection by  $\theta$-nice, isosceles trapezoids that 
     has $ N=O(n/\theta)$ elements but only $M=O(n)$ chains. 
     Theorem  \ref{quad mesh lemma} says that we can remove  
    $O(M)=O(n)$ nice quadrilaterals with uniformly bounded 
     eccentricity from $W$ to obtain a region $W'$, such 
     that  $W'$ can be $2\theta$-nicely meshed using $O(MN/\theta) =
     O(n^2/  \theta^{2})$ quadrilaterals.
      The mesh creates at most $O(n/\theta)$ new 
     vertices on the $Q$-sides  $\partial W'$ and at most $O(1)$ 
     vertices on each $P$-side of $W'$. 

{\bf Step 7:}
    Place sinks in each of the nice quadrilaterals  in $W \setminus W'$ 
    (i.e., the quadrilaterals removed from $W$). Any extra vertices
    coming from propagation through the mesh of  $W'$ always land 
    on a single pair of opposite sides of such a quadrilaterals
    (the $Q$-sides), 
    and so the removed quadrilaterals 
    can be remeshed using a number of elements that 
    is comparable to the number of extra vertices. 

{\bf Step 8:} 
    On the boundary of $W'$ there may be  points that correspond 
    to vertices of adjacent sinks, but are not vertices of the 
    mesh of $W'$. Propagate such points  through the mesh of $W'$
    until they connect with another boundary point of $W'$.
    There are at most $O(n/\theta)$ such boundary points 
    ($O(1)$ for each sink) and each propagation path 
    has length $O(N) = O(n/\theta)$ by 
     Theorem  \ref{quad mesh lemma}. Thus at most 
    $O(n^2/\theta^2)$ new mesh elements are created.

{\bf Step 9:}
    Next we use the re-meshing property of sinks. To make
    sure every sink has an even number of points on its 
    boundary, we first
    bisect the mesh on $W'$ (i.e., split every quadrilateral 
    into four by bisecting all four edges) and then 
    we split every boundary edge of 
    every sink into two equal sub-edges. This includes both the 
    sinks placed in thick quadrilaterals (Step 4) 
      and thin parabolic quadrilaterals (Step 5), 
  as well as the protecting 
    sinks placed around each original vertex (Step 2). 
    This doubles the number 
    of vertices on the boundary of each sink, and hence this number 
    is even for each sink.  Re-mesh the interiors 
    of all the sinks to conform with these vertices. 
    The meshes
    on the sinks and on $W'$ now match each other along all 
    common edges so we have  the desired conforming mesh of  $\Gamma$.

    This completes the construction of the  mesh.

    Finally, we have to count the total number of mesh elements 
    that we have created.
     We are interested both in the total 
    number of elements and in the number of elements that are 
    not $2\theta$-nice.
     From our remarks above, the number of mesh elements inside 
    $W'$ is $O(n^2/\theta^2)$ and all of these are $2\theta$-nice.
    (The original mesh of $W'$ is $2\theta$-nice by 
    Theorem \ref{quad mesh lemma}; any further quadrilaterals 
    in $W'$ are formed by standard propagation paths, and 
     these preserve niceness by Lemma \ref{split quad}.)

    Next we consider elements created when remeshing the sinks. 
    Each sink has some share of the $O(n/\theta)$ vertices that 
    are created on the boundary of $W'$. It also has $O(1)$ 
    other extra vertices that come from adjacent sinks. 
     There are four  types of sinks that have to be considered:
    
      {\bf Protecting sinks:} We claimed earlier that at most 
      $O(d)$ points would be added to boundary of a sink 
      protecting a vertex of degree $d$. To prove this claim 
      we will show that the boundary of  a protecting sink 
      can only touch thick or parabolic quadrilaterals, never 
      a hyperbolic quadrilateral, and hence the protecting 
      sinks never touch $W'$.  More precisely, 

      \begin{lemma}
     An edge $e$  of the boundary of a 
     sink protecting a vertex $v$ of $\Gamma$ 
      never touches a hyperbolic thin part in the 
     unprotected region. 
    \end{lemma}
     
     \begin{proof} 
     Suppose $e$ is a boundary edge of a protecting sink 
      that also lies on the boundary of a face  $P$  of the 
     unprotected region,  and suppose that  $e$ contains 
     the side of a hyperbolic thin part in $P$.
     Then $e$ would have 
     to have small extremal distance in $P$  to some other side $f$
     of $P$.
      By the definition  of hyperbolic thin 
      parts, $e$ and $f$ cannot be  
     adjacent on $P$, and by properties of extremal distance, 
      the path distance between $e$ and $f$ inside 
     $P$ would have to be much smaller than the diameter 
     of either $e$ or $f$. 
      The disk $D$ containing the sink was chosen to be
      small compared to the distance  of $v$ to other vertices and 
      edges of $\Gamma$, and hence  all edges of 
      $P$ that do not touch the boundary of the sink 
      have a distance from the sink that is much larger 
      than the diameter of the sink itself.
      Hence $f$ must lie on one of the edges of 
      $\Gamma$ that define the sector containing $e$, 
     or $f$ must be another boundary edge of the 
     same sink (and inside the same sector as $e$).
     However, both these alternatives  are impossible.
     By the way that the protecting sinks are constructed in 
     Step 2, each boundary edge of the sink in the sector 
     has the same length. If $f$ is 
     a subsegment of $\Gamma$, then $e$ is not adjacent to $f$
     and hence it must be  separated from $f$ by a distance comparable 
     to its own diameter. On the other hand, if $f$ is on the
      boundary of the sink, then 
      $e$ and $f$ have the same diameter and since they
      are not adjacent,  they must be  separated by a third 
     edge of the  same size.
      Hence there is no side $f$ of  $P$ with 
     small extremal distance to $e$, and the lemma is proven. 
      \end{proof}

      Thus the only 
      extra vertices that are added to the boundaries of protecting 
     sinks are due to sink vertices  corresponding to 
     adjacent thick or parabolic 
     quadrilaterals. The number of such points is $O(1)$ per 
     quadrilateral and there are at most $O(d)$ such 
     quadrilaterals,  so by Lemma \ref{sector conform sink}
     the total number of mesh elements 
         used in a protecting 
     sink is $O(d^2)$ where $d$ is the degree of the vertex 
     being protected. 
      Also by Lemma \ref{sector conform sink},
 at most $O( d/\theta)$  of the mesh elements are 
      not $\theta$-nice. Since summing $d$ over all protecting 
      sinks gives at most $O(n)$, 
      see that at most $O(n^2)$ mesh elements are used inside 
     protecting sinks, and that 
     at most $O( n/\theta)$ of these are not $2\theta$-nice.

     {\bf Thick quadrilateral sinks:} there are $O(n/\theta^2)$
       such quadrilaterals and hence $O(n/\theta^2)$ elements
       in the corresponding sinks (because of uniformly bounded
      eccentricity). Each thick quadrilateral  shares at most one 
      edge with $P$. Hence each thick quadrilateral can 
      share  at most one edge with a $Q$-side of $W'$ and may share
      a second side with $P$-side of $W'$ (this would be a 
      $P$-side of a hyperbolic quadrilateral in $P$).
       Thus there is at most one side of the thick quadrilateral 
      that gets more than $O(1)$ extra boundary vertices.
       Thus if $K$ extra 
      vertices are added from $W'$, the re-meshing can be 
      done with $O(K)$ extra elements.  Summing over all 
      thick quadrilaterals gives $O(n/\theta)$ extra mesh 
       elements, so the total is still $O(n/\theta^2)$  
      
        {\bf Thin quadrilateral sinks:}
   A thin quadrilateral  $Q$ can share either zero, one or 
   two sides with the simple polygon $P$ that contains 
   it. In the first two cases,  $Q$ can  share at most  one side  with 
       $W'$ and the other sides each have at most $O(1)$ 
       extra vertices, due to  sinks in adjacent quadrilaterals. 
       Moreover, there are $O(n/\theta^2)$ such quadrilaterals and 
       there are $O(n/\theta)$ vertices on the boundary of $W'$.
       Thus summing  the number of elements in   the remeshings of 
         all such quadrilaterals gives at 
       most $O(n/\theta^2)$  mesh elements in total. All could be 
       non-$\theta$-nice. 
 
   If $Q$ shares two sides with $P$, we claim
     that  at most one of these sides is  on $\partial W'$.
       To see this, note  if $Q$ has  two sides on $P$, then
      $Q$ contains a vertex $v$ of $P$.
     If the angle of $P$ at $v$  close to $180^\circ$,
     then the method  from \cite{Bishop-optimal}
     subdivides the angle at $v$ and  any quadrilaterals
     containing $v$  have at most  one edge with $P$.
     Thus the angle of $P$ at $v$ must be close to $90^\circ$.
     However,  this implies  $Q$  has one side on the boundary
     of a protecting sink, and hence  at most one side
      can be on $ \partial W'$.

        Thus the analysis in this case (2 sides on $P$) 
       is the same as in 
        the two previous cases (0 or 1 side on $P$). Actually, it is 
      a little better since there 
       are only $O(n)$ such parabolic quadrilaterals instead of 
       $O(n/\theta^2)$ as in the previous two cases. 

     {\bf Quadrilateral sinks removed from $W$:} these all have 
       uniformly bounded eccentricity and extra vertices coming 
       from propagation through $W'$ are only added to a single 
       pair of opposite sides (the $Q$-sides). Thus if $K$ vertices are added 
       to such a sink, the re-meshing can be accomplished with 
       $O(K)$ elements.  There are
        $O(n/\theta)$ propagation paths 
       that can add extra vertices, so summing over all such quadrilaterals, 
       gives $O(n/\theta)$ mesh elements inside the removed
       quadrilaterals.

      In conclusion, the total number of mesh elements is $O(n^2/\theta^2)$ 
      (the dominant terms  are $O(n/\theta^2)$ from the mesh of $W'$
      and $O(n^2)$ from the protecting sinks). 
      The number of mesh elements that are not 
        $2\theta$-nice is $O(n/\theta^2)$ (this includes  all the mesh 
      elements in the last four categories of sinks above,  plus
      the number of non-$\theta$-nice mesh elements that can occur in 
       protecting sinks).  This completes the proof of Theorem
        \ref{near 90}, and hence of Theorem \ref{Quad Mesh} as well. 
\end{proof}

Taking a fixed $\theta$, say $\theta = 15^\circ$, 
the mesh of $\Gamma$ has only $O(n)$ vertices on the 
boundary of the polynomial hull of $\Gamma$ and all but 
$O(n)$ of the interior vertices have all angles close 
to $90^\circ$, so 

\begin{cor}
The mesh  in Theorem \ref{Quad Mesh} can be taken so that 
all but $O(n)$  vertices have  degree four.  
\end{cor}

This implies that 
 the mesh can be partitioned into  $O(n)$
(combinatorial) rectangular grids  using motorcycle graphs as 
described by Eppstein, Goodrich, Kim and Tamstorf in 
\cite{Eppstein-MotorCycle}.
Hence the mesh is completely described by a graph of size $O(n)$ with 
integer labels on the edges;  the vertices of the graph 
are the corners of the rectangular pieces, edges of the graph correspond 
to a shared boundary arc of two  rectangular pieces, and the labels give the 
number of mesh elements on this shared boundary arc (e.g., see 
Figure 5 of \cite{Eppstein-MotorCycle}).
Thus, although the mesh  in Theorem \ref{Quad Mesh} may have 
$O(n^2)$ elements, it only has complexity 
$O(n)$ in a certain sense.  
Can this observation be exploited in numerical methods that make 
use of quadrilateral meshes?


\bibliography{nonobtuse}
\bibliographystyle{plain}
\end{document}